\documentclass[a4paper,12pt,table]{article}
\usepackage{amssymb,amsmath,latexsym,color,amsthm,authblk,mathpazo,palatino,bbding,setspace,datetime,breakcites,hyphenat,microtype,diagbox,tikz,xcolor,tkz-graph,subcaption,graphicx,titlesec,accents, svg}

\usepackage[round,authoryear]{natbib}
\usepackage{hyperref,theoremref}
\usepackage[titletoc,title]{appendix}
\usepackage[shortlabels]{enumitem}
\usepackage{float}
\usepackage{amssymb}
\usepackage{tikz}
\restylefloat{table}

\everymath{\displaystyle}

\definecolor{armygreen}{rgb}{0.29, 0.33, 0.13}
\definecolor{auburn}{rgb}{0.43, 0.21, 0.1}
\definecolor{burgundy}{rgb}{0.5, 0.0, 0.13}
\definecolor{medium red}{rgb}{.490,.298,.337}
\definecolor{dark red}{rgb}{.235,.141,.161}

\hypersetup{
	colorlinks = true,
	linkcolor = {burgundy},
	citecolor = {burgundy}, 		
	linkbordercolor = {white},
}

\let\OLDthebibliography\thebibliography
\renewcommand\thebibliography[1]{
	\OLDthebibliography{#1}
	\setlength{\parskip}{0pt}
	\setlength{\itemsep}{0pt plus 0.1ex}
}

\DeclareFontFamily{U}{mathx}{\hyphenchar\font45}
\DeclareFontShape{U}{mathx}{m}{n}{<-> mathx10}{}
\DeclareSymbolFont{mathx}{U}{mathx}{m}{n}
\DeclareMathAccent{\widebar}{0}{mathx}{"73}

\newcommand*{\TitleFont}{%
	\usefont{\encodingdefault}{\rmdefault}{b}{n}%
	\fontsize{16}{20}%
	\selectfont}

\titleformat{\section}[block]{\normalfont\scshape\normalsize\filcenter}{\thesection .}{1em}{}
\titleformat{\subsection}{\normalfont\scshape\normalsize}{\thesubsection}{1em}{}
\titleformat{\subsubsection}{\normalfont\scshape\normalsize}{\thesubsubsection}{1em}{}

\newtheorem{theorem}{Theorem}[section]
\newtheorem{proposition}{Proposition}[section]
\newtheorem{claim}{Claim}[section]
\newtheorem{lemma}{Lemma}[section]
\newtheorem{corollary}{Corollary}[section]

\theoremstyle{definition}
\newtheorem{definition}{Definition}[section]
\newtheorem{example}{Example}[section]
\theoremstyle{remark}
\newtheorem{remark}{\textsc{Remark}}[section]
\theoremstyle{observation}
\newtheorem{observation}{\textsc{Observation}}[section]

\newdateformat{monthyeardate}{\monthname[\THEMONTH], \THEYEAR}

\title{\TitleFont \textsc{Identifiable Preference Types and Optimal Menu-Choice Experiments}
	\thanks{The author would like to thank Florian Brandl, Albin Erlanson, Benny Moldovanu, Souvik Roy, Joel Watson and Sonal Yadav for their valuable discussions and suggestions. The author would also like to thank the participants of the trimester program ``Advances in Mechanism Design" held at the University of Bonn.}}

\author{Ujjwal Kumar\footnote{Contact: ujju489@gmail.com}}

\affil {Institute for Microeconomics, University of Bonn}

\date{\today}

\begin{document}
	
	\maketitle
	\vspace*{-10mm}		
	
	\begin{abstract}
		
A planner wants to elicit partial information about a subject's preference relation. Preferences are grouped into ``types" relevant to the planner's purpose. Can the planner identify the subject's type through menu-choice data or experiment, without distinguishing preferences within the same type? We provide a complete answer by characterizing identifiable type spaces. The characterization uses a novel geometric condition, called face consistency, which governs distinctions within and across square
and hexagonal faces of the permutohedron. For every identifiable type space,
we explicitly construct the unique irreducible identifying experiment
from its boundary. This experiment is optimal among all identifying experiments under any cost criterion for which adding menus cannot reduce cost. Using the equivalence between identifiability and elicitation with acts, as shown in \citet{appACH2021}, we study symmetric policies requesting unranked shortlists that contain only acceptable products, include every product preferred to a selected one, and are nonempty whenever an acceptable product exists. Only two such policies are elicitable with acts: requesting the favorite acceptable product or the entire acceptable set.

\end{abstract}
	
	

	\newpage
\section{Introduction} A planner wants to elicit partial information about a subject's true preference over a finite set of alternatives, but not
necessarily her complete preference ranking.
For example, a product survey may seek the subject's favorite
product or the set of products she prefers to an outside option, without
eliciting her complete preference ranking. The planner may wish to leave other preference information
undisclosed to protect the subject's privacy or to limit elicitation to the
survey's purpose.\footnote{This concern is reflected in the data-minimisation principle
	of the European Union's General Data Protection Regulation.
	Article 5(1)(c) requires personal data to be ``adequate, relevant
	and limited to what is necessary in relation to the purposes
	for which they are processed".} The planner's objective is represented by a partition of the space
of preferences. Each element of the partition is a possible
\textit{type} of the subject, and the partition is a \textit{type space}. Preferences belong to the same type when the planner does not seek to
distinguish between them. The planner seeks \emph{type-relevant elicitation}:
to identify the subject's type without eliciting information that distinguishes
preferences belonging to the same type.

Menu-choice experiments, or simply experiments, are a natural elicitation device. An experiment consists of a collection of menus
of alternatives. The subject chooses her most-preferred alternative,
according to her true preference, from each menu. Each menu-choice observation rules out some preferences and leaves others possible. Therefore, an experiment induces a type space in which two preferences belong to the same type \textit{if and only if} they generate the same choices from every menu in the experiment. A type space $T$ is called \textit{identifiable} if it is induced by some experiment $\mathcal M$. In such cases, we say that $\mathcal M$ identifies $T$.

We ask which type spaces are identifiable and,
whenever identification is possible, which menus are indispensable for identification. This paper provides a complete answer by characterizing all identifiable type
spaces and explicitly determining their unique irreducible identifying
experiments. An experiment is irreducible if it contains no redundant menu,
where a menu is redundant if removing it leaves the induced type space
unchanged.

This identification problem also has a direct incentive interpretation.
In the framework of \citet{appACH2021}, the subject reports her type
and receives an alternative as a reward. The alternative she receives
depends on her report and the realized state of a randomization device.
Each report is therefore assigned an \emph{act}, specifying the
alternative awarded to the subject in each state. A type space is
\emph{elicitable with acts} if these acts can be chosen so that, for
every preference, truthful reporting yields a weakly preferred outcome
in every state and a strictly preferred outcome in at least one state
relative to each false report.\footnote{See
	Definition~\ref{IC_acts} for a formal definition of elicitability with acts.}
\citet{appACH2021} (Proposition 1) show that a type space is elicitable
with acts if and only if it is identifiable. This equivalence gives
an incentive foundation for characterizing identifiable type spaces.

Our characterization uses the geometry of the permutohedron graph (or simply, the permutohedron), whose
vertices represent preferences and whose (undirected) edges connect preferences
that differ only by swapping two consecutively ranked alternatives, that is, one adjacent swap of alternatives.
An experiment distinguishes an edge of the permutohedron if its endpoint preferences generate different choices from at least one menu in the experiment. A set of edges of the permutohedron is called \emph{identifiable} if there exists an experiment which distinguishes precisely those
edges. Our main result (Theorem \ref{Characterization_of_Identifiability}) establishes that a set of edges of the permutohedron is identifiable
if and only if it satisfies \emph{face consistency}, a condition
governing the pattern of distinguished edges within and across the
square and hexagonal faces of the permutohedron. Its proof also
explicitly constructs the unique irreducible experiment identifying
any such edge set.

The connection to identifiable type spaces rests on the observation (Observation \ref{Obs_1}) that an
experiment's distinguished edges completely determine its induced
type space. Each observed menu-choice imposes comparisons between the
chosen alternative and the other alternatives in the menu.
Preferences consistent with the same observations can be connected
by a sequence of adjacent swaps that preserve these comparisons.
Consequently, deleting the distinguished edges from the permutohedron
leaves connected components that are precisely the types induced by the
experiment. Conversely, the induced type space determines the
distinguished edges: they are precisely the edges whose endpoints
belong to different types. This establishes a one-to-one correspondence
between identifiable type spaces and identifiable sets of edges of the permutohedron. For a type space $T$, let its
\emph{boundary} $\partial T$ consist of the edges of the permutohedron whose endpoints
belong to different types. We call $T$ \emph{connected} if any two
preferences in the same type can be connected by a sequence of adjacent swaps that remain
within that type. Our main result (Theorem \ref{Characterization_of_Identifiability}) then yields the following
answer to the original identification problem: a type space is
identifiable if and only if it is connected and its boundary satisfies
face consistency.

The information revealed by observing a choice from a menu cannot be chosen arbitrarily.
Choosing an alternative from a menu reveals that it is preferred to every other
alternative in that menu. Consequently, a menu that distinguishes
an edge of the permutohedron may also distinguish other edges whose
endpoints the planner wishes to leave indistinguishable.
Face consistency captures these restrictions on sets of edges
through square consistency and hexagon consistency. To
understand these requirements, fix a set of edges
$\widehat{\mathcal E}$ of the permutohedron. A square of the permutohedron is a cycle of
four edges obtained from a preference by performing two adjacent
swaps involving disjoint pairs of alternatives.
Opposite edges in a square swap the same pair of alternatives. Performing one
swap leaves unchanged whether a menu distinguishes the other.
Thus, a menu distinguishes one edge of a square if and only if it
distinguishes the opposite edge. Square consistency requires the
proposed edge set to respect this restriction: opposite edges of
every square must either both belong to $\widehat{\mathcal E}$
or both lie outside it.

A hexagon of the permutohedron is a cycle of six edges obtained by varying only the order
of three alternatives occupying consecutive positions in a preference,
while leaving all other alternatives fixed. These three alternatives
form its \emph{changing triple}. Its \emph{higher} edges swap the
alternatives in the first two of these positions, while its
\emph{lower} edges swap those in the last two of these positions. Each pair of opposite
edges in a hexagon consists of one lower edge and one higher edge, both swapping
the same pair of alternatives. For a lower edge $e$, denote its opposite higher edge by $e^h$. Any menu distinguishing a lower edge $e$ must contain the swapped
pair and exclude every alternative ranked above that pair in
both endpoint preferences, including the third alternative
in the changing triple.
The same menu therefore also distinguishes the opposite higher edge $e^h$.
The first requirement of hexagon consistency is consequently that,
in every hexagon, a lower edge $e$ belonging to $\widehat{\mathcal E}$
has its opposite higher edge $e^h$ in $\widehat{\mathcal E}$.
The reverse implication need not hold. We call a hexagon
$\widehat{\mathcal E}$-\emph{active} if some lower edge lies outside
$\widehat{\mathcal E}$ while its opposite higher edge belongs to
$\widehat{\mathcal E}$.

The second requirement of hexagon consistency links different
hexagons and reveals why particular menus are indispensable.
Associate with each hexagon $H$ the set $A(H)$, consisting of its
changing triple and all alternatives ranked below that triple at
its vertices. An active hexagon $H$ is \emph{primitive} if no active
hexagon with the same changing triple has an associated set
strictly contained in $A(H)$. Consider a primitive hexagon $H$ and
a lower edge $e$ of $H$ outside $\widehat{\mathcal E}$ whose opposite higher
edge $e^h$ belongs to $\widehat{\mathcal E}$. Any experiment
identifying $\widehat{\mathcal E}$ must contain a menu $M$ that
distinguishes $e^h$ but not $e$. Such a menu must contain the entire
changing triple and satisfy $M \subseteq A(H)$. Primitivity forces
$M=A(H)$: otherwise, there would be an active hexagon with the same
changing triple and associated set $M$, contradicting the
primitivity of $H$. Thus, $A(H)$ must appear as a menu in every
identifying experiment. Since the menu $A(H)$ distinguishes every higher
edge of every hexagon with that same associated set, the second
requirement of hexagon consistency requires all these edges to
belong to $\widehat{\mathcal E}$.

These restrictions specify an explicit candidate experiment for
identifying $\widehat{\mathcal E}$. Binary menus are forced whenever
an edge in $\widehat{\mathcal E}$ swaps the two bottom-ranked
alternatives: only the menu consisting of the swapped pair can distinguish such an edge, because every other alternative is ranked above both members of the pair.
Collect these binary menus together with the associated sets of all
primitive hexagons, and denote the resulting experiment by
$\mathcal M^*$. Our main result
(Theorem~\ref{Characterization_of_Identifiability}) shows that,
whenever $\widehat{\mathcal E}$ satisfies face consistency,
$\mathcal M^*$ distinguishes precisely the edges in
$\widehat{\mathcal E}$ and hence identifies it. Thus, $\mathcal M^*$ is the common indispensable core of all identifying experiments and also identifies $\widehat{\mathcal E}$. This establishes that
$\mathcal M^*$ is the unique irreducible identifying experiment
and that any additional menus in an identifying experiment are
redundant. Hence, $\mathcal M^*$ is optimal among all experiments
identifying $\widehat{\mathcal E}$ under any cost criterion for
which adding menus cannot reduce the cost of an experiment.
Costs may depend on the identities and sizes of menus or on
interactions among them; neither additivity nor symmetry across
menus is required. The same experiment $\mathcal M^*$ is optimal
under every cost criterion in this class.

We apply our characterization to the elicitation of unranked shortlists
of acceptable products. The set of alternatives consists of $m\geq 2$
products and an outside option. A product is acceptable
if the subject prefers it to the outside option. A shortlist policy
assigns a set of products to each preference. The planner specifies
this policy in advance and asks the subject to report the set of
products prescribed by her preference, without reporting the ordering
of its elements. We require the policy to select only acceptable
products, include every product preferred to a selected one, and
select at least one product whenever an acceptable product exists.
We also impose symmetry under relabeling products while keeping
the outside option fixed. These requirements allow the shortlist's
length to vary with the number of acceptable products. Each policy
induces a type space by grouping preferences that prescribe the
same shortlist. We ask which such type spaces are identifiable or, equivalently,
elicitable with acts, allowing the planner to learn exactly the
prescribed shortlist.

We show that a shortlist policy satisfying these requirements induces an
identifiable type space if and only if it is the
\emph{first-choice shortlist policy} or the
\emph{full-acceptability shortlist policy}. The former selects the
subject's favorite product whenever any product is acceptable;
the latter selects all acceptable products. Both select the empty
set when no product is acceptable. This characterization concerns
the entire policy: the same rule must apply at every preference.
In particular, consider a policy requesting the best $k$ acceptable
products, or all acceptable products if fewer than $k$ exist.
Its induced type space is not identifiable whenever $1<k<m$.
Equivalently, no incentive compatible mechanism with acts can
elicit exactly that shortlist.

The proof of this policy characterization uses both connectedness
and face consistency. Hexagon consistency rules out selecting more
than one but fewer than all acceptable products, while connectedness
and a further application of hexagon consistency rule out switching between the two
policies as the number of acceptable products varies. The
characterization also determines the unique irreducible identifying
experiment for each feasible policy. For first-choice shortlist policy,
it consists of the grand menu containing all products and the
outside option. For full-acceptability shortlist policy, it consists of
the binary menus pairing each product with the outside option.
The latter experiment determines the acceptable set while leaving the ordering
within both the acceptable and unacceptable sets undisclosed.
We also provide incentive compatible direct mechanisms with acts
in which the subject reports only the prescribed shortlist.
For each of the two characterized policies, the identifying experiment and the direct mechanism
reveal the prescribed shortlist without distinguishing preferences
that prescribe the same shortlist.

\subsection{Related literature}

The papers closest to ours are \citet{appACH2021} and
\citet{healy2026experiments}. \citet{appACH2021} study which
type spaces admit incentive compatible elicitation using acts
or objective lotteries. Their Proposition~1 establishes that
a type space is elicitable with acts if and only if it is
generated by top elements of menus, or identifiable in our
terminology. Their Proposition~2 also establishes that every
identifiable type space admits a unique irreducible identifying
experiment, which they call a ``unique independent list of
menus''. They prove the existence of a unique independent list of
menus using an anti-exchange closure operator.
We provide a geometric characterization of identifiability (equivalently, elicitability with acts) through connectedness and face
consistency and explicitly construct the unique irreducible identifying experiment from the
boundary of the type space, specifying the menus forced by
the prescribed preference distinctions.

\citet{healy2026experiments} take a type space as given and characterize
the experiments that reveal enough information to determine
the subject's type, allowing the experiment to distinguish preferences within the
same type. Their main result, Theorem~1, expressed in our terminology,
shows that an experiment determines the subject's type if
and only if it distinguishes every edge in the boundary of
the given type space. The experiment may also distinguish
edges outside this boundary, thereby distinguishing
preferences belonging to the same type. Their
labeled permutohedron assigns to each edge all menus that
distinguish its endpoint preferences, yielding a set-cover formulation
for finding experiments of minimum size or cost. Our
identification requirement rules out these additional
distinctions. We use the relationships among edges within
and across square and hexagonal faces of the permutohedron
to characterize all type spaces that admit an experiment which determines the subject's type without distinguishing preferences within the same type. Moreover, for every such type space, we explicitly construct the unique irreducible identifying experiment, which uniquely minimizes the number of menus among all experiments satisfying our identification requirement.

The incentive interpretation connects our paper to the broader
literature on incentives in experiments and mechanisms with
acts. \citet{azrieli2018incentives} study incentive compatible
payment rules for a given collection of menus and justify
rewarding the subject with her choice from a randomly selected
menu under statewise monotonicity.
\citet{azrieli2020incentives} show how objective probabilities can enlarge the class of incentive compatible payment rules
relative to payments with acts. \citet{bahel2020strategyproof}
characterize strategyproof and unanimous rules for choosing
acts when agents report their (subjective expected utility) preferences.

Related work studies the identification and elicitation of
information about probability distributions.
\citet{caradonna2026identification} study identification in
stochastic choice, characterizing distributions over preferences
that generate identical menu-choice probabilities.
In the elicitation literature, \citet{lambert2008eliciting}
study which properties of a distribution can be truthfully
elicited without requesting the entire distribution.
\citet{lambert2011elicitation} provides a geometric characterization of elicitable properties taking finitely many values, showing that such a property is elicitable if and only if it partitions the simplex of distributions into a power diagram.

Restrictions on the information communicated also arise in
mechanism design. \citet{nisan2006communication} establish
communication requirements for efficient allocation, showing
that communication sufficient to determine an efficient
allocation must also reveal supporting prices.
\citet{haeringer2009constrained} study the strategic
consequences of restricting the number of acceptable schools
that students may list in school-choice mechanisms, a
restriction related to the reporting constraints in our
shortlist application.

The requirement of type-relevant elicitation also connects
our paper to privacy and information disclosure.
\citet{strack2024privacy} characterize privacy-preserving signals, which leave the posterior probabilities of specified events unchanged after every signal realization. \citet{haupt2026contextually} study dynamic elicitation protocols that limit information disclosed
beyond what is needed to implement a choice rule.

\citet{morton2009convex} also use the geometry of the permutohedron to study partitions of rankings underlying statistical rank tests. In \citet{kumarelicitation}, type spaces elicitable with objective lotteries are characterized
in terms of the existence of positive edge weights satisfying balance conditions on basic cycles of a type graph derived from the permutohedron. The adjacency structure of the permutohedron also appears in social choice theory: \citet{carroll2012local}, \citet{sato2013sufficient} and \citet{kumar2021local} establish conditions under which local incentive compatibility,
which rules out profitable misreports to adjacent preferences,
is equivalent to full incentive compatibility.

\subsection{Organization of the paper}

Section \ref{Model} presents the model and the one-to-one
correspondence between identifiable type spaces and identifiable
sets of edges of the permutohedron.
Section~\ref{Ex_illustrate_results} uses examples with three and
four alternatives to motivate the geometric conditions and
illustrate the characterization of identifiable sets of edges of the permutohedron and the explicit construction
of the unique irreducible identifying experiment. Section~\ref{Main} states the main result that characterizes identifiable sets of edges of the permutohedron and gives the explicit
construction. Section \ref{Char_Id_Typespace} derives the
characterization of identifiable type spaces as a corollary
of our main result. Section \ref{sec:applications-acts-shortlists}
applies this corollary to elicitation with acts and unranked
shortlists. 
 %Section~7 concludes. 

\section{The Model}\label{Model}
Let $X$ be a finite set of alternatives with $|X|=n\geq 3$. A \textit{preference} $P$ is an antisymmetric, complete and transitive binary relation over $X$, i.e. a {\it linear order}. For any $x,y \in X$, by $xPy$ we denote that $x$ is strictly preferred to $y$ according to the preference $P$. Let $\mathcal{P}$ be the set of all preferences.

A \textit{type} is a non-empty set of preferences. A \textit{type space} $T=\{t_1,\ldots,t_k\}$ is a collection of types that forms a partition of $\mathcal{P}$, that is, $t_i \cap t_j =\emptyset$ for all $i \neq j$ and $\mathcal{P} = \cup_{i=1}^k t_i$. Given a type space, for each $P\in\mathcal P$, let $t(P)$ denote the unique type containing $P$.

The planner's objective is \textit{type-relevant elicitation}, that is, to identify $t(P)$ without eliciting information that distinguishes $P$ from other preferences in $t(P)$. Thus, preferences belonging to different types should be observationally distinguishable, whereas preferences belonging to the same type should remain observationally indistinguishable. 

Menus are the planner's elicitation device. A \textit{menu} $M$ is a subset of $X$ such that $|M| \geq 2$. A \textit{menu-choice experiment} $\mathcal{M}$ is a (possibly empty) set of menus. To simplify terminology, we refer to menu-choice experiments simply as \textit{experiments} throughout the paper. Formally, an experiment $\mathcal{M} \subseteq \{M \subseteq X : |M|\geq 2\}$.

When the planner offers a menu $M$, the subject chooses her most-preferred (according to her true preference)
alternative in $M$.  For any preference $P$ and menu $M$, let $\mbox{\textit{dom}}_P(M)$ denote the unique most-preferred alternative in $M$ according to $P$. Formally, $\mbox{\textit{dom}}_P(M)= x$ if $x \in M$ and $xPy$ for all $y \in M \setminus \{x\}$. Given any experiment $\mathcal{M}$, let $\mathcal R_{\mathcal{M}}$ denote the type space induced by $\mathcal{M}$: two preferences $P$ and $P'$ belong to the same type \textit{if and only if} $\mbox{\textit{dom}}_P(M)=\mbox{\textit{dom}}_{P'}(M)$ for all $M \in \mathcal{M}$. A menu $M \in \mathcal{M}$ is \textit{redundant} if $\mathcal R_{\mathcal{M}}=\mathcal R_{\mathcal{M}\setminus \{M\}}$. An experiment $\mathcal{M}$ is called \textit{irreducible} if it contains no redundant menu.

\begin{definition}
	A type space $T$ is \textit{identifiable} if there exists an experiment $\mathcal{M}$ such that $T=\mathcal R_{\mathcal{M}}$. In this case, we say that $\mathcal{M}$ identifies $T$.
\end{definition}

This paper characterizes identifiable type spaces and, for every identifiable type space, determines explicitly its \textit{unique} irreducible identifying experiment. The next subsection presents some important observations about identifiable type spaces that facilitate our results.

\subsection{Preliminary observations} The goal of this subsection is to highlight the key equivalence used throughout the paper: characterizing identifiable type spaces is \textit{equivalent} to characterizing identifiable sets of edges of the associated permutohedron (described below).

We say that a pair of alternatives $\{a,b\}\subseteq X$ is
\textit{reversed} between preferences $P$ and $P'$ if either $aPb$ and
$bP'a$, or $bPa$ and $aP'b$. Let
$P\vartriangle P'=\{\{a,b\}\subseteq X: a \mbox{ and } b \text{ are reversed between }
P\text{ and }P'\}$ be the set of all reversed pairs of alternatives between $P$ and $P'$.
Two preferences $P$ and $P'$ are called \textit{adjacent} if
$|P\vartriangle P'|=1$.\footnote{Equivalently, the Kemeny distance between $P$ and $P'$ is one.} Thus, adjacent preferences are obtained from one another by a single adjacent swap and therefore rank exactly one pair of alternatives differently. Consider the (undirected) graph $G=\langle \mathcal{P}, \mathcal{E} \rangle$ where $\{P,P'\} \in \mathcal{E}$ \textit{if and only if} $P$ and $P'$ are adjacent preferences. We call $G$ the permutohedron graph, or simply the permutohedron. We label each edge of the permutohedron with the pair of alternatives reversed between the two adjacent preferences. For example, see Figure \ref{fig1:permutohedron-3}, where $X=\{x,y,z\}$.\footnote{By a preference $P=xyz$, we mean $xPyPz$.} 

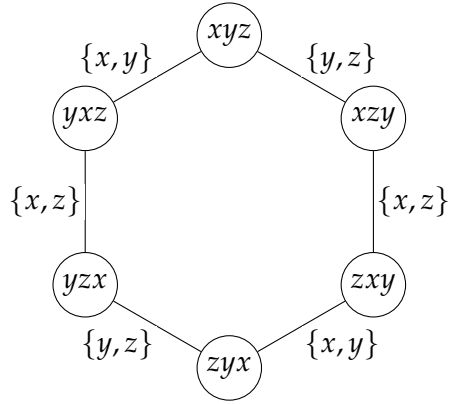
\begin{figure}[h]
	\centering
	\begin{tikzpicture}[
		scale=2.2,
		every node/.style={font=\small},
		pref/.style={circle, draw, inner sep=2pt}
		]
		
		% vertices
		\node[pref] (xyz) at (90:1)   {$xyz$};
		\node[pref] (xzy) at (30:1)   {$xzy$};
		\node[pref] (zxy) at (-30:1)  {$zxy$};
		\node[pref] (zyx) at (-90:1)  {$zyx$};
		\node[pref] (yzx) at (-150:1) {$yzx$};
		\node[pref] (yxz) at (150:1)  {$yxz$};
		
		% edges with labels
		\draw (xyz) -- node[midway, above right, fill=white, inner sep=1pt] {$\{y,z\}$} (xzy);
		\draw (xzy) -- node[midway, right,       fill=white, inner sep=1pt] {$\{x,z\}$} (zxy);
		\draw (zxy) -- node[midway, below right, fill=white, inner sep=1pt] {$\{x,y\}$} (zyx);
		\draw (zyx) -- node[midway, below left,  fill=white, inner sep=1pt] {$\{y,z\}$} (yzx);
		\draw (yzx) -- node[midway, left,        fill=white, inner sep=1pt] {$\{x,z\}$} (yxz);
		\draw (yxz) -- node[midway, above left,  fill=white, inner sep=1pt] {$\{x,y\}$} (xyz);
		
	\end{tikzpicture}
	\caption{The permutohedron for three alternatives.}
	\label{fig1:permutohedron-3}
\end{figure} 

A menu $M$ \textit{distinguishes} an edge of the permutohedron $\{P,P'\} \in \mathcal{E}$ if $P$ and $P'$ choose different alternatives from $M$, that is, if
$\mbox{\textit{dom}}_P(M)\neq \mbox{\textit{dom}}_{P'}(M)$. An experiment
$\mathcal{M}$ distinguishes an edge if at least one menu in $\mathcal{M}$
distinguishes it. Given an experiment $\mathcal{M}$, let $\widehat{\mathcal{E}}_{\mathcal{M}} \subseteq \mathcal{E}$ be the set of \textit{all} edges that the experiment $\mathcal{M}$ distinguishes.

For any $\widehat{\mathcal{E}} \subseteq \mathcal{E}$, let $G_{\widehat{\mathcal{E}}}$ be
the graph obtained from the permutohedron by deleting all edges in $\widehat{\mathcal{E}}$, formally, $G_{\widehat{\mathcal{E}}}=\langle \mathcal{P}, \mathcal{E} \setminus \widehat{\mathcal{E}} \rangle$. A \textit{path} $\pi=(P^1,\ldots,P^s)$ in $G_{\widehat{\mathcal{E}}}$ is a sequence of distinct preferences in $\mathcal P$ satisfying the property that $\{P^i,P^{i+1}\} \in \mathcal{E} \setminus \widehat{\mathcal{E}}$ for all $i=1,\ldots,s-1$. In such scenarios, we say that $\pi$ is a path in $G_{\widehat{\mathcal{E}}}$ between $P^1$ and $P^s$. Two preferences $P,P' \in \mathcal P$ are \textit{connected} in $G_{\widehat{\mathcal{E}}}$ if there exists a path in $G_{\widehat{\mathcal{E}}}$ between $P$ and $P'$. A \textit{connected component} of $G_{\widehat{\mathcal{E}}}$ is a maximal subset of $\mathcal P$ in which every pair of preferences is connected in $G_{\widehat{\mathcal{E}}}$.

 Let $\mathcal R_{\widehat{\mathcal{E}}}$ be the type space whose types are the connected components of $G_{\widehat{\mathcal{E}}}$. We often refer to $\mathcal R_{\widehat{\mathcal{E}}}$ as the type space obtained from the permutohedron by deleting all edges in $\widehat{\mathcal{E}}$. Fix an experiment $\mathcal M$. Recall that $\widehat{\mathcal{E}}_{\mathcal{M}}$ is the set of all edges that the experiment $\mathcal{M}$ distinguishes. Consider the type space $\mathcal R_{\widehat{\mathcal{E}}_{\mathcal{M}}}$. Deleting all edges of $\widehat{\mathcal{E}}_{\mathcal{M}}$ ensures that every preference belonging to the same type (connected component) make same choices over all menus of $\mathcal M$. 
 
 Next, we observe that given any experiment $\mathcal M$, the type space obtained from the permutohedron by deleting all edges in $\widehat{\mathcal{E}}_{\mathcal{M}}$ is the same as the type space induced by $\mathcal M$.

\begin{observation} \label{Obs_1}
	For every experiment $\mathcal{M}$, $\mathcal R_{\mathcal{M}}=\mathcal R_{\widehat{\mathcal{E}}_{\mathcal{M}}}$.
\end{observation}

 The proof of this observation is relegated to Appendix \ref{B}. %\uk{maybe give a direct proof. $\implies$ direction: if two preference $P$ and $P'$ belong to the same type in $\mathcal R_{\mathcal{M}}$, then for every reversed pair $\{a,b\}$ in $P \vartriangle P'$ contained in any menu $M$, $M$ must contain some alternative $x$ such that $x \in X \setminus P \vartriangle P'$. This implies that for any $\tilde P$ in between $P$ and $P'$, $\tilde P$ must also belong to the same type. Hence, there is a shortest path from $P$ to $P'$ fully contained in the same type, hence connected.} 

\begin{figure}[h]
	\centering
	\begin{tikzpicture}[
		scale=2.2,
		every node/.style={font=\small},
		pref/.style={circle, draw, inner sep=2pt}
		]
		
		% vertices
		\node[pref] (xyz) at (90:1)   {$xyz$};
		\node[pref] (xzy) at (30:1)   {$xzy$};
		\node[pref] (zxy) at (-30:1)  {$zxy$};
		\node[pref] (zyx) at (-90:1)  {$zyx$};
		\node[pref] (yzx) at (-150:1) {$yzx$};
		\node[pref] (yxz) at (150:1)  {$yxz$};
		
		% edges with labels
		\draw (xyz) -- node[midway, above right, fill=white, inner sep=1pt] {$\{y,z\}$} (xzy);
		%\draw (xzy) -- node[midway, right,       fill=white, inner sep=1pt] {$\{x,z\}$} (zxy);
		\draw (zxy) -- node[midway, below right, fill=white, inner sep=1pt] {$\{x,y\}$} (zyx);
		%\draw (zyx) -- node[midway, below left,  fill=white, inner sep=1pt] {$\{y,z\}$} (yzx);
		\draw (yzx) -- node[midway, left,        fill=white, inner sep=1pt] {$\{x,z\}$} (yxz);
		%\draw (yxz) -- node[midway, above left,  fill=white, inner sep=1pt] {$\{x,y\}$} (xyz);
		
	\end{tikzpicture}
	\caption{The graph $G_{\widehat{\mathcal{E}}_{\mathcal{M}}}$, where $\mathcal{M}=\{\{x,y,z\}\}$.}
	\label{fig2:permutohedron-3}
\end{figure}
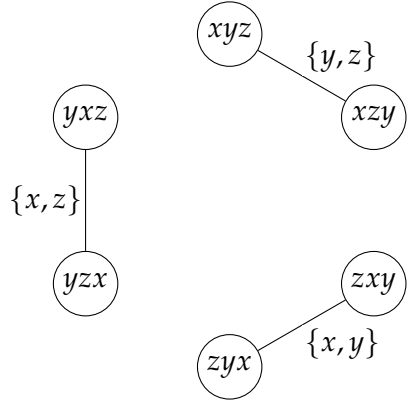

We illustrate Observation \ref{Obs_1} in the case of three alternatives. Let
$X=\{x,y,z\}$ and consider the experiment $\mathcal{M}=\{\{x,y,z\}\}$. Then $\mathcal R_{\mathcal{M}}=\{\{xyz,xzy\},\{yxz,yzx\},\{zxy,zyx\}\}$. The graph 
$G_{\widehat{\mathcal{E}}_{\mathcal{M}}}$ is shown in Figure \ref{fig2:permutohedron-3}, where $\widehat{\mathcal{E}}_{\mathcal{M}}=\{\{xyz,yxz\}, \{xzy,zxy\}, \{zyx,yzx\}\}$.
Hence, $$\mathcal R_{\widehat{\mathcal{E}}_{\mathcal{M}}}
=
\{\{xyz,xzy\},\{yxz,yzx\},\{zxy,zyx\}\}
=
\mathcal R_{\mathcal{M}}.
$$

\begin{definition}
	
	A set of edges $\widehat{\mathcal{E}} \subseteq \mathcal{E}$ of the permutohedron is called \textit{identifiable} if there exists an experiment $\mathcal{M}$ such that $\widehat{\mathcal{E}}=\widehat{\mathcal{E}}_{\mathcal{M}}$.\footnote{We purposefully use the term ``identifiable" here in the context of set of edges to emphasize that there is a one to one correspondence between identifiable type spaces and identifiable set of edges, as discussed below in Remark \ref{equiv}.} In this case, we say that $\mathcal{M}$ identifies $\widehat{\mathcal{E}}$.
	
\end{definition}

%It is worth noting that if $\widehat{\mathcal{E}}$ is identifiable by an experiment $\mathcal{M}$ then, by Observation \ref{Obs_1}, the type space $\mathcal R_{\widehat{\mathcal{E}}}$ is induced by the experiment $\mathcal{M}$. 

\begin{remark}\label{equiv}
Observation~\ref{Obs_1} yields a one-to-one correspondence between identifiable
type spaces and identifiable sets of edges of the permutohedron. Indeed, if an experiment
$\mathcal M$ identifies a type space $T$, then
$\widehat{\mathcal E}_{\mathcal M}$ consists precisely of those edges of the
permutohedron whose endpoints belong to different types in $T$. Hence
$\widehat{\mathcal E}_{\mathcal M}$ is uniquely determined by $T$. Conversely,
if $\widehat{\mathcal E}$ is identifiable, then any experiment identifying
$\widehat{\mathcal E}$ induces the type space
$\mathcal R_{\widehat{\mathcal E}}$. Therefore, characterizing identifiable
type spaces is equivalent to characterizing identifiable sets of edges of the
permutohedron. 
\end{remark}

%a type space $T$ is identifiable \textit{if only if} there exists an identifiable set of edges $\widehat{\mathcal{E}} \subseteq \mathcal{E}$ such that $T=\mathcal R_{\widehat{\mathcal{E}}}$. So, the problem of characterizing identifiable type spaces is equivalent to the problem of characterizing identifiable set of edges of the permutohedron.

It is worth noting that not every set of edges of the permutohedron is identifiable. We illustrate this via an example (Example \ref{Ex_1}) below. 

\begin{example}\label{Ex_1}
Consider the permutohedron with $X=\{x,y,z\}$, as shown in Figure \ref{fig1:permutohedron-3}. Consider the set of edges $\widehat{\mathcal{E}}=\{\{xyz,xzy\}, \{zxy,zyx\}\}$. We show that $\widehat{\mathcal{E}}$ is not identifiable. Suppose, to the contrary, that some experiment $\mathcal M$ identifies $\widehat{\mathcal{E}}$. This implies that there exists a menu $M \in \mathcal M$ which distinguishes the edge $\{xyz,xzy\} \in \widehat{\mathcal{E}}$. Then, it must be the case that $M=\{y,z\}$ because $\{y,z\}$ is the only menu that distinguishes the edge $\{xyz,xzy\}$. Therefore, the menu $M$ must also distinguish the edge $\{yzx,zyx\} \notin \widehat{\mathcal{E}}$, a contradiction to our assumption that $\mathcal M$ identifies $\widehat{\mathcal{E}}$. Hence
	$\widehat{\mathcal E}$ is not identifiable. Here, the edges in $\widehat{\mathcal E}$ are precisely those whose endpoints belong to different types of the type space
	$$
	\mathcal R_{\widehat{\mathcal E}}
	=\{\{xzy,zxy\},\{xyz,yxz,yzx,zyx\}\}.
	$$ By Remark \ref{equiv}, any experiment identifying this type space would also identify $\widehat{\mathcal E}$. Since $\widehat{\mathcal E}$ is not identifiable, neither is $\mathcal R_{\widehat{\mathcal E}}$. \hfill $\square$
\end{example}

Moreover, Observation \ref{Obs_1} immediately rules out some type spaces from being identifiable. For example, let $X=\{x,y,z\}$ and
consider
$$
T=\{\{xyz,zyx\},\{yxz,xzy,yzx,zxy\}\}.
$$
To show that $T$ is identifiable, by Observation \ref{Obs_1}, it is equivalent to showing that there exists an identifiable set of edges $\widehat{\mathcal{E}} \subseteq \mathcal{E}$ such that $T=\mathcal R_{\widehat{\mathcal{E}}}$. Consider the permutohedron as shown in Figure \ref{fig1:permutohedron-3}. We claim that there does not exist any set of edges $\widehat{\mathcal{E}} \subseteq \mathcal{E}$ such that $T=\mathcal R_{\widehat{\mathcal{E}}}$.\footnote{Notice that we show a stronger statement here, not only there is no identifiable set of edges $\widehat{\mathcal{E}}$ with $T=\mathcal R_{\widehat{\mathcal{E}}}$, we show that there is no set of edges $\widehat{\mathcal{E}}$ satisfying $T=\mathcal R_{\widehat{\mathcal{E}}}$.} To see this, suppose, to the contrary, that $T=\mathcal R_{\widehat{\mathcal{E}}}$ for some $\widehat{\mathcal{E}} \subseteq \mathcal{E}$. Since the preferences $xyz$ and $zyx$ belong to the same type (connected component), at least one of the two
paths (clockwise or anti-clockwise) between $xyz$ and $zyx$ must remain in $t(xyz)$. But either path contains
intermediate vertices different from $xyz$ and $zyx$, which would then also belong to $t(xyz)$, a contradiction to $T=\mathcal R_{\widehat{\mathcal{E}}}$. Hence, there does not exist any set of edges $\widehat{\mathcal{E}} \subseteq \mathcal{E}$ such that $T=\mathcal R_{\widehat{\mathcal{E}}}$. Therefore, by Observation \ref{Obs_1}, $T$ is not identifiable.

\section{Examples Illustrating the Main Result}\label{Ex_illustrate_results}
This section illustrates the characterization of identifiable sets of edges of the permutohedron via examples. First, we present the characterization result for $n=3$. Then, we present it for $n=4$, which also continues to hold for $n>4$.

\subsection{Three alternatives} 
In this subsection, we assume $X=\{x,y,z\}$. The permutohedron is depicted in Figure \ref{fig4:permutohedron-3}. Since there are only three alternatives, the permutohedron is a hexagon. In Figure \ref{fig4:permutohedron-3}, we additionally label each edge in the hexagon in the following manner:
\begin{enumerate}[(i)]
	\item An edge is labelled $\boldsymbol{\ell}$ if the pair of alternatives reversed between the two vertices (adjacent preferences) occupy the bottom two positions at both of these adjacent preferences.
	\item An edge is labelled $\boldsymbol{h}$ if the pair of alternatives reversed between the two vertices (adjacent preferences) occupy the top two positions at both of these adjacent preferences.
\end{enumerate}

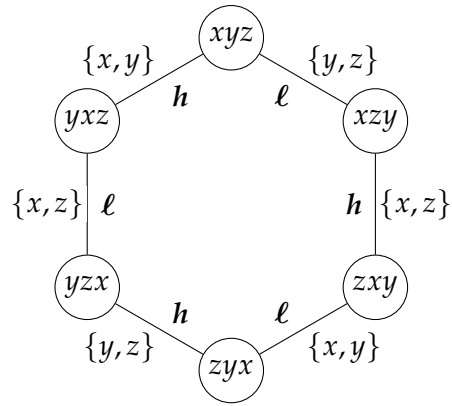
\begin{figure}[h]
	\centering
	\begin{tikzpicture}[
		scale=2.2,
		every node/.style={font=\small},
		pref/.style={circle, draw, inner sep=2pt}
		]
		
		% vertices
		\node[pref] (xyz) at (90:1)   {$xyz$};
		\node[pref] (xzy) at (30:1)   {$xzy$};
		\node[pref] (zxy) at (-30:1)  {$zxy$};
		\node[pref] (zyx) at (-90:1)  {$zyx$};
		\node[pref] (yzx) at (-150:1) {$yzx$};
		\node[pref] (yxz) at (150:1)  {$yxz$};
		
		% edges with labels
		\draw (xyz) -- 
		node[midway, above right, fill=white, inner sep=1pt] {$\{y,z\}$}
		node[midway, fill=white, inner sep=0.8pt, xshift=-8pt, yshift=-6pt] {$\boldsymbol{\ell}$}
		(xzy);
		
		\draw (xzy) -- 
		node[midway, right, fill=white, inner sep=1pt] {$\{x,z\}$}
		node[midway, fill=white, inner sep=0.8pt, xshift=-8pt] {$\boldsymbol{h}$}
		(zxy);
		
		\draw (zxy) -- 
		node[midway, below right, fill=white, inner sep=1pt] {$\{x,y\}$}
		node[midway, fill=white, inner sep=0.8pt, xshift=-8pt, yshift=6pt] {$\boldsymbol{\ell}$}
		(zyx);
		
		\draw (zyx) -- 
		node[midway, below left, fill=white, inner sep=1pt] {$\{y,z\}$}
		node[midway, fill=white, inner sep=0.8pt, xshift=8pt, yshift=6pt] {$\boldsymbol{h}$}
		(yzx);
		
		\draw (yzx) -- 
		node[midway, left, fill=white, inner sep=1pt] {$\{x,z\}$}
		node[midway, fill=white, inner sep=0.8pt, xshift=8pt] {$\boldsymbol{\ell}$}
		(yxz);
		
		\draw (yxz) -- 
		node[midway, above left, fill=white, inner sep=1pt] {$\{x,y\}$}
		node[midway, fill=white, inner sep=0.8pt, xshift=8pt, yshift=-6pt] {$\boldsymbol{h}$}
		(xyz);
		
	\end{tikzpicture}
	\caption{The permutohedron for three alternatives, with lower and higher edges labelled.}
	\label{fig4:permutohedron-3}
\end{figure}

The labellings $\boldsymbol{\ell}$ and $\boldsymbol{h}$ reflect the position at which the adjacent swap occurs. Along a lower edge, the swap affects only the two alternatives that occupy bottom two positions. Along a higher edge, the swap affects the two alternatives that occupy top two positions.

For a lower edge $e$ of the hexagon labelled by the pair $\{u,v\}$, where $u,v \in X$, by $e^h$ we denote the unique opposite higher edge of the hexagon that is also labelled by $\{u,v\}$. For example, $\{xyz,xzy\}$ is a lower edge labelled with $\{y,z\}$, and $\{yzx,zyx\}$ is the unique opposite higher edge that is also labelled $\{y,z\}$. 

Consider any lower ($\boldsymbol{\ell}$)-edge, say $\{xyz,xzy\}$. Since the preferences $xyz$ and $xzy$ differ only in the ranking of alternatives $y$ and $z$ at bottom two positions, for any menu $M$ that distinguishes this edge, it must be the case that $M=\{y,z\}$. Therefore, lower edges force binary menus, in this case the edge $\{xyz,xzy\}$ forces the menu $M=\{y,z\}$. Further, notice that the menu $M=\{y,z\}$ must also distinguish the opposite higher ($\boldsymbol{h}$)-edge $\{yzx,zyx\}$. 

Let $\widehat{\mathcal E}$ be an identifiable set of edges. Therefore, from the discussion in the previous paragraph it follows that if a lower edge $e \in \widehat{\mathcal E}$, then $e^h \in \widehat{\mathcal E}$. Next, we analyze what happens if for some lower edge $e \notin \widehat{\mathcal E}$ but $e^h \in \widehat{\mathcal E}$.

Consider any lower edge, say, $e=\{xyz,xzy\}$. Therefore, $e^h=\{yzx,zyx\}$. Suppose $e \notin \widehat{\mathcal E}$ but $e^h \in \widehat{\mathcal E}$. Since $\widehat{\mathcal E}$ is identifiable, there exists an experiment $\mathcal M$ that identifies $\widehat{\mathcal E}$. Therefore, there exists a menu $M \in \mathcal M$ such that $M$ distinguishes $e^h$ but does not distinguish $e$. Since the preferences $yzx$ and $zyx$ differ only in the ranking of alternatives $y$ and $z$, the fact that menu $M$ distinguishes $e^h=\{yzx,zyx\}$ implies $\{y,z\} \subseteq M$. If $M=\{y,z\}$, then $M$ also distinguishes the edge $e$, which will be a contradiction to the fact that $e \notin \widehat{\mathcal E}$. Therefore, it must be the case that $M=\{x,y,z\}$. It is worth noting here that the exclusion of a lower edge but inclusion of the opposite higher edge forces a ternary menu, in this case $e=\{xyz,xzy\} \notin \widehat{\mathcal E}$ but $e^h=\{yzx,zyx\} \in \widehat{\mathcal E}$, hence it forces the menu $M=\{x,y,z\}$. Thus, $M$ distinguishes \textit{all} higher edges of the hexagon. Hence, all higher edges of the hexagon must belong to $\widehat{\mathcal E}$. 

We call a hexagon $\widehat{\mathcal{E}}$-active if for some lower edge $e$ of the hexagon, $e \notin \widehat{\mathcal{E}}$ but $e^h \in \widehat{\mathcal{E}}$. Therefore, by the discussion in the previous paragraph, it follows that for identifiability, all higher edges of an $\widehat{\mathcal{E}}$-active hexagon must be included in $\widehat{\mathcal{E}}$. Moreover, if the hexagon is $\widehat{\mathcal{E}}$-active, then it must force the menu $M=\{x,y,z\}$. Next, we define a consistency condition which completely characterizes identifiable set of edges for the case of three alternatives.

\begin{definition}\label{Hex_consistency_3}
A set of edges $\widehat{\mathcal{E}}$ satisfies
\textit{hexagon consistency} if the following conditions hold:

\begin{enumerate}[(i)]
	\item For every lower edge $e$ of the hexagon, $e\in \widehat{\mathcal{E}} \implies e^h\in \widehat{\mathcal{E}}$.
	
	\item If the hexagon is $\widehat{\mathcal{E}}$-active, then every higher edge of the hexagon belongs to $\widehat{\mathcal{E}}$.
\end{enumerate}	
\end{definition}

Before presenting our main result for the case of three alternatives, let us first illustrate the notion of \textit{irreducible} experiments. Consider two experiments $\mathcal M=\{\{x,y\},\{y,z\},\{x,z\}\}$ and $\mathcal M'=\{\{x,y\},\{y,z\},\{x,z\},\{x,y,z\}\}$. Notice that the type space induced by both these experiments are same, that is,  
$$\mathcal R_{\mathcal M}=\{\{xyz\},\{xzy\},\{zxy\},\{zyx\},\{yzx\},\{yxz\}\}=\mathcal R_{\mathcal M'}.$$
Thus the menu $\{x,y,z\}$ is redundant in $\mathcal M'$. Removing it does not change the induced type space. By contrast, no menu in $\mathcal M$ is redundant: each binary menu distinguishes a pair of adjacent preferences that the other two binary menus do not distinguish. Hence, $\mathcal M$ is irreducible, while $\mathcal M'$ is not. Formally, we say that an experiment is irreducible if it contains no redundant menus.

Now, we are ready to present our main result for the case of three alternatives. 
\begin{proposition} Let $|X|=3$, and let $\widehat{\mathcal{E}}$ be a set of edges of the hexagon. Then,
$\widehat{\mathcal{E}}$ is identifiable if and only if $\widehat{\mathcal{E}}$
satisfies hexagon consistency. Moreover, if $\widehat{\mathcal{E}}$ is
identifiable, then there exists a unique irreducible experiment that identifies it.
\end{proposition}

For the case of three alternatives, hexagon consistency completely characterizes identifiability. Further, whenever a set of edges is identifiable, there exists a \textit{unique} irreducible experiment that identifies it and we explicitly determine it. It is worth noting that this unique irreducible experiment is contained in every identifying experiment, hence it is the common indispensable core of all identifying experiments.\footnote{We discuss this in more details in the next subsection (case of four alternatives).} Below, we illustrate our main result via an example.

Consider the edge set $\widebar{\mathcal E}$ that contains every edge of the hexagon except the edge $\{xyz,xzy\}$. In Figure \ref{fig5:permutohedron-3}, the edges in $\widebar{\mathcal E}$ are highlighted in red. It can be visually verified in Figure \ref{fig5:permutohedron-3} that $\widebar{\mathcal E}$ satisfies hexagon consistency.

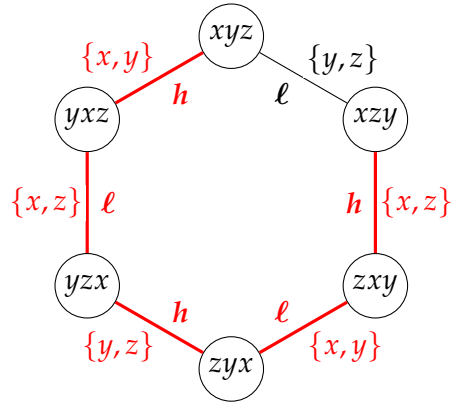
\begin{figure}[h]
	\centering
	\begin{tikzpicture}[
		scale=2.2,
		every node/.style={font=\small},
		pref/.style={circle, draw, inner sep=2pt},
		highlight/.style={draw=red, line width=1.2pt}
		]
		
		% vertices
		\node[pref] (xyz) at (90:1)   {$xyz$};
		\node[pref] (xzy) at (30:1)   {$xzy$};
		\node[pref] (zxy) at (-30:1)  {$zxy$};
		\node[pref] (zyx) at (-90:1)  {$zyx$};
		\node[pref] (yzx) at (-150:1) {$yzx$};
		\node[pref] (yxz) at (150:1)  {$yxz$};
		
		% edges with labels
		\draw (xyz) -- 
		node[midway, above right, fill=white, inner sep=1pt] {$\{y,z\}$}
		node[midway, fill=white, inner sep=0.8pt, xshift=-8pt, yshift=-6pt] {$\boldsymbol{\ell}$}
		(xzy);
		
		\draw[highlight] (xzy) -- 
		node[midway, right, fill=white, inner sep=1pt, text=red] {$\{x,z\}$}
		node[midway, fill=white, inner sep=0.8pt, xshift=-8pt, text=red] {$\boldsymbol{h}$}
		(zxy);
		
		\draw[highlight] (zxy) -- 
		node[midway, below right, fill=white, inner sep=1pt, text=red] {$\{x,y\}$}
		node[midway, fill=white, inner sep=0.8pt, xshift=-8pt, yshift=6pt, text=red] {$\boldsymbol{\ell}$}
		(zyx);
		
		\draw[highlight] (zyx) -- 
		node[midway, below left, fill=white, inner sep=1pt, text=red] {$\{y,z\}$}
		node[midway, fill=white, inner sep=0.8pt, xshift=8pt, yshift=6pt, text=red] {$\boldsymbol{h}$}
		(yzx);
		
		\draw[highlight] (yzx) -- 
		node[midway, left, fill=white, inner sep=1pt, text=red] {$\{x,z\}$}
		node[midway, fill=white, inner sep=0.8pt, xshift=8pt, text=red] {$\boldsymbol{\ell}$}
		(yxz);
		
		\draw[highlight] (yxz) -- 
		node[midway, above left, fill=white, inner sep=1pt, text=red] {$\{x,y\}$}
		node[midway, fill=white, inner sep=0.8pt, xshift=8pt, yshift=-6pt, text=red] {$\boldsymbol{h}$}
		(xyz);
		
	\end{tikzpicture}
	\caption{Identifiable set of edges $\widebar{\mathcal{E}}$.}
	\label{fig5:permutohedron-3}
\end{figure}

Define an experiment 
$$\mathcal M^*=\mathcal M_2^* \cup \mathcal M_3^*,$$
where $\mathcal M_2^*=\{\{x,y\},\{x,z\}\}$ is a collection of binary menus and $\mathcal M_3^*=\{\{x,y,z\}\}$ is a collection of ternary menus. Notice that from our previous discussions in this subsection,
\begin{enumerate}[(i)]
	\item the lower edge $\{zxy,zyx\}$ forces the binary menu $\{x,y\}$,
	\item the lower edge $\{yzx,yxz\}$ forces the binary menu $\{x,z\}$, and
	\item the hexagon is $\widebar{\mathcal{E}}$-active (because the lower edge $\{xyz,xzy\}$ is excluded but the opposite higher edge $\{yzx,zyx\}$ is included) and thus forces the menu $\{x,y,z\}$.
\end{enumerate}

Note that none of the menus in $\mathcal M^*$ distinguishes the omitted edge $\{xyz,xzy\}$. Therefore, $\mathcal M^*$ is the unique irreducible experiment that identifies $\widebar{\mathcal E}$. Consequently, for every experiment $\mathcal M$ that identifies $\widebar{\mathcal E}$, it must be the case that $\mathcal M^* \subseteq \mathcal M$, that is, $\mathcal M^*$ is the common indispensable core of \textit{all} experiments that identify $\widebar{\mathcal E}$.

\subsection{Four Alternatives} \label{4 alt}
Throughout this subsection, let $X=\{w,x,y,z\}$.\footnote{For the simplicity of presentation, we consider only four alternatives in this subsection. The results continue to hold for more than four alternatives.} Figure \ref{fig:permutohedron-4} depicts the corresponding permutohedron.\footnote{The figure is adapted from David Eppstein's public-domain Wikimedia Commons file \textit{Permutohedron.svg}; the vertex labels have been relabelled to match the notation used in this paper.} In addition to hexagons, the permutohedron contains squares. For example, see the square $(wyxz,ywxz,ywzx,wyzx,wyxz)$ in Figure \ref{fig:permutohedron-4}. Therefore, we need some kind of consistency condition on both squares and hexagons to characterize identifiability. First, we present a consistency condition on squares, which will be used to characterize identifiability.

\begin{figure}[h]
	\centering
	\includegraphics[width=.85\textwidth]{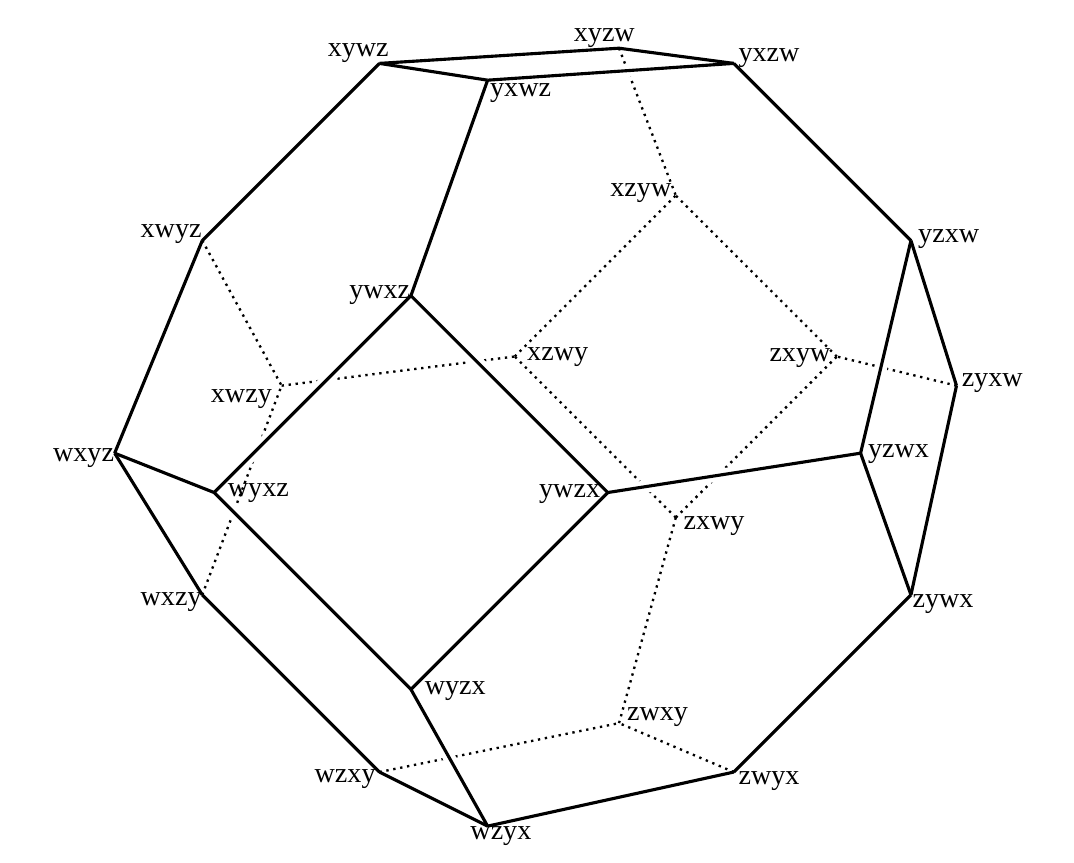}
	\caption{The permutohedron for four alternatives.}
	\label{fig:permutohedron-4}
\end{figure}

Consider any square, say $(wyxz,ywxz,ywzx,wyzx,wyxz)$ in Figure \ref{fig:permutohedron-4}, and any edge of the square, say $e=\{wyxz,ywxz\}$. In this square, $f=\{wyzx,ywzx\}$ is the opposite edge of $e$. Notice that since the preferences $wyxz$ and $ywxz$ differ only in the ranking of $w$ and $y$, any menu $M$ that distinguishes the edge $e$ must contain both $w$ and $y$, that is, $\{w,y\} \subseteq M$. Therefore, the menu $M$ must also distinguish the edge $f$. Using the same logic, we can say that any menu that distinguishes the edge $f$ must also distinguish the edge $e$. This phenomena holds for all squares of the permutohedron. Consequently, for any square and any menu $M$, opposite edges of the square are either both distinguished by $M$ or both are not distinguished by $M$. Now, we define square consistency.

\begin{definition}
	A set of edges $\widehat{\mathcal{E}}$ satisfies
	\textit{square consistency} if for every square of the permutohedron, opposite edges of the square are either both in
	$\widehat{\mathcal{E}}$ or both outside $\widehat{\mathcal{E}}$.
	
\end{definition}

Next, we present a consistency condition on hexagons, which will be used to characterize identifiability. In the previous subsection, for the case of three alternatives, the permutohedron was a single hexagon. For four (or more) alternatives, the permutohedron consists of multiple hexagons, see Figure \ref{fig:permutohedron-4}. We start with a similar analysis as in the previous subsection (for the case of three alternatives), however, the consistency condition on hexagons needed to characterize identifiability is \textit{not} the same when there are four (or more) alternatives. 

Consider the hexagon $(wxyz, wyxz, wyzx, wzyx, wzxy, wxzy, wxyz)$ of the permutohedron in Figure \ref{fig:permutohedron-4}. Here, the changing triple is $\{x,y,z\}$, that is, $x,y$ and $z$ are the only three alternatives whose ranks vary along the vertices of this hexagon. This hexagon is shown in Figure \ref{fig6:permutohedron-3-4} with the edge labels.

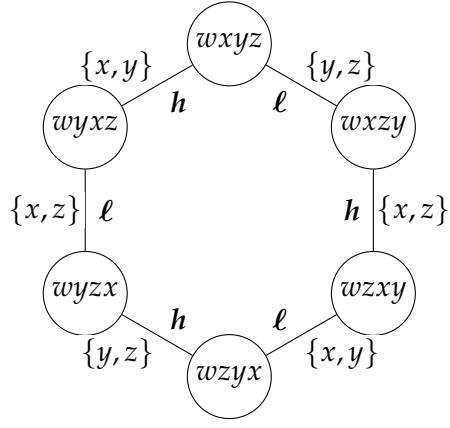
\begin{figure}[h]
	\centering
	\begin{tikzpicture}[
		scale=2.2,
		every node/.style={font=\small},
		pref/.style={circle, draw, inner sep=2pt}
		]
		
		% vertices
		\node[pref] (xyz) at (90:1)   {$wxyz$};
		\node[pref] (xzy) at (30:1)   {$wxzy$};
		\node[pref] (zxy) at (-30:1)  {$wzxy$};
		\node[pref] (zyx) at (-90:1)  {$wzyx$};
		\node[pref] (yzx) at (-150:1) {$wyzx$};
		\node[pref] (yxz) at (150:1)  {$wyxz$};
		
		% edges with labels
		\draw (xyz) -- 
		node[midway, above right, fill=white, inner sep=1pt] {$\{y,z\}$}
		node[midway, fill=white, inner sep=0.8pt, xshift=-8pt, yshift=-6pt] {$\boldsymbol{\ell}$}
		(xzy);
		
		\draw (xzy) -- 
		node[midway, right, fill=white, inner sep=1pt] {$\{x,z\}$}
		node[midway, fill=white, inner sep=0.8pt, xshift=-8pt] {$\boldsymbol{h}$}
		(zxy);
		
		\draw (zxy) -- 
		node[midway, below right, fill=white, inner sep=1pt] {$\{x,y\}$}
		node[midway, fill=white, inner sep=0.8pt, xshift=-8pt, yshift=6pt] {$\boldsymbol{\ell}$}
		(zyx);
		
		\draw (zyx) -- 
		node[midway, below left, fill=white, inner sep=1pt] {$\{y,z\}$}
		node[midway, fill=white, inner sep=0.8pt, xshift=8pt, yshift=6pt] {$\boldsymbol{h}$}
		(yzx);
		
		\draw (yzx) -- 
		node[midway, left, fill=white, inner sep=1pt] {$\{x,z\}$}
		node[midway, fill=white, inner sep=0.8pt, xshift=8pt] {$\boldsymbol{\ell}$}
		(yxz);
		
		\draw (yxz) -- 
		node[midway, above left, fill=white, inner sep=1pt] {$\{x,y\}$}
		node[midway, fill=white, inner sep=0.8pt, xshift=8pt, yshift=-6pt] {$\boldsymbol{h}$}
		(xyz);
		
	\end{tikzpicture}
	\caption{A hexagon of the permutohedron for four alternatives where the changing triple is ranked at bottom three positions.}
	\label{fig6:permutohedron-3-4}
\end{figure}

Consider any lower ($\boldsymbol{\ell}$)-edge in Figure \ref{fig6:permutohedron-3-4}, say $\{wxyz,wxzy\}$. Since the preferences $wxyz$ and $wxzy$ differ only in the ranking of alternatives $y$ and $z$ at bottom two positions, for any menu $M$ that distinguishes this edge, it must be the case that $M=\{y,z\}$. Therefore, lower edges of this particular hexagon force binary menus, in this case the edge $\{wxyz,wxzy\}$ forces the menu $M=\{y,z\}$. Further, notice that the menu $M=\{y,z\}$ must also distinguish the opposite higher ($\boldsymbol{h}$)-edge $\{wyzx,wzyx\}$. Hence, similar to the case of three alternatives, it is still true that if a lower edge $e$ of a hexagon is distinguished by a menu $M$, then $M$ also distinguishes the opposite higher edge $e^h$. 

It is worth noting that there are multiple hexagons when there are four (or more) alternatives. \textit{Only} the lower edges of a hexagon whose changing triple occupies the bottom three positions force binary menus.\footnote{Equivalently, only those edges that swap the bottom two positions force binary menus.} For example, consider the hexagon $(xywz, yxwz, ywxz, wyxz, wxyz, xwyz, xywz)$, where the changing triple $\{x,y,w\}$ is \textit{not} ranked at bottom three positions, as shown in Figure \ref{fig7:permutohedron-3-4}. Pick any lower edge of this hexagon, say, $\{ywxz, yxwz\}$. The preferences $ywxz$ and $yxwz$ differ only in the ranking of alternatives $w$ and $x$ but these two alternatives are not ranked at bottom two positions and $z$ is ranked at the bottom position at both preferences. Therefore, only two menus, namely $\{x,w\}$ and $\{x,w,z\}$, distinguish this lower edge. Hence, the lower edge $\{ywxz, yxwz\}$ cannot force a binary menu unlike the case of three alternatives. Moreover, notice that the two menus, namely $\{x,w\}$ and $\{x,w,z\}$, also distinguish the opposite higher edge $\{wxyz, xwyz\}$.

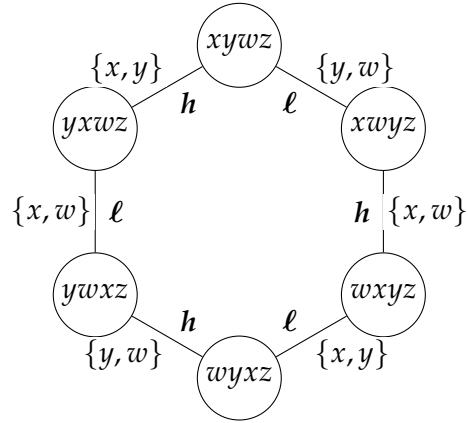
\begin{figure}[h]
	\centering
	\begin{tikzpicture}[
		scale=2.2,
		every node/.style={font=\small},
		pref/.style={circle, draw, inner sep=2pt}
		]
		
		% vertices
		\node[pref] (xyz) at (90:1)   {$xywz$};
		\node[pref] (xzy) at (30:1)   {$xwyz$};
		\node[pref] (zxy) at (-30:1)  {$wxyz$};
		\node[pref] (zyx) at (-90:1)  {$wyxz$};
		\node[pref] (yzx) at (-150:1) {$ywxz$};
		\node[pref] (yxz) at (150:1)  {$yxwz$};
		
		% edges with labels
		\draw (xyz) -- 
		node[midway, above right, fill=white, inner sep=1pt] {$\{y,w\}$}
		node[midway, fill=white, inner sep=0.8pt, xshift=-8pt, yshift=-6pt] {$\boldsymbol{\ell}$}
		(xzy);
		
		\draw (xzy) -- 
		node[midway, right, fill=white, inner sep=1pt] {$\{x,w\}$}
		node[midway, fill=white, inner sep=0.8pt, xshift=-8pt] {$\boldsymbol{h}$}
		(zxy);
		
		\draw (zxy) -- 
		node[midway, below right, fill=white, inner sep=1pt] {$\{x,y\}$}
		node[midway, fill=white, inner sep=0.8pt, xshift=-8pt, yshift=6pt] {$\boldsymbol{\ell}$}
		(zyx);
		
		\draw (zyx) -- 
		node[midway, below left, fill=white, inner sep=1pt] {$\{y,w\}$}
		node[midway, fill=white, inner sep=0.8pt, xshift=8pt, yshift=6pt] {$\boldsymbol{h}$}
		(yzx);
		
		\draw (yzx) -- 
		node[midway, left, fill=white, inner sep=1pt] {$\{x,w\}$}
		node[midway, fill=white, inner sep=0.8pt, xshift=8pt] {$\boldsymbol{\ell}$}
		(yxz);
		
		\draw (yxz) -- 
		node[midway, above left, fill=white, inner sep=1pt] {$\{x,y\}$}
		node[midway, fill=white, inner sep=0.8pt, xshift=8pt, yshift=-6pt] {$\boldsymbol{h}$}
		(xyz);
		
	\end{tikzpicture}
	\caption{A hexagon of the permutohedron for four alternatives where the changing triple is \textit{not} ranked at bottom three positions.}
	\label{fig7:permutohedron-3-4}
\end{figure}

Let $\widehat{\mathcal E}$ be an identifiable set of edges. Therefore, from the discussion in the preceding two paragraphs it follows that for a hexagon, if a lower edge $e$ of the hexagon belongs to $\widehat{\mathcal E}$, then $e^h \in \widehat{\mathcal E}$, the same conclusion we got for the case of three alternatives. Next, we analyze what happens if a hexagon is $\widehat{\mathcal{E}}$-active, that is, some lower edge $e \notin \widehat{\mathcal E}$ but $e^h \in \widehat{\mathcal E}$. Following a similar analysis to the case of three alternatives, we get the same conclusion, except for the uniqueness of menu. We illustrate these via an example below.  

Consider the hexagon depicted in Figure \ref{fig7:permutohedron-3-4} and pick any lower edge, say, $e=\{ywxz,yxwz\}$. Therefore, $e^h=\{wxyz,xwyz\}$. Suppose $e \notin \widehat{\mathcal E}$ but $e^h \in \widehat{\mathcal E}$. Since $\widehat{\mathcal E}$ is identifiable, there exists an experiment $\mathcal M$ that identifies $\widehat{\mathcal E}$. Therefore, there exists a menu $M \in \mathcal M$ such that $M$ distinguishes $e^h$ but does not distinguish $e$. Since the preferences $wxyz$ and $xwyz$ differ only in the ranking of alternatives $x$ and $w$, the fact that menu $M$ distinguishes $e^h=\{wxyz,xwyz\}$ implies $\{x,w\} \subseteq M$. If $M$ does not contain the alternative $y$, that is, $M=\{x,w\}$ or $M=\{x,w,z\}$, then $M$ also distinguishes the edge $e$, which will be a contradiction to the fact that $e \notin \widehat{\mathcal E}$. Therefore, it must be the case that $y \in M$, that is, $M=\{x,w,y\}$ or $M=\{x,w,y,z\}$. Both these menus distinguish $e^h$ but does not distinguish $e$. It is worth noting here that, unlike the case of three alternatives where the exclusion of a lower edge but inclusion of the opposite higher edge forces a unique menu, in this case it does \textit{not} force a unique menu. Since, either $M=\{x,w,y\}$ or $M=\{x,w,y,z\}$, the menu $M$ distinguishes \textit{all} higher edges of the hexagon. Hence, all higher edges of the hexagon must belong to $\widehat{\mathcal E}$.

Therefore, by the discussion in the previous paragraph, it follows that for identifiability, all higher edges of an $\widehat{\mathcal{E}}$-active hexagon must be included in $\widehat{\mathcal{E}}$, the same conclusion we got for the case of three alternatives. Summarizing, identifiability implies that every hexagon of the permutohedron must satisfy hexagon consistency (as defined in the previous subsection). Given that hexagon consistency characterizes identifiability for the case of three alternatives, at this point one might be tempted to claim that square consistency together with hexagon consistency (as defined in previous subsection) satisfied by every hexagon of the permutohedron characterizes identifiability. However, this is \textit{false}. Below, we provide an example (Example \ref{3 not 4}) of a set of edges of the permutohedron that is \textit{not} identifiable but satisfies square consistency and hexagon consistency (as defined in the previous subsection) for every hexagon of the permutohedron.

\begin{example}\label{3 not 4}
Recall that each edge of the permutohedron is labelled with the pair of alternatives reversed between the two adjacent preferences at the endpoints of the edge. Consider the set of edges $\widetilde{\mathcal{E}}$ which contains
\begin{enumerate}[(i)]
	
	\item none of the edges with the label $\{w,z\}$,
	\item the two edges with label $\{x,y\}$, where the swap between $x$ and $y$ occurs at top two positions, and 
	\item all edges with labels $\{w,x\}, \{w,y\}, \{x,z\}$ and $\{y,z\}$.
\end{enumerate}

\begin{figure}[h]
	\centering
	\includegraphics[width=.85\textwidth]{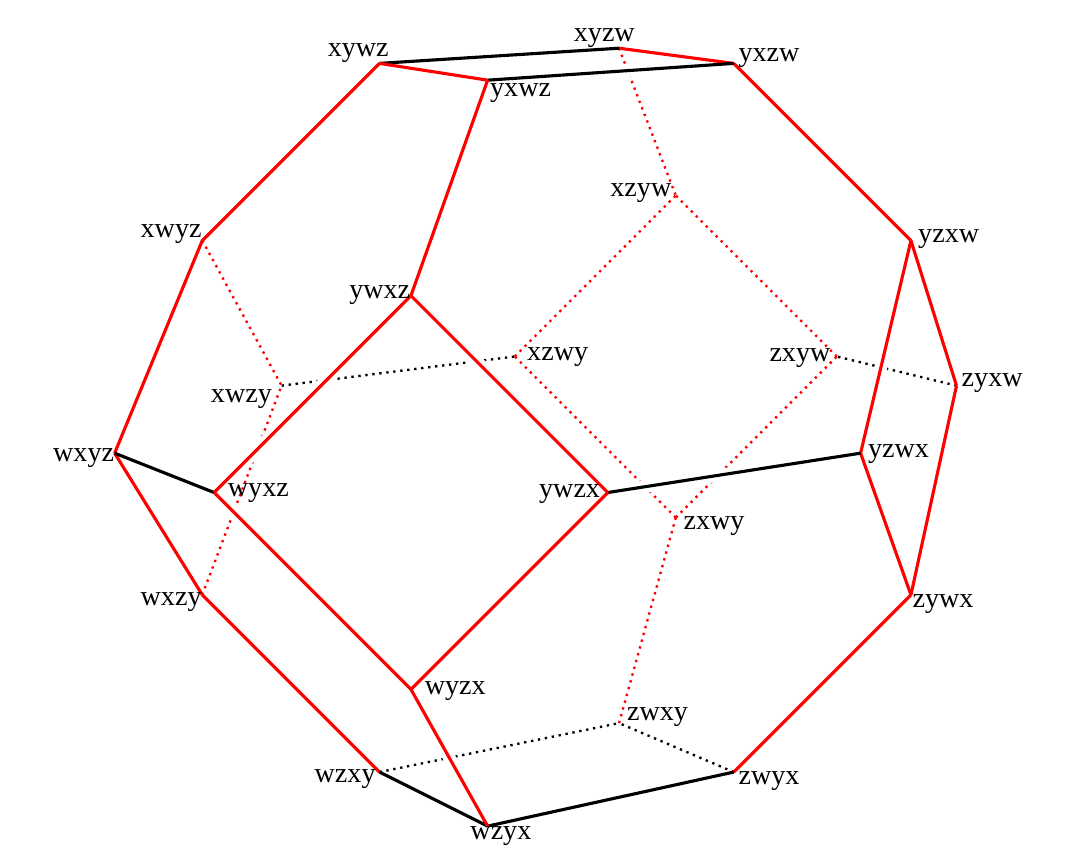}
	\caption{Edges in $\widetilde{\mathcal{E}}$ are depicted in red colour.}
	\label{fig:permutohedron-4_red}
\end{figure}

The edges of the permutohedron that belong to $\widetilde{\mathcal{E}}$ are shown in red colour in Figure \ref{fig:permutohedron-4_red}. It is not hard to check (visually) that $\widetilde{\mathcal{E}}$ satisfies square consistency and hexagon consistency (as defined in the previous subsection) for every hexagon of the permutohedron. Now, we argue that $\widetilde{\mathcal{E}}$ is \textit{not} identifiable. 

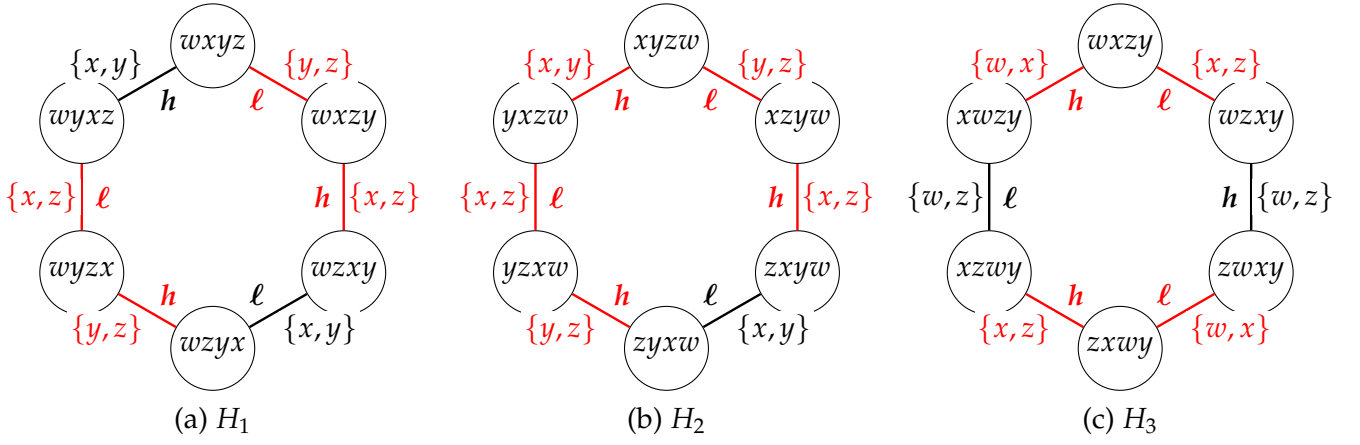
\begin{figure}[h]
	\centering
	\begin{tikzpicture}[
		scale=2,
		every node/.style={font=\small},
		pref/.style={circle, draw, inner sep=2pt},
		rededge/.style={draw=red, line width=0.9pt},
		blackedge/.style={draw=black, line width=0.9pt},
		redlabel/.style={midway, fill=white, inner sep=1pt, text=red},
		blacklabel/.style={midway, fill=white, inner sep=1pt, text=black},
		redlevel/.style={midway, fill=white, inner sep=0.8pt,
			text=red},
		blacklevel/.style={midway, fill=white, inner sep=0.8pt,
			text=black}
		]
		
		%==================================================
		% (a) Changing triple {x,y,z} at bottom three positions
		%==================================================
		\begin{scope}[shift={(-1,0)}]
			
			% vertices
			\node[pref] (axyz) at (90:1)   {$wxyz$};
			\node[pref] (axzy) at (30:1)   {$wxzy$};
			\node[pref] (azxy) at (-30:1)  {$wzxy$};
			\node[pref] (azyx) at (-90:1)  {$wzyx$};
			\node[pref] (ayzx) at (-150:1) {$wyzx$};
			\node[pref] (ayxz) at (150:1)  {$wyxz$};
			
			% edges with labels
			\draw[rededge] (axyz) --
			node[redlabel, above right] {$\{y,z\}$}
			node[redlevel, xshift=-8pt, yshift=-6pt]
			{$\boldsymbol{\ell}$}
			(axzy);
			
			\draw[rededge] (axzy) --
			node[redlabel, right] {$\{x,z\}$}
			node[redlevel, xshift=-8pt]
			{$\boldsymbol{h}$}
			(azxy);
			
			\draw[blackedge] (azxy) --
			node[blacklabel, below right] {$\{x,y\}$}
			node[blacklevel, xshift=-8pt, yshift=6pt]
			{$\boldsymbol{\ell}$}
			(azyx);
			
			\draw[rededge] (azyx) --
			node[redlabel, below left] {$\{y,z\}$}
			node[redlevel, xshift=8pt, yshift=6pt]
			{$\boldsymbol{h}$}
			(ayzx);
			
			\draw[rededge] (ayzx) --
			node[redlabel, left] {$\{x,z\}$}
			node[redlevel, xshift=8pt]
			{$\boldsymbol{\ell}$}
			(ayxz);
			
			\draw[blackedge] (ayxz) --
			node[blacklabel, above left] {$\{x,y\}$}
			node[blacklevel, xshift=8pt, yshift=-6pt]
			{$\boldsymbol{h}$}
			(axyz);
			
			\node at (0,-1.48) {\textnormal{(a) $H_1$}};
			
		\end{scope}
		
		%==================================================
		% (b) Changing triple {x,y,z} at top three positions
		%==================================================
		\begin{scope}[shift={(2,0)}]
			
			% vertices
			\node[pref] (bxyz) at (90:1)   {$xyzw$};
			\node[pref] (bxzy) at (30:1)   {$xzyw$};
			\node[pref] (bzxy) at (-30:1)  {$zxyw$};
			\node[pref] (bzyx) at (-90:1)  {$zyxw$};
			\node[pref] (byzx) at (-150:1) {$yzxw$};
			\node[pref] (byxz) at (150:1)  {$yxzw$};
			
			% edges with labels
			\draw[rededge] (bxyz) --
			node[redlabel, above right] {$\{y,z\}$}
			node[redlevel, xshift=-8pt, yshift=-6pt]
			{$\boldsymbol{\ell}$}
			(bxzy);
			
			\draw[rededge] (bxzy) --
			node[redlabel, right] {$\{x,z\}$}
			node[redlevel, xshift=-8pt]
			{$\boldsymbol{h}$}
			(bzxy);
			
			\draw[blackedge] (bzxy) --
			node[blacklabel, below right] {$\{x,y\}$}
			node[blacklevel, xshift=-8pt, yshift=6pt]
			{$\boldsymbol{\ell}$}
			(bzyx);
			
			\draw[rededge] (bzyx) --
			node[redlabel, below left] {$\{y,z\}$}
			node[redlevel, xshift=8pt, yshift=6pt]
			{$\boldsymbol{h}$}
			(byzx);
			
			\draw[rededge] (byzx) --
			node[redlabel, left] {$\{x,z\}$}
			node[redlevel, xshift=8pt]
			{$\boldsymbol{\ell}$}
			(byxz);
			
			\draw[rededge] (byxz) --
			node[redlabel, above left] {$\{x,y\}$}
			node[redlevel, xshift=8pt, yshift=-6pt]
			{$\boldsymbol{h}$}
			(bxyz);
			
			\node at (0,-1.48) {\textnormal{(b) $H_2$}};
			
		\end{scope}
		
		%==================================================
		% (c) Changing triple {w,x,z} at top three positions
		%==================================================
		\begin{scope}[shift={(5,0)}]
			
			% vertices
			\node[pref] (cwxz) at (90:1)   {$wxzy$};
			\node[pref] (cwzx) at (30:1)   {$wzxy$};
			\node[pref] (czwx) at (-30:1)  {$zwxy$};
			\node[pref] (czxw) at (-90:1)  {$zxwy$};
			\node[pref] (cxzw) at (-150:1) {$xzwy$};
			\node[pref] (cxwz) at (150:1)  {$xwzy$};
			
			% edges with labels
			\draw[rededge] (cwxz) --
			node[redlabel, above right] {$\{x,z\}$}
			node[redlevel, xshift=-8pt, yshift=-6pt]
			{$\boldsymbol{\ell}$}
			(cwzx);
			
			\draw[blackedge] (cwzx) --
			node[blacklabel, right] {$\{w,z\}$}
			node[blacklevel, xshift=-8pt]
			{$\boldsymbol{h}$}
			(czwx);
			
			\draw[rededge] (czwx) --
			node[redlabel, below right] {$\{w,x\}$}
			node[redlevel, xshift=-8pt, yshift=6pt]
			{$\boldsymbol{\ell}$}
			(czxw);
			
			\draw[rededge] (czxw) --
			node[redlabel, below left] {$\{x,z\}$}
			node[redlevel, xshift=8pt, yshift=6pt]
			{$\boldsymbol{h}$}
			(cxzw);
			
			\draw[blackedge] (cxzw) --
			node[blacklabel, left] {$\{w,z\}$}
			node[blacklevel, xshift=8pt]
			{$\boldsymbol{\ell}$}
			(cxwz);
			
			\draw[rededge] (cxwz) --
			node[redlabel, above left] {$\{w,x\}$}
			node[redlevel, xshift=8pt, yshift=-6pt]
			{$\boldsymbol{h}$}
			(cwxz);
			
			\node at (0,-1.48) {\textnormal{(c) $H_3$}};
			
		\end{scope}
		
	\end{tikzpicture}
	\caption{Three hexagonal faces with the inherited red colour edges.}
	\label{fig:three-hexagons-face-consistency}
\end{figure}

Suppose, to the contrary, that $\widetilde{\mathcal{E}}$ is identifiable. Consider the three hexagons $H_1, H_2$ and $H_3$ of the permutohedron, as shown in Figure \ref{fig:three-hexagons-face-consistency}. Note that only the red coloured edges of these hexagons belong to $\widetilde{\mathcal{E}}$. Since $\widetilde{\mathcal{E}}$ is identifiable, there exists an experiment $\mathcal M$ that identifies $\widetilde{\mathcal{E}}$. Therefore, for the hexagon $H_2$, there exists a menu $M \in \mathcal M$ such that $M$ does not distinguish the non-red lower ($\boldsymbol{\ell}$) edge $e=\{zyxw,zxyw\}$ but $M$ distinguishes the opposite higher ($\boldsymbol{h}$) edge $e^h=\{yxzw,xyzw\}$, as shown in Figure \ref{fig:three-hexagons-face-consistency} (b). The preferences at the endpoints of the edge $e^h$ differ only in the ranking of $x$ and $y$. Therefore, since $M$ distinguishes the edge $e^h$, it must be the case that $M$ contains both $x$ and $y$, that is, $\{x,y\} \subseteq M$. However, at both the endpoints of the edge $e$, the alternative $z$ is ranked above both $x$ and $y$. Consequently, if $z\notin M$, the menu $M$ would also distinguish $e$, contrary to $e\notin\widetilde{\mathcal E}$. Thus $z\in M$, and hence, either $M=\{x,y,z\}$ or $M=\{w,x,y,z\}$.

The possibility that $M=\{x,y,z\}$ is ruled out by the hexagon $H_1$ (Figure \ref{fig:three-hexagons-face-consistency} (a)). To see this, consider the non-red higher ($\boldsymbol{h}$) edge $f=\{wyxz, wxyz\}$ of $H_1$. If $M=\{x,y,z\}$, then $M$ distinguishes the edge $f$. Since $M \in \mathcal M$ and $f \notin \widetilde{\mathcal E}$, this leads to a contradiction to the fact that $\mathcal M$ identifies $\widetilde{\mathcal E}$. Therefore, it must be the case that $M=\{w,x,y,z\}$.

Finally, consider the hexagon $H_3$ (Figure \ref{fig:three-hexagons-face-consistency} (c)). Since $M=\{w,x,y,z\}$, the menu $M$ distinguishes all higher ($\boldsymbol{h}$) edges of $H_3$, in particular, the menu $M$ distinguishes the non-red higher edge $\{wzxy, zwxy\}$ of $H_3$. This contradicts to the fact that $\{wzxy, zwxy\} \notin \widetilde{\mathcal E}$. Therefore, $\widetilde{\mathcal E}$ is not identifiable. This completes the example. \hfill$\square$

\end{example}

The preceding example (Example \ref{3 not 4}) reveals the additional restriction required to characterize identifiability for four (or more) alternatives. It also foreshadows why uniqueness of the irreducible experiment, established for three alternatives in the previous subsection, continues to hold with four or more alternatives. Before making these results precise, we introduce some notations. 

For any hexagon $H$ of the permutohedron, let $C(H)$ denote its changing triple. For example, in Figure \ref{fig:three-hexagons-face-consistency}, $$C(H_1)=\{x,y,z\}=C(H_2) \mbox{ and } C(H_3)=\{w,x,z\}.$$ Also, let $A(H)$ denote the associated set of the hexagon $H$, which consists of all three alternatives in $C(H)$ together with all alternatives that are ranked below them at every vertex (preference) of the hexagon $H$. For example, in Figure \ref{fig:three-hexagons-face-consistency}, $$A(H_1)=\{x,y,z\} \mbox{ and } A(H_2)=\{x,y,z,w\}=A(H_3).$$ 

\begin{definition}\label{Primitive}
	Let $H$ be a hexagon and $\widehat{\mathcal E}$ be a set of edges. We say that $H$ is $\widehat{\mathcal E}$-primitive if 
	\begin{enumerate}[(i)]
		\item $H$ is $\widehat{\mathcal E}$-active, and
		\item there is no $\widehat{\mathcal E}$-active hexagon $H'$ such that $C(H')=C(H)$ and $A(H') \subsetneq A(H)$.\footnote{The notation ``$\subsetneq$" denotes a strict subset.}
	\end{enumerate}
\end{definition}

Recall that a hexagon $H$ is $\widehat{\mathcal E}$-active if for some lower edge $e$ of $H$, $e \notin \widehat{\mathcal{E}}$ but $e^h \in \widehat{\mathcal{E}}$, that is, a lower edge of $H$ is excluded but its opposite higher edge of $H$ is included. For a fixed changing triple, condition (ii) of Definition \ref{Primitive} orders active hexagons according to inclusion of their associated sets. A primitive hexagon is an inclusion-minimal occurrence of the active pattern: no active hexagon with the same changing triple has a strictly smaller associated set.

Going back to Example \ref{3 not 4}, in Figure \ref{fig:three-hexagons-face-consistency}, the hexagon $H_1$ is not $\widetilde{\mathcal E}$-active. Indeed, its excluded lower edge labelled $\{x,y\}$ has an excluded opposite higher edge, whereas all other edges are included. By contrast, $H_2$ is $\widetilde{\mathcal E}$-active: its lower edge labelled $\{x,y\}$ is excluded, whereas its opposite higher edge is included. Moreover, $H_2$ is $\widetilde{\mathcal E}$-primitive. This is because the only hexagon of the permutohedron with the same changing triple and a strictly smaller associated set is $H_1$, which is not $\widetilde{\mathcal E}$-active.

The arguments in Example \ref{3 not 4} explain the significance of primitive hexagons. The active pattern in $H_2$ requires an identifying experiment to contain a menu that distinguishes the included higher edge without distinguishing its excluded opposite lower edge. Such a menu must contain the changing triple $\{x,y,z\}$. The menu $\{x,y,z\}$ is ruled out by $H_1$, and therefore the menu is forced to be
$$
A(H_2)=\{w,x,y,z\}.
$$
Equivalently, if a proper subset of $A(H_2)$ could implement the required distinction in $H_2$, then the same active pattern would already occur in a hexagon with changing triple $C(H_2)$ and a strictly smaller associated set, contradicting the primitivity of $H_2$.

The implications of this forced menu ($A(H_2)$) are not only confined to $H_2$. Once the menu $A(H_2)=\{w,x,y,z\}$ is forced to be included in every identifying experiment, it distinguishes every higher edge of any hexagon $H$ satisfying $A(H)=A(H_2)$.
In particular, $A(H_3)=A(H_2)=\{w,x,y,z\}$,
so every higher edge of $H_3$ must belong to $\widetilde{\mathcal E}$. The failure of this closure requirement is precisely why $\widetilde{\mathcal E}$ is not identifiable.

Primitive hexagons therefore play two roles in the characterization. First, each primitive hexagon generates a restriction across all hexagons having the same associated set. Second, it forces the associated set itself to appear as a menu in \textit{every} identifying experiment. We show in the proof of Theorem \ref{Characterization_of_Identifiability} that primitive hexagons are precisely the source of the non-binary menus in the unique irreducible experiment: every primitive hexagon $H$ forces the menu $A(H)$, and every indispensable menu of size at least three arises in this way. Together with the forced binary menus (corresponding to lower edges of a hexagon whose changing triple are ranked at bottom three positions, as discussed earlier), these associated sets determine the unique irreducible experiment.

Now, we strengthen the hexagon consistency condition that we introduced for three alternatives in the previous subsection.

\begin{definition}\label{Hex_consistency_4}
	A set of edges $\widehat{\mathcal{E}}$ satisfies
	\textit{hexagon consistency} if the following conditions hold:
	
	\begin{enumerate}[(i)]
		\item For every hexagon $H$ of the permutohedron and every lower edge $e$ of
		$H, e\in \widehat{\mathcal E}
		\implies
		e^h\in \widehat{\mathcal E}$.
		
		\item For every $\widehat{\mathcal{E}}$-primitive hexagon
		$H$ and every hexagon $K$ such that $A(K)=A(H)$, every higher edge of $K$ belongs to $\widehat{\mathcal{E}}$.
	\end{enumerate}
\end{definition}

Condition (ii) of Definition \ref{Hex_consistency_4} captures the additional restriction revealed by Example~\ref{3 not 4}. A primitive hexagon forces the menu given by its associated set, and this menu necessarily distinguishes every higher edge of every hexagon having that same associated set. In particular, by taking $K=H$ in condition~(ii), every higher edge of an $\widehat{\mathcal E}$-primitive hexagon $H$ must belong to $\widehat{\mathcal E}$.

When there are three alternatives, the permutohedron consists of a single hexagon. Consequently, Definition~\ref{Hex_consistency_4} reduces to Definition~\ref{Hex_consistency_3}.

We are now ready to state the characterization condition. The square and hexagonal faces of the permutohedron considered above are precisely the two-dimensional faces of the permutohedron; see \citet{morton2009convex}. This motivates the following terminology. 

\begin{definition}
	A set of edges $\widehat{\mathcal{E}}$ of the permutohedron satisfies
\textit{face consistency} if $\widehat{\mathcal{E}}$ satisfies both square consistency
and hexagon consistency.
\end{definition}

Our main result says that face consistency completely characterizes identifiability. Moreover, whenever an edge set is identifiable, each of its identifying experiment has a common indispensable core: the unique irreducible identifying experiment.

\begin{proposition}\label{4_pro}
 Let $\widehat{\mathcal{E}}$ be a set of edges of the permutohedron. Then, $\widehat{\mathcal{E}}$ is identifiable if and only if $\widehat{\mathcal{E}}$ satisfies face consistency. Moreover, for every identifiable $\widehat{\mathcal{E}}$, there exists a unique irreducible experiment that identifies it.
 \end{proposition}

The uniqueness conclusion in Proposition \ref{4_pro} has a stronger implication. Suppose $\widehat{\mathcal E}$ is identifiable and let $\mathcal M^*$ denote the unique irreducible experiment identifying $\widehat{\mathcal E}$. Starting from any experiment $\mathcal M$ that identifies $\widehat{\mathcal E}$, one may successively remove redundant menus until an irreducible identifying subexperiment remains. By uniqueness, the resulting subexperiment must be $\mathcal M^*$. Hence, for every experiment $\mathcal M$ that identifies $\widehat{\mathcal E}$ it must be the case that
$$
\mathcal M^*\subseteq\mathcal M.
$$
Thus, $\mathcal M^*$ is the common indispensable core of all identifying experiments and we \textit{explicitly} determine $\mathcal M^*$. Every menu in $\mathcal M^*$ must appear in every experiment identifying $\widehat{\mathcal E}$. Moreover, if an identifying experiment contains menus in addition to those in $\mathcal M^*$, each such menu is redundant because removing it leaves an experiment that still contains $\mathcal M^*$ and hence still identifies $\widehat{\mathcal E}$.

In particular, $\mathcal M^*$ is the unique identifying experiment
with the minimum number of menus. Since $\mathcal M^*$ is contained
in every identifying experiment, it is also optimal under any cost
criterion for which adding menus cannot reduce the cost of an
experiment. This monotonicity requirement is the only restriction
placed on the cost criterion. In particular, the cost need not depend
only on the number of menus. Different menus may have different
costs: for example, one binary menu may be more costly to administer
than another binary menu. More generally, cost may depend on the
identities and sizes of the menus, or on interactions among them.
We impose neither additivity nor symmetry across menus.

The main result therefore solves both an identification problem and an experimental design problem. It characterizes which sets of distinctions can be generated through menu choices and, whenever identification is possible, determines explicitly the unique collection of menus that every identifying experiment must contain.

\section{The Main Result}\label{Main}
In this section we formally characterize identifiable sets of edges of the permutohedron. First, we define some terminologies that we will need.

For any preference $P$ and any $k \in \{1,\ldots,n\}$, let us denote by $r_k(P)$ the $k$-th ranked alternative according to $P$, that is, $r_k(P)=x$ if $|\{y \in X : yPx\}|=k-1$. In such scenarios, where $r_k(P)=x$, we say that $k$ is the rank of $x$ at the preference $P$. For $x,z \in X$, we say that $z$ is ranked above $x$ at the preference $P$ if $zPx$, equivalently if the rank of $z$ at $P$ is smaller than the rank of $x$ at $P$. Similarly, we say that $z$ is ranked below $x$ at $P$ if $xPz$, equivalently if the rank of $x$ at $P$ is smaller than the rank of $z$ at $P$.

For $i \in \{1,...,n-1\}$, let $P^i$ denote the preference obtained from $P$ by interchanging the $i$-th and $(i+1)$-th ranked alternative and keeping the ranks of all other alternatives the same, that is, $r_i(P^i)=r_{i+1}(P)$,  $r_{i+1}(P^i)=r_{i}(P)$, and  $r_k(P^i)=r_k(P)$ for all $k \in \{1,2,\ldots,n\} \setminus \{i,i+1\}$. Similarly, for $i,j \in \{1,...,n-1\}$, define $P^{ij}:= (P^i)^j$ as the preference obtained from $P^i$ by interchanging the $j$-th and $(j+1)$-th ranked alternative. Also, for $i, j, l \in \{1,...,n-1\}$, we similarly define $P^{ijl}$ in a similar manner.

A \textit{square} of the permutohedron is a cycle of the form $(P,P^i,P^{ij},P^j,P)$, where $|i-j|>1$. Notice that each consecutive pair of preferences are adjacent. Intuitively, squares arise from pairs of adjacent swaps that commute. Figure \ref{fig:square} illustrates a square $(P,P^i,P^{ij},P^j,P)$ of the permutohedron,  where $j-i>1$. Let $\mathcal{S}$ be the set of all squares of the permutohedron.

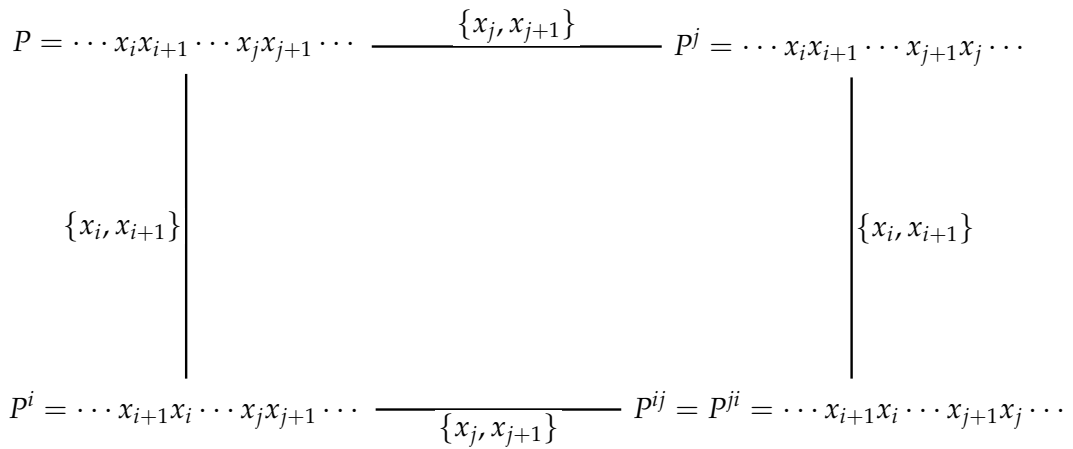
\begin{figure}[h]
	\centering
	\begin{tikzpicture}[
		x=4.4cm,
		y=3.2cm,
		every node/.style={font=\small},
		every path/.style={draw, line width=0.9pt}
		]
		
		% vertices
		\node (P)    at (-1, 0.75) {$P=\cdots x_i x_{i+1}\cdots x_j x_{j+1}\cdots$};
		\node (Pj)   at ( 1, 0.75) {$P^j=\cdots x_i x_{i+1}\cdots x_{j+1} x_j\cdots$};
		\node (Pij)  at ( 1,-0.75) {$P^{ij}=P^{ji}=\cdots x_{i+1} x_i\cdots x_{j+1} x_j\cdots$};
		\node (Pi)   at (-1,-0.75) {$P^i=\cdots x_{i+1} x_i\cdots x_j x_{j+1}\cdots$};
		
		% edges with labels
		\draw (P) -- 
		node[midway, above, fill=white, inner sep=1pt] {$\{x_j,x_{j+1}\}$}
		(Pj);
		
		\draw (Pj) -- 
		node[midway, right, fill=white, inner sep=1pt] {$\{x_i,x_{i+1}\}$}
		(Pij);
		
		\draw (Pij) -- 
		node[midway, below, fill=white, inner sep=1pt] {$\{x_j,x_{j+1}\}$}
		(Pi);
		
		\draw (Pi) -- 
		node[midway, left, fill=white, inner sep=1pt] {$\{x_i,x_{i+1}\}$}
		(P);
		
	\end{tikzpicture}
	\caption{A square of the permutohedron where $X=\{x_1,\ldots,x_i,x_{i+1},\dots,x_j, x_{j+1},\ldots,x_n\}$.}
	\label{fig:square}
\end{figure}

\begin{definition}\label{square_consistency}
	A set of edges $\widehat{\mathcal{E}}\subseteq \mathcal{E}$ satisfies
	\textit{square consistency} if, for every square
	$(P,P^i,P^{ij},P^j,P) \in \mathcal{S}$,
	\begin{enumerate}[(i)]
		\item $\{P,P^i\}\in \widehat{\mathcal{E}}$ \textit{if and only if} $\{P^{ij},P^j\}\in \widehat{\mathcal{E}}$, and
		\item $\{P,P^j\}\in \widehat{\mathcal{E}}$ \textit{if and only if} $\{P^i,P^{ij}\}\in \widehat{\mathcal{E}}$.
	\end{enumerate}
	
\end{definition}

Equivalently, square consistency says that opposite edges of every square are either both in
$\widehat{\mathcal{E}}$ or both outside $\widehat{\mathcal{E}}$.

A \textit{hexagon} of the permutohedron is a cycle of the form $(P,P^i,P^{i{(i+1)}},P^{i{(i+1)}i},P^{{(i+1)}i},P^{(i+1)},P)$, where $1 \leq i \leq n-2$, as shown in Figure \ref{fig:hexagon}. Note that each consecutive pair of preferences are adjacent. So, intuitively hexagons arise from adjacent swaps on three consecutive ranks, where the changing triple occupies ranks $i, i+1$ and $i+2$. We label each edge in the hexagon as a lower $(\boldsymbol{\ell})$-edge or a higher $(\boldsymbol{h})$-edge as illustrated in Figure \ref{fig:hexagon}. An edge is labelled $\boldsymbol{h}$ if the corresponding adjacent swap takes place at ranks $i$ and $i+1$. Similarly, an edge is labelled $\boldsymbol{\ell}$ if the corresponding adjacent swap takes place at ranks $i+1$ and $i+2$. As illustrated in Figure \ref{fig:hexagon}, every lower ($\boldsymbol{\ell}$)-edge has a unique higher ($\boldsymbol{h}$)-edge on the opposite side of the hexagon with the same reversed pair of alternatives. We call this higher edge as the opposite higher edge of the lower edge.
More formally, if a lower edge $e$ is labelled by the pair $\{u,v\}$, then its opposite higher edge, denoted by $e^h$, is the unique higher edge in the same hexagon that is also labelled by $\{u,v\}$. Let $\mathcal{H}$ be the set of all hexagons of the permutohedron.

\begin{figure}[h]
	\centering
	\begin{tikzpicture}[
		x=4.2cm,
		y=3.0cm,
		every node/.style={font=\small},
		every path/.style={draw, line width=0.9pt}
		]
		
		% vertices
		\node (P)        at (0,1)         {$P=\cdots x_i x_{i+1} x_{i+2}\cdots$};
		\node (Pip1)     at (1.15,0.58)   {$P^{({i+1})}=\cdots x_i x_{i+2} x_{i+1}\cdots$};
		\node (Pip1i)    at (1.15,-0.58)  {$P^{({i+1})i}=\cdots x_{i+2} x_i x_{i+1}\cdots$};
		\node (Piip1i)   at (0,-1.16)     {$P^{i({i+1})i}=\cdots x_{i+2} x_{i+1} x_i\cdots$};
		\node (Piip1)    at (-1.15,-0.58) {$P^{i({i+1})}=\cdots x_{i+1} x_{i+2} x_i\cdots$};
		\node (Pi)       at (-1.15,0.58)  {$P^i=\cdots x_{i+1} x_i x_{i+2}\cdots$};
		
		% edges with labels
		\draw (P) -- 
		node[midway, above right, fill=white, inner sep=1pt] {$\{x_{i+1},x_{i+2}\}$}
		node[midway, fill=white, inner sep=0.8pt, xshift=-8pt, yshift=-6pt] {$\boldsymbol{\ell}$}
		(Pip1);
		
		\draw (Pip1) -- 
		node[midway, right, fill=white, inner sep=1pt] {$\{x_i,x_{i+2}\}$}
		node[midway, fill=white, inner sep=0.8pt, xshift=-8pt] {$\boldsymbol{h}$}
		(Pip1i);
		
		\draw (Pip1i) -- 
		node[midway, below right, fill=white, inner sep=1pt] {$\{x_i,x_{i+1}\}$}
		node[midway, fill=white, inner sep=0.8pt, xshift=-8pt, yshift=6pt] {$\boldsymbol{\ell}$}
		(Piip1i);
		
		\draw (Piip1i) -- 
		node[midway, below left, xshift=-10pt, fill=white, inner sep=1pt] {$\{x_{i+1},x_{i+2}\}$}
		node[midway, fill=white, inner sep=0.8pt, xshift=8pt, yshift=6pt] {$\boldsymbol{h}$}
		(Piip1);
		
		\draw (Piip1) -- 
		node[midway, left, fill=white, inner sep=1pt] {$\{x_i,x_{i+2}\}$}
		node[midway, fill=white, inner sep=0.8pt, xshift=8pt] {$\boldsymbol{\ell}$}
		(Pi);
		
		\draw (Pi) -- 
		node[midway, above left, fill=white, inner sep=1pt] {$\{x_i,x_{i+1}\}$}
		node[midway, fill=white, inner sep=0.8pt, xshift=8pt, yshift=-6pt] {$\boldsymbol{h}$}
		(P);
		
	\end{tikzpicture}
	\caption{A hexagon of the permutohedron where $X=\{x_1,\ldots,x_i,x_{i+1},x_{i+2},\ldots,x_n\}$.}
	\label{fig:hexagon}
\end{figure}
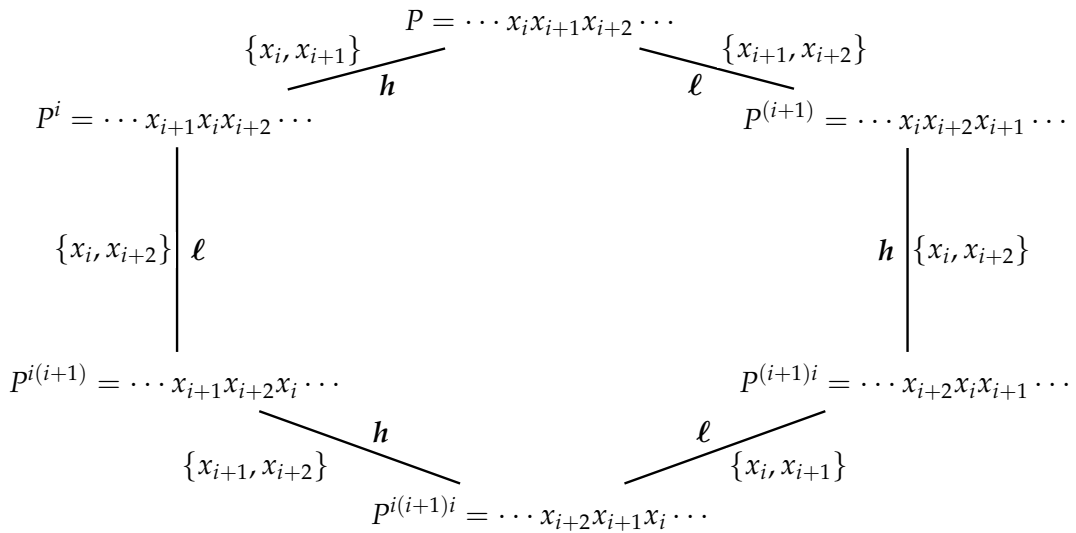

For a hexagon
$$
H=\bigl(P,P^i,P^{i(i+1)},P^{i(i+1)i},P^{(i+1)i},P^{(i+1)},P\bigr)
\in \mathcal H,
$$
let
$
C(H)=\{r_i(P),r_{i+1}(P),r_{i+2}(P)\}
$
be the changing triple, that is, the set of three alternatives whose ranks vary along the
hexagon $H$. Let
$$
A(H)=\{r_i(P),r_{i+1}(P),\ldots,r_n(P)\}.
$$
Thus $A(H)$ is the associated set of the hexagon $H$, consisting of the three alternatives in
$C(H)$ together with all alternatives ranked below them at the preference $P$. For example, consider the hexagon $H$ as shown in Figure \ref{fig:hexagon}. Here, $C(H)=\{x_i,x_{i+1},x_{i+2}\}$ and $A(H)=\{x_i,x_{i+1},x_{i+2},x_{i+3},\ldots,x_n\}$. Notice that each alternative in $A(H) \setminus C(H)$ is ranked below every alternative in $C(H)$ at every vertex (preference) in the hexagon $H$. Moreover,
$C(H)$ and $A(H)$ are well-defined for the hexagon $H$, independent
of the particular vertex (preference) of $H$ used to represent it.

\begin{definition}
	Fix $\widehat{\mathcal E}\subseteq \mathcal E$. A hexagon $H\in\mathcal H$
	is called $\widehat{\mathcal E}$-\textit{active} if there exists a lower
	edge $e$ of $H$ such that $e\notin \widehat{\mathcal E}$ and $e^h\in \widehat{\mathcal E}$.
\end{definition}

\begin{definition}
	An $\widehat{\mathcal E}$-active hexagon $H$ is called
	$\widehat{\mathcal E}$-\textit{primitive} if there does not exist an
	$\widehat{\mathcal E}$-active hexagon $H'\in\mathcal H$ such that
	\begin{enumerate}[(i)]
		\item $C(H')=C(H)$, and
		\item $A(H')\subsetneq A(H)$\footnote{By $\subsetneq$, we mean strict subset.}.
	\end{enumerate}
	
\end{definition}

	When $\widehat{\mathcal E}$ is clear from the context, we simply say active
and primitive hexagons.

	\begin{definition} \label{hexagon_consistency}
	A set of edges $\widehat{\mathcal{E}}\subseteq \mathcal{E}$ satisfies
	\textit{hexagon consistency} if the following conditions hold:
	
	\begin{enumerate}[(i)]
		\item For every hexagon $H \in \mathcal{H}$ and every lower edge $e$ of
		$H, e\in \widehat{\mathcal E}
		\implies
		e^h\in \widehat{\mathcal E}$.
		
		\item For every $\widehat{\mathcal{E}}$-primitive hexagon
		$H\in\mathcal{H}$ and every hexagon $K\in\mathcal{H}$ such that $A(K)=A(H)$, every higher edge of $K$ belongs to $\widehat{\mathcal{E}}$.\footnote{Condition (ii), applied with $K=H$, requires every higher edge of
			each $\widehat{\mathcal E}$-primitive hexagon to belong to
			$\widehat{\mathcal E}$.}
	\end{enumerate}
\end{definition}

Condition $(i)$ is a monotonicity condition within each hexagon: if a lower
edge is included, then the corresponding higher edge with the same label
must also be included. Condition $(ii)$ is the global closure condition:
when an active pattern is primitive, it forces all higher edges in every
hexagon with the same associated set to be included.

The two-dimensional faces of the permutohedron are precisely squares and
hexagons.\footnote{This is a standard fact about the permutohedron; see, e.g.,
	\citet{morton2009convex}, who note that the
	two-dimensional faces of the permutohedron are squares and hexagons.} Thus, the following condition imposes consistency within and across these faces.

\begin{definition}\label{face_consistency}
	A set of edges $\widehat{\mathcal{E}}\subseteq \mathcal{E}$ satisfies
	\textit{face consistency} if it satisfies both square consistency
	and hexagon consistency.
\end{definition}

\begin{theorem}\label{Characterization_of_Identifiability}
	Let $\widehat{\mathcal{E}}\subseteq \mathcal{E}$ be a set of edges. Then,
	$\widehat{\mathcal{E}}$ is identifiable if and only if $\widehat{\mathcal{E}}$
	satisfies face consistency. Moreover, if $\widehat{\mathcal{E}}$ is
	identifiable, then there exists a unique irreducible experiment that identifies it.
\end{theorem}

The proof of this theorem is relegated to Appendix \ref{A}. Theorem \ref{Characterization_of_Identifiability} establishes that face consistency completely characterizes identifiability. Moreover, if $\widehat{\mathcal{E}}$ is identifiable, then there exists a unique irreducible experiment (say, $\mathcal M^*$) that identifies $\widehat{\mathcal{E}}$. The proof in Appendix \ref{A} provides an explicit construction:
$$
\begin{aligned}
	\mathcal M^*={}&
	\big\{\{r_{n-1}(P),r_n(P)\}:P\in\mathcal P \mbox{ and }\;
	\{P,P^{n-1}\}\in\widehat{\mathcal E}\big\}\\
	&{}\cup
	\big\{A(H):H\in\mathcal H
	\text{ is }\widehat{\mathcal E}\text{-primitive}\big\}.
\end{aligned}
$$
The first collection contains the binary menus forced by included edges that swap the bottom two positions. The second collection contains the associated sets of primitive hexagons. As already illustrated in Section \ref{Ex_illustrate_results}, $\mathcal M^*$ is the common indispensable core of all experiments that identify $\widehat{\mathcal{E}}$. Thus, every menu in $\mathcal M^*$ must appear in every experiment that identifies $\widehat{\mathcal{E}}$, while any additional menu present in an identifying experiment is redundant for identification of $\widehat{\mathcal{E}}$.

\section{A Characterization of Identifiable Type Spaces} \label{Char_Id_Typespace}
Having characterized identifiable set of edges of the permutohedron in the previous section, we now move on to characterize identifiable type spaces. As noted in Remark \ref{equiv}, characterizing identifiable
type spaces is equivalent to characterizing identifiable sets of edges of the
permutohedron, we now explicitly provide a characterization of identifiable type spaces. 

First, we introduce the notion of connected type spaces. We say that a type space $T$ is connected,  if, for every type
$t\in T$ and every pair $P,P'\in t$, there is a path in the
permutohedron between $P$ and $P'$ that is contained in $t$.\footnote{By a path in the permutohedron, we mean a path whose consecutive vertices are adjacent preferences, that is, a path in $G_{\emptyset}$.}

For a type space $T$, define its boundary in the permutohedron by $\partial T=\bigl\{\{P,P'\}\in \mathcal E : t(P)\neq t(P')\bigr\}$, that is, $\partial T$ is the set of edges of the permutohedron whose endpoints belong to
different types. Verbally speaking, $\partial T$ is the set of edges of the permutohedron that connects two distinct types of $T$. Since every edge of the permutohedron has the form
$\{P,P^i\}$ for some $P\in\mathcal P$ and some
$i\in\{1,\ldots,n-1\}$, we equivalently define 
	$$\partial T=\bigl\{\{P,P^i\} : P\in\mathcal P, i \in \{1,\ldots,n-1\} \mbox{ and } t(P)\neq t(P^i)\bigr\}.$$

We now spell out what it means for the boundary $\partial T$ to satisfy
face consistency. This is only a restatement of square and hexagon
consistency in terms of edges that belong to the boundary $\partial T$.	

First, $\partial T$ satisfies \textit{square consistency} if for every
$P\in\mathcal P$ and every pair $i,j\in\{1,\ldots,n-1\}$ with
$|i-j|>1, \{P,P^i\}\in\partial T$ \textit{if and only if} $\{P^j,P^{ji}\}\in\partial T$.
Thus two opposite edges of a square either both belong to the boundary or neither does. Note that this is a literal translation of Definition \ref{square_consistency}.

Next, we define \textit{hexagon consistency} for $\partial T$. We say that $\partial T$ satisfies \textit{hexagon consistency} if the following conditions hold:
\begin{enumerate}[(i)]
	\item For every
	$P\in\mathcal P$ and every $i\in\{2,\ldots,n-1\}$,
	$\{P,P^i\}\in\partial T
	\Longrightarrow
	\{P^{(i-1)i},P^{(i-1)i(i-1)}\}\in\partial T$.
	\item For
	every $\partial T$-primitive hexagon $H \in \mathcal H$ and every hexagon $K \in \mathcal H$ with $A(H)=A(K)$, every higher edge of $K$ belongs to $\partial T$.
\end{enumerate}
Note that Condition (i) and Condition (ii) of hexagon consistency stated here are literal translations of Condition (i) and Condition (ii) of hexagon consistency provided in Definition \ref{hexagon_consistency}. The fact that Condition (i) is a literal translation of Condition (i) in Definition \ref{hexagon_consistency} is because $\{P^{(i-1)i},P^{(i-1)i(i-1)}\}$ is the opposite higher edge in a hexagon containing
$\{P,P^i\}$ as a lower edge.

Finally, we say that $\partial T$ satisfies \emph{face consistency} if it satisfies both square consistency and hexagon consistency. It is worth noting that face consistency defined here is a literal translation of the face consistency definition provided in Definition \ref{face_consistency}.

At this point, one might be tempted to conclude that a type space $T$ is identifiable if and only if $\partial T$ satisfies face consistency. This conclusion is \textit{false}. Recall that, for a set of edges $\widehat {\mathcal E}$, by $\mathcal R_{\widehat {\mathcal E}}$ we denote the type space obtained from the permutohedron by deleting all edges in $\widehat {\mathcal E}$. By Theorem \ref{Characterization_of_Identifiability} and Observation \ref{Obs_1}, it follows that given a type space $T$, the type space $\mathcal R_{\partial T}$ is identifiable if and only if $\partial T$ satisfies face consistency. However, the issue is that $T$ need not coincide with $\mathcal R_{\partial T}$.

For example, let $X=\{x,y,z\}$ and consider the following two type spaces:
\begin{align*}
	\widehat T
	&=\{\{xyz\},\{xzy\},\{zxy\},\{zyx\},\{yzx\},\{yxz\}\}, \mbox{ and }\\
	\widebar T
	&=\{\{xyz,zyx\},\{xzy\},\{zxy\},\{yzx\},\{yxz\}\}.
\end{align*} Notice that $\partial {\widehat T}=\partial {\widebar T}=\mathcal E$, that is, the boundary of both $\widehat T$ and $\widebar T$ includes all the edges of the permutohedron. Clearly $\mathcal E$ satisfies face consistency and hence Theorem \ref{Characterization_of_Identifiability} implies that $\mathcal E$ is identifiable.\footnote{The unique irreducible experiment identifying $\mathcal E$ is $\{\{x,y\},\{y,z\},\{x,z\}\}$.} Also, $\mathcal R_{\mathcal E}=\{\{xyz\},\{xzy\},\{zxy\},\{zyx\},\{yzx\},\{yxz\}\}=\widehat T$. Therefore, $\widehat T$ is identifiable.

On the other hand, $\widebar T \neq \mathcal R_{\mathcal E}=\mathcal R_{\partial {\widebar T}}$. Equivalently, deleting edges in the boundary of $\widebar T$ does not recover the
types of $\widebar T$. In particular, $\widebar T$ is not connected, and hence not identifiable. The example shows that face consistency of the boundary of a type space must be combined with a
recoverability requirement. Connectedness supplies such a requirement, that is, if $T$ is
connected, then $T=\mathcal R_{\partial T}$. 

Moreover, note that Observation \ref{Obs_1} already establishes that every identifiable type space must be connected.\footnote{It is worth noting that the converse, that is, connected type spaces must be identifiable, is \textit{not} true. For example, consider the connected type space $\mathcal R_{\widehat{\mathcal E}}=\{\{xzy,zxy\},\{xyz,yxz,yzx,zyx\}\}$ in Example \ref{Ex_1} which is not identifiable.} Therefore, when translating our characterization condition of identifiable set of  edges back into type spaces, it is
natural to restrict our attention to connected type spaces. Hence, for connected type spaces, Theorem
\ref{Characterization_of_Identifiability} gives the following characterization as a corollary.

\begin{corollary}\label{corollary}
	Let $T$ be a connected type space. Then $T$ is identifiable if and only if
	$\partial T$ satisfies face consistency. Moreover, if $T$ is identifiable,
	then there exists a unique irreducible experiment that identifies it.
\end{corollary}

The proof of this corollary is relegated to Appendix \ref{C}.

\section{Applications}\label{sec:applications-acts-shortlists}

A planner conducting a product survey may wish to learn a subject's
favorite product or the set of products she prefers to an outside option,
without eliciting her complete preference ranking. The planner can ask
the subject to provide the desired information and offer a reward
contingent on her answer. This raises an incentive question: which
informational objectives admit a reward rule that makes truthful reporting
optimal while requesting only the information of interest? We study this
question in the framework of elicitation with acts developed by
\citet{appACH2021}. Corollary~\ref{corollary} provides a geometric
characterization of the type spaces that can be elicited in this
framework. We then apply this characterization to unranked shortlists of
acceptable products. Under some natural requirements, only two
policies are feasible: requesting the subject's favorite product whenever
one is acceptable, or requesting her entire acceptable set.

\subsection{A geometric characterization of type spaces elicitable with acts}
\label{subsec:applications-acts}

Fix a type space $T$ representing the planner's informational objective.
When the subject's preference is $P$, the planner wants to learn
$t(P)$ without obtaining additional information that distinguishes
preferences within that type. The subject is therefore asked to report
an element $t\in T$. For example, if types are defined by the subject's
favorite product, reporting a type amounts to naming that product. The
planner commits in advance to a reward rule under which the subject
receives an alternative in $X$ as a function of her report and the outcome
of a randomization device. All alternatives in $X$ are assumed available
as rewards.

Following \citet{appACH2021}, we describe random rewards by acts, without specifying probabilities for the states of the randomization device, thus allowing
the subject to have her own beliefs about the randomization device. An
act specifies the alternative received at each possible state of that
device. Formally, a \emph{$T$-mechanism with acts} is a pair $(\Omega,f)$,
where $\Omega$ is a finite nonempty state space and
$f:T\to X^\Omega$. A report $t$ is rewarded with the act $f(t)$:
if state $\omega\in\Omega$ occurs, the subject receives
$f(t)(\omega)\in X$. The subject submits her report before the state is
realized, and the randomization device is unaffected by her report.

\begin{definition} \label{IC_acts}

Following \citet{appACH2021} (Definitions 5 and 6), a $T$-mechanism with
acts $(\Omega,f)$ is \emph{incentive compatible} if, for every
$P\in\mathcal P$ and every $t'\in T\setminus\{t(P)\}$,
$$
	\begin{gathered}
		f(t(P))(\omega)=f(t')(\omega)
		\quad\text{or}\quad
		f(t(P))(\omega)\,P\,f(t')(\omega)
		\quad\text{for every }\omega\in\Omega ,\mbox{ and}\\
		f(t(P))(\omega)\,P\,f(t')(\omega)
		\quad\text{for some }\omega\in\Omega.
	\end{gathered}
$$
\end{definition}
Thus, truthful reporting gives a weakly preferred outcome in every
state and a strictly preferred outcome in at least one state.
This requirement compares outcomes state by state and does not
require probabilities on $\Omega$. In particular, truth-telling is uniquely optimal for any subject who maximizes subjective expected utility, provided that her utility function over $X$ represents $P$ and her beliefs on $\Omega$ have full support. The planner need not
know these beliefs or require subjects to agree on the probabilities of
the randomization states. A type space $T$ is \emph{elicitable with acts}
if there exists an incentive compatible $T$-mechanism with acts.

\citet{appACH2021} (Proposition 1) show that a type space is elicitable
with acts if and only if it is identifiable.\footnote{\citet{appACH2021} use the term
	``generated by top elements'', which has the same definition as identifiability
	in our terminology.}
Combining their result with Observation~\ref{Obs_1} and
Corollary \ref{corollary} immediately yields the following proposition.

\begin{proposition}
	\label{prop:applications-equivalence}
	Let $T$ be a type space. The following statements are equivalent:
	\begin{enumerate}
		\item[(i)] $T$ is elicitable with acts.
		\item[(ii)] $T$ is identifiable.
		\item[(iii)] $T$ is connected and $\partial T$ satisfies face consistency.
	\end{enumerate}
\end{proposition}

Proposition~\ref{prop:applications-equivalence} characterizes elicitation
with acts directly through the geometry of the desired type space.
Connectedness and face consistency determine whether the planner's
informational objective admits truthful elicitation with acts.
Theorem \ref{Characterization_of_Identifiability} also explicitly determines the unique irreducible
experiment identifying an elicitable with acts type space. This experiment is
contained in every identifying experiment and is therefore their
indispensable core.\footnote{\citet{appACH2021} (Proposition 2) also establish
	existence and uniqueness of this experiment, which they call a ``unique
	independent list of menus''. Their proof uses an anti-exchange closure operator to prove the existence, whereas
	our characterization determines this list explicitly from the boundary
	$\partial T$.}
Next, we apply Proposition \ref{prop:applications-equivalence} to policies that request an
unranked shortlist of acceptable products.

\subsection{Eliciting unranked shortlists}
\label{subsec:applications-shortlists}

Consider a planner who wants to elicit an unranked shortlist of a
subject's preferred products. Depending on the purpose of the survey,
the planner may ask for a single favorite product, all products preferred
to an outside option, or the best such products subject to a limit on the
number of products reported. The requested information is an unranked set
of products. To induce truthful reporting, the planner can make the
subject's reward depend on the reported set and the outcome of a
randomization device, as in the preceding subsection. We ask which
shortlist policies admit truthful elicitation with acts, allowing the
planner to learn exactly the prescribed shortlist.

Write $X=Y\cup\{o\}$, where $o\notin Y$ and $|Y|=m\geq2$. The elements
of $Y$ are products, and $o$ is a distinguished outside option, such as
receiving no product. All alternatives in $X$ are available as rewards, that is,
the elicitation mechanism may award a product or the outside option after
the subject submits her report. For a preference $P \in \mathcal P$, the \emph{acceptable set} at $P$ is
$$
\operatorname{Acc}(P)=\{x\in Y:xPo\}.
$$
A \emph{shortlist policy} is a map $S:\mathcal P\to2^Y$, where $2^Y$ is
the set of all subsets of $Y$. The planner specifies a shortlist policy $S$ in advance. A subject with preference $P$ is asked to report the set $S(P)$, without reporting the ordering of its elements. We consider policies satisfying the following natural properties:
\begin{enumerate}
	\item[(A1)] \emph{Acceptability:} $S(P)\subseteq \operatorname{Acc}(P)$ for every $P$.
	\item[(A2)] \emph{No skipped superior products:} if $x\in S(P)$ and
	$y\in Y$ satisfies $yPx$, then $y\in S(P)$.
	\item[(A3)] \emph{Non-emptiness:} if $\operatorname{Acc}(P)\neq\emptyset$, then
	$S(P)\neq\emptyset$.
	\item[(A4)] \emph{Symmetry:} for every permutation $\sigma$ of $Y$,
	extended to $X$ by $\sigma(o)=o$,
	$$
	S(\sigma P)=\sigma S(P).
	$$
	Here, $\sigma P$ is obtained by relabeling each product $x$ in $P$ as
	$\sigma(x)$, and $\sigma S(P)=\{\sigma(x):x\in S(P)\}$.
\end{enumerate}
Acceptability (A1) restricts the shortlist to products that the subject would
prefer compared to receiving the outside option. The requirement of no skipped
superior products (A2) expresses the objective of selecting the subject's most preferred
acceptable products: retaining a product requires retaining every product
she prefers to it. Together, (A1) and (A2) make the shortlist an initial
segment of the acceptable products, while allowing its length to vary.
Non-emptiness (A3) rules out an empty report when at least one acceptable
product exists. Symmetry (A4) requires the shortlist policy to depend on
the subject's ranking rather than on product names. The outside
option serves as a fixed acceptability threshold.\footnote{Condition (A4) leaves the outside option fixed. \citet{appACH2021} (Section 6, Proposition 9) characterize elicitation with acts for type spaces invariant under permutations of all alternatives in $X$; this additional invariance is not imposed here.}

Fix any shortlist policy $S$. Let $T_S$ denote the type space induced by $S$, in which two
preferences $P$ and $P'$ belong to the same type if and only if
$S(P)=S(P')$. Learning the subject's type in $T_S$ therefore determines
the prescribed shortlist. Distinguishing preferences within a type would
reveal additional information that is unnecessary for determining that
shortlist. The planner may wish to avoid acquiring this information to
protect the subject's privacy or to limit collection to the survey's
purpose. An experiment $\mathcal M$ meets this objective exactly when
$\mathcal R_{\mathcal M}=T_S$, that is, menu-choice outcomes of $\mathcal M$ determine the
shortlist, and any two preferences prescribing the same shortlist produce
the same menu-choice outcomes of $\mathcal M$. Identifiability of $T_S$ thus determines whether such
an experiment exists.

A $T_S$-mechanism with acts asks the subject to report a type
in $T_S$ and assigns an act to each report, as defined in the
preceding subsection. By Proposition~\ref{prop:applications-equivalence}, $T_S$ is elicitable with acts if and only if $T_S$ is
identifiable. Identifiability therefore also characterizes the shortlist
policies that admit truthful reporting with acts.

Properties (A1)-(A4) allow many shortlist policies. For example, a policy
may select up to two of the subject's most preferred acceptable products,
or vary the number of products selected with the number of acceptable products.
First, we define two simple shortlist policies, namely first-choice shortlist
policy and full-acceptability shortlist policy. Define the \emph{first-choice shortlist
	policy} $S^{\mathrm{first}}:\mathcal P\to2^Y$ by
\begin{equation*}\label{eq:applications-first-policy}
	S^{\mathrm{first}}(P)=
	\begin{cases}
		\{r_1(P)\},&r_1(P)\in Y,\\
		\emptyset,&r_1(P)=o,
	\end{cases}
\end{equation*}
and the \emph{full-acceptability shortlist policy} $S^{\mathrm{acc}}:\mathcal P\to2^Y$ by
\begin{equation*}\label{eq:applications-approval-policy}
	S^{\mathrm{acc}}(P)=\operatorname{Acc}(P).
\end{equation*}
The first-choice shortlist policy selects the subject's favorite product
whenever any product is acceptable, and the empty set otherwise. The
full-acceptability shortlist policy selects every acceptable product,
without their internal ordering. The next proposition shows that these are the only policies
satisfying (A1)--(A4) whose induced type spaces are identifiable.
By Proposition \ref{prop:applications-equivalence}, they are also the only such policies that
admit truthful elicitation with acts.

\begin{proposition}
	\label{prop:applications-shortlist-dichotomy}
	Suppose a shortlist policy $S$ satisfies \textnormal{(A1)-(A4)}. Then $T_S$ is identifiable
	if and only if $S=S^{\mathrm{first}}$ or $S=S^{\mathrm{acc}}$.
	Moreover, their unique irreducible identifying experiments are, respectively,
	\begin{equation*}\label{eq:applications-shortlist-cores}
		\mathcal M^*_{\mathrm{first}}=\{X\},
		\qquad
		\mathcal M^*_{\mathrm{acc}}=\{\{o,x\}:x\in Y\}.
	\end{equation*}
\end{proposition}

The proof of this proposition is relegated to Appendix \ref{D}.

Given a shortlist policy $S$, let
$\mathcal B_S=\{S(P):P\in\mathcal P\}$ be the set of possible shortlists.
Each type in $T_S$ corresponds to exactly one shortlist
$B\in\mathcal B_S$, so a $T_S$-mechanism with acts can equivalently
be written as a direct mechanism $(\Omega,f)$ with
$f:\mathcal B_S\to X^\Omega$. The subject reports a shortlist
$B\in \mathcal B_S$ and receives $f(B)(\omega)$ in state $\omega$.
In this representation, incentive compatibility is defined using Definition \ref{IC_acts}, with $t(P)$ replaced by $S(P)$
and $t'$ replaced by any $B\in\mathcal B_S\setminus\{S(P)\}$.
We now give incentive compatible direct mechanisms for the two
policies characterized in
Proposition~\ref{prop:applications-shortlist-dichotomy}.

For the first-choice shortlist policy,
the report space is
$$
\mathcal B_{\mathrm{first}}
:=\mathcal B_{S^{\mathrm{first}}}
=\{\emptyset\}\cup\{\{x\}:x\in Y\}.
$$ The first-choice mechanism is deterministic, represented by a
singleton state space $\Omega_{\mathrm{first}}=\{\omega_0\}$,
where $\omega_0$ denotes its sole state. Define
$f_{\mathrm{first}}:\mathcal B_{\mathrm{first}}
\to X^{\Omega_{\mathrm{first}}}$ by
\begin{equation*}\label{eq:applications-first-mechanism}
	f_{\mathrm{first}}(B)(\omega_0)=
	\begin{cases}
		x,&B=\{x\},\\
		o,&B=\emptyset.
	\end{cases}
\end{equation*}
If a product is preferred to $o$, the truthful report names the
subject's favorite product and awards that product. If no product
is acceptable, the truthful report is empty and awards $o$, her
favorite alternative in $X$. Every false report awards a strictly
worse alternative. Thus the mechanism is incentive compatible.

For the full-acceptability shortlist policy, the report space is
$\mathcal B_{\mathrm{acc}}:=\mathcal B_{S^{\mathrm{acc}}}=2^Y$. Take
$\Omega_{\mathrm{acc}}=Y$ and define
$f_{\mathrm{acc}}:\mathcal B_{\mathrm{acc}}
\to X^{\Omega_{\mathrm{acc}}}$ by
\begin{equation*}\label{eq:applications-acceptability-mechanism}
	f_{\mathrm{acc}}(B)(x)=
	\begin{cases}
		x,&x\in B,\\
		o,&x\notin B.
	\end{cases}
\end{equation*}
The state selects a product $x$. The subject receives that product
if she included it in her report, and receives $o$ otherwise.
Reporting $B=\operatorname{Acc}(P)$ selects her preferred alternative from $\{x,o\}$
in every state $x$. Any false report $B\neq \operatorname{Acc}(P)$ either omits an acceptable product or includes an unacceptable one. If $x\in \operatorname{Acc}(P)\setminus B$, reporting $B$ awards $o$ in state $x$, whereas truthful reporting awards the preferred product $x$. If $x\in B\setminus \operatorname{Acc}(P)$, reporting $B$ awards $x$ in state $x$, whereas truthful reporting awards the preferred outside option $o$. Thus every false report gives a strictly worse outcome in at least one state. It cannot improve the outcome in any state, because truthful reporting always awards the subject’s preferred alternative from $\{x,o\}$. This establishes incentive
compatibility without eliciting an ordering of the reported products.

Next, we examine an upper bound on the number of products reported. Fix
$1\leq k\leq m$ and define the shortlist policy
$S_k:\mathcal P\to2^Y$ by
$$
	S_k(P)=
	\{r_j(P):1\leq j\leq \min\{k,|\operatorname{Acc}(P)|\}\}, \mbox{ where } P \in \mathcal P
$$
such that $S_k(P)=\emptyset$ when $\operatorname{Acc}(P)=\emptyset$. Thus the subject reports
her best $k$ acceptable products if at least $k$ exist, and all acceptable
products otherwise.

\begin{corollary}
	\label{cor:applications-capped-lists} 
	The type space $T_{S_k}$ is identifiable if and only if $k=1$ or $k=m$.
\end{corollary}

\begin{proof}
	For each $1 \leq k \leq m$, the policy $S_k$ satisfies (A1)-(A4). For $k=1$,
	$S_1=S^{\mathrm{first}}$, and for $k=m$, $S_m=S^{\mathrm{acc}}$. If $1<k<m$, take a preference $\tilde P$ with all $m$ products acceptable. Then $|S_k(\tilde P)|=k$, whereas $|S^{\mathrm{first}}(\tilde P)|=1$ and $|S^{\mathrm{acc}}(\tilde P)|=m$. The policy $S_k$ is therefore neither of those characterized in
	Proposition~\ref{prop:applications-shortlist-dichotomy} and hence $T_{S_k}$ is not identifiable. Therefore, the type space $T_{S_k}$ is identifiable if and only if $k=1$ or $k=m$.
\end{proof}

\begin{figure}[htbp]
	\centering
	\begin{tikzpicture}[
		scale=2.0,
		pref/.style={align=center,font=\small,inner sep=4pt},
		boundary/.style={draw=red!70!black,line width=1.1pt},
		internal/.style={draw=black!65,line width=0.65pt}
		]
		\node[pref] (abc) at (90:1) {$abco$\\$\{a,b\}$};
		\node[pref] (acb) at (30:1) {$acbo$\\$\{a,c\}$};
		\node[pref] (cab) at (-30:1) {$cabo$\\$\{a,c\}$};
		\node[pref] (cba) at (-90:1) {$cbao$\\$\{b,c\}$};
		\node[pref] (bca) at (-150:1) {$bcao$\\$\{b,c\}$};
		\node[pref] (bac) at (150:1) {$baco$\\$\{a,b\}$};
		\draw[boundary] (abc)--node[above right,font=\small,text=red!70!black] {$\ell$}(acb);
		\draw[internal] (acb)--node[right,font=\small] {$h$}(cab);
		\draw[boundary] (cab)--node[below right,font=\small,text=red!70!black] {$\ell$}(cba);
		\draw[internal] (cba)--node[below left,font=\small] {$h$}(bca);
		\draw[boundary] (bca)--node[left,font=\small,text=red!70!black] {$\ell$}(bac);
		\draw[internal] (bac)--node[above left,font=\small] {$h$}(abc);
	\end{tikzpicture}
	\caption{A shortlist of at most two products with $Y=\{a,b,c\}$.
		Each vertex displays a preference and its shortlist. Only red edges
		belong to $\partial T_{S_2}$.}
	\label{fig:applications-capped-hexagon}
\end{figure}
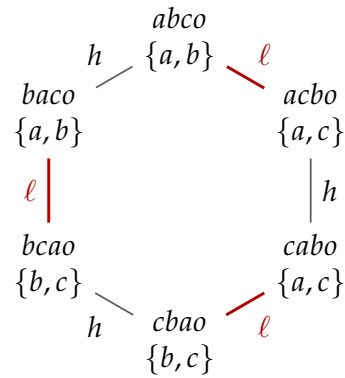

Figure \ref{fig:applications-capped-hexagon} illustrates why $T_{S_2}$
is not identifiable when $Y=\{a,b,c\}$. Throughout the displayed hexagon of the permutohedron,
all three products are above $o$ and are therefore acceptable at every vertex. Exchanging
the first two products leaves their reported set unchanged, whereas exchanging
the second and third changes it. Consequently, all lower edges belong to
$\partial T_{S_2}$, whereas all higher edges lie outside $\partial T_{S_2}$. For example,
the lower edge $\{abco,acbo\}$ has different reports, $\{a,b\}$ and
$\{a,c\}$, at its endpoints, but its opposite higher edge
$\{bcao,cbao\}$ has the same report $\{b,c\}$ at both endpoints. Hence, the lower edge $\{abco,acbo\} \in \partial T_{S_2}$, but its opposite higher edge $\{bcao,cbao\} \notin \partial T_{S_2}$. Therefore, $\partial T_{S_2}$
violates hexagon consistency and hence violates face consistency. Therefore, Proposition \ref{prop:applications-equivalence} implies that $T_{S_2}$ is not identifiable.

	The two identifiable type spaces correspond to distinct survey objectives.
	First-choice reporting determines the subject's favorite alternative in
	$X$. Full-acceptability reporting determines which products she prefers
	to $o$, leaving the ordering within both the acceptable and unacceptable
	sets undisclosed. For each objective, the irreducible experiment in
	Proposition \ref{prop:applications-shortlist-dichotomy} induces exactly
	the desired type space, and the associated direct mechanism elicits the
	prescribed set without requesting further preference information.
	If $1<k<m$, then $S_k$ also defines a shortlist policy satisfying
	(A1)-(A4), but its type space $T_{S_k}$ is not identifiable. Equivalently,
	no incentive compatible mechanism with acts can elicit exactly that
	shortlist.\footnote{The incentive conclusion here concerns elicitation with acts, as defined in Definition \ref{IC_acts}. Elicitation with objective lotteries
		allows a larger class of type spaces; see Section 5 in \citet{appACH2021}.} Proposition \ref{prop:applications-shortlist-dichotomy} characterizes the two feasible policies and their
	indispensable experiments. Corollary \ref{cor:applications-capped-lists} rules out every
	intermediate cap on the number of reported products.

%\section{Conclusion}

	\begin{appendices}
		\section{Appendix}
		
		\subsection{Proof of Observation \ref{Obs_1}} \label{B}
		\vspace{2mm}
		
		\noindent
\begin{proof} Fix an experiment $\mathcal M$. Pick any type $t'$ in $\mathcal R_{\widehat{\mathcal{E}}_{\mathcal{M}}}$. Let $P,P' \in t'$. By the definition of the type space $\mathcal R_{\widehat{\mathcal{E}}_{\mathcal{M}}}$, it follows that there exists a path $(\tilde P_1=P,\ldots,\tilde P_s=P')$ in $G_{\widehat{\mathcal E}_{\mathcal M}}$ between $P$ and $P'$. By construction, no edge of this path is distinguished by $\mathcal M$. Hence, for every $M\in\mathcal M$, $\mbox{\textit{dom}}_{\tilde P_i}(M)=\mbox{\textit{dom}}_{\tilde P_{i+1}}(M)$ for all $i=1,\ldots,s-1$. This implies that $\mbox{\textit{dom}}_{P}(M)=\mbox{\textit{dom}}_{P'}(M)$
	for every $M\in\mathcal M$. Therefore, $P$ and $P'$ belong to the same type in $\mathcal R_{\mathcal M}$. Since $P,P' \in t'$ were arbitrary preferences in $t'$, it follows that there exists a type $t \in \mathcal R_{\mathcal M}$ such that $t' \subseteq t$. Since $t' \in \mathcal R_{\widehat{\mathcal{E}}_{\mathcal{M}}}$ was arbitrary, it follows that for every type $t' \in \mathcal R_{\widehat{\mathcal{E}}_{\mathcal{M}}}$, there exists a type $t \in \mathcal R_{\mathcal M}$ such that $t' \subseteq t$.

	Now, fix any pair of types $t' \in \mathcal R_{\widehat{\mathcal{E}}_{\mathcal{M}}}$ and $t \in \mathcal R_{\mathcal M}$ such that $t' \subseteq t$. We will show that $t'=t$, which will conclude that $\mathcal R_{\mathcal{M}}=\mathcal R_{\widehat{\mathcal{E}}_{\mathcal{M}}}$. To show that $t=t'$, it is enough to show that $t \subseteq t'$. 
	
	First, we show that every pair of preferences in $t$ is connected in $G_{\widehat{\mathcal E}_{\mathcal M}}$. Pick any $P,P' \in t$. Since $t \in \mathcal R_{\mathcal M}$, it follows that $\mbox{\textit{dom}}_{P}(M)=\mbox{\textit{dom}}_{P'}(M)$ for every $M \in \mathcal M$. Fix any menu $M\in\mathcal M$ and let $x=\mbox{\textit{dom}}_P(M)=\mbox{\textit{dom}}_{P'}(M)$.
	Then, for every $y\in M\setminus\{x\}$, $xPy$ and $xP'y$.
	Thus, $\{x,y\}\notin P\vartriangle P'$
	for every $y\in M\setminus\{x\}$. Consider any $\tilde P \in \mathcal P$ such that $P\vartriangle \tilde P \subseteq P\vartriangle P'$. Since $P\vartriangle \tilde P \subseteq P\vartriangle P'$ and $\{x,y\}\notin P\vartriangle P'$
	for every $y\in M\setminus\{x\}$, it follows that $\{x,y\}\notin P\vartriangle \tilde P$
	for every $y\in M\setminus\{x\}$. This, together with the fact that $x=\mbox{\textit{dom}}_P(M)$ implies $x=\mbox{\textit{dom}}_{\tilde P}(M)$. Since $M \in \mathcal M$ is arbitrary, it follows that $\tilde P \in t$ whenever $P\vartriangle \tilde P \subseteq P\vartriangle P'$. Consider a shortest path $\pi=(P_1=P,\ldots,P_s=P')$ between $P$ and $P'$ in $G_{\emptyset}$.\footnote{\label{shortest-path} A shortest path minimizes the number of edges among
		paths with the same endpoints. The proof of Proposition~4.1
		in \citet{sato2013sufficient} constructs a path without restoration between
		any two preferences. In the permutohedron, such paths are
		precisely the shortest paths: they reverse each pair ranked
		differently at the endpoints exactly once and no other pair.
		Consequently, every pairwise comparison on which the endpoints
		agree is preserved throughout the path; see also Remark~4.1
		in \citet{sato2013sufficient}.} Along the path $\pi$, every adjacent swap reverses a pair in $P\vartriangle P'$, and hence $P\vartriangle P_i\subseteq P\vartriangle P'$ for every $i=1,\ldots,s$. Hence, by the preceding argument, $P_i\in t$ for every $i=1,\ldots,s$. Therefore, no edge of the path $\pi$ is distinguished by $\mathcal M$, and consequently $\pi$ is a path between $P$ and $P'$ in $G_{\widehat{\mathcal E}_{\mathcal M}}$. Since $P,P' \in t$ were arbitrary preferences in $t$, it follows that every pair of preferences in $t$ is connected in $G_{\widehat{\mathcal E}_{\mathcal M}}$.

	Since every pair of preferences in $t$ is connected in $G_{\widehat{\mathcal E}_{\mathcal M}}$, all preferences in $t$ belong to the same connected component of $G_{\widehat{\mathcal E}_{\mathcal M}}$. This implies that all preferences in $t$ belong to the same type in $\mathcal R_{\widehat{\mathcal{E}}_{\mathcal{M}}}$. This, together with the fact that $t'$ is a type in $\mathcal R_{\widehat{\mathcal{E}}_{\mathcal{M}}}$ and $t' \subseteq t$, it follows that $t \subseteq t'$. This completes the proof of the observation.
\end{proof}

		\subsection{Proof of Theorem \ref{Characterization_of_Identifiability}} \label{A}
		\vspace{2mm}
		
		\noindent
\begin{proof} Let $\widehat{\mathcal{E}}\subseteq \mathcal{E}$ be a set of edges. First, we prove the sufficiency part of the theorem.
	
	\textbf{Sufficiency:}
	Let $\widehat{\mathcal E}\subseteq\mathcal E$ be a set of edges satisfying
	face consistency. If $\widehat{\mathcal E}=\emptyset$, then the empty
	experiment identifies $\widehat{\mathcal E}$. Since the empty experiment
	contains no menus, it is vacuously irreducible. Moreover, any nonempty experiment
	contains some menu $M$, and every menu distinguishes at least one edge of
	the permutohedron. Hence no nonempty experiment identifies the empty edge
	set. Thus the empty experiment is the unique irreducible experiment that
	identifies $\widehat{\mathcal E}$. For the remainder of the proof, assume
	$\widehat{\mathcal E}\neq\emptyset$.
	
	For an edge $e=\{P,P'\}\in\mathcal E$, let $\lambda(e)$ be the
	unique element of $P\triangle P'$. Thus,
	$P\triangle P'=\{\lambda(e)\}$, and $\lambda(e)$ is the
	two-element subset of $X$ reversed along the edge $e$. Define $U(e)=\{z\in X\setminus \lambda(e): zPw \text{ for all } w \in \lambda(e)\}$. Note that since $P$ and $P'$ differ only in the relative ranking of $x$ and $y$, where $x,y \in \lambda(e)$, the set $U(e)$ is independent of the
	choice of $P$ or $P'$. Finally, define $A(e)=X \setminus U(e)$.\footnote{We slightly abuse notation: the notation $A(\cdot)$ is also used for hexagons. The two definitions are analogous, however. In both cases, $A(\cdot)$ denotes the relevant set of alternatives: the changing pair for edges, or the changing triple for hexagons, together with all alternatives ranked below them.} Thus $A(e)$ consists of the reversed pair $\lambda(e)$ together with all
	alternatives ranked below both alternatives in $\lambda(e)$.
	
	First, we prove a simple and elementary fact that will be used repeatedly in the proof.
	
	\begin{lemma}\label{lem:menu-edge}
		Fix any menu $M$ and any edge $e\in\mathcal E$. The menu $M$
		distinguishes the edge $e$ if and only if $\lambda(e)\subseteq M\subseteq A(e)$.
	\end{lemma}
	
	\begin{proof}
		Let $e=\{P,P'\}$, and let $\lambda(e)=\{x,y\}$ where $xPy$ and $yP'x$. Therefore, $P$ and $P'$ differ only in the relative ranking of $x$ and $y$. Suppose the menu $M$ distinguishes the edge $e$. Since $\lambda(e)=\{x,y\}$, this implies that both $x$ and $y$ must belong to $M$. Therefore, $\lambda(e)\subseteq M$. Further, if $M \cap U(e) \neq \emptyset$, then it must be the case that $\mbox{\textit{dom}}_P(M)=\mbox{\textit{dom}}_{P'}(M)$, which is a contradiction to our assumption that the menu $M$ distinguishes the edge $e$. Hence, it must be the case that $M \cap U(e)= \emptyset$. Therefore, $M\subseteq A(e)$. Consequently, $\lambda(e)\subseteq M\subseteq A(e)$.

		Conversely, suppose that $\lambda(e)\subseteq M\subseteq A(e)$. Then $M \cap U(e)= \emptyset$. Therefore, $\mbox{\textit{dom}}_P(M)=x$ and $\mbox{\textit{dom}}_{P'}(M)=y$. Hence, $M$ distinguishes $e$. This completes the proof of the lemma.
	\end{proof}
	
	Next, we prove an important lemma that will be used repeatedly in the proof.

	\begin{lemma}\label{lem:monotonicity}
		Let $e,f\in\mathcal E$ be two edges such that $\lambda(e)=\lambda(f)$ and $A(e)\subseteq A(f)$. Then, $e\in\widehat{\mathcal E} \implies f\in\widehat{\mathcal E}$.
	\end{lemma}
	
	\begin{proof}
		Let $e,f\in\mathcal E$ be two edges such that $\lambda(e)=\lambda(f)$. First, we prove the lemma for the special case in which $A(e)=A(f)$. 
		
		\begin{claim} \label{A(e)=A(f)}
			If $A(e)= A(f)$ then $e\in\widehat{\mathcal E} \implies f\in\widehat{\mathcal E}$.
		\end{claim}
		
		\begin{proof}
			Suppose $e\in\widehat{\mathcal E}$. Let $\lambda(e)=\lambda(f)=\{x,y\}$. Since $A(e)=A(f)$, we also have
			$U(e)=U(f)$. Let $i=|U(e)|+1=|U(f)|+1$. Then, it must be the case that $e=\{P,P^i\}$ and $f=\{\widetilde P,\widetilde P^i\}$ for some preferences $P, \widetilde P \in \mathcal{P}$, where $r_i(P)=r_i(\widetilde P)=x$ and $r_{i+1}(P)=r_{i+1}(\widetilde P)=y$. Since $U(e)=U(f)$, $A(e)=A(f)$ and $\lambda(e)=\lambda(f)$, it must be the case that $\{r_1(P),\ldots,r_{i-1}(P)\}=\{r_1(\widetilde P),\ldots,r_{i-1}(\widetilde P)\}$ and $\{r_{i+2}(P),\ldots,r_n(P)\}=\{r_{i+2}(\widetilde P),\ldots,r_n(\widetilde P)\}$.
			Thus $P$ and $\widetilde P$ have the same set of alternatives ranked above $x$ and $y$, the same set of alternatives ranked below $x$ and $y$, and
			the same rank of $x$ and $y$ (at ranks $i$ and $i+1$, respectively). They may differ only
			in the ranking of alternatives within $U(e)$ (equivalently $U(f)$) and within $A(e) \setminus \lambda(e)$ (equivalently $A(f) \setminus \lambda(f)$).

			If $P=\widetilde P$, then there is nothing to
			prove. So, assume that $P\neq \widetilde P$. Therefore, $P\vartriangle \widetilde P \neq \emptyset$. Define 
			$$ k=\min\{j:r_j(P)\neq r_j(\widetilde P)\}.$$
			Since $P\neq \widetilde P$, it must be the case that $k \geq 1$. Since $P$ and $\widetilde P$ differ only
			in the ranking of alternatives within $U(e)$ (equivalently $U(f)$) and within $A(e) \setminus \lambda(e)$ (equivalently $A(f) \setminus \lambda(f)$), it must be the case that either $k \leq i-1$ or $k \geq i+2$.
			
			Let $z=r_k(\widetilde P)$, and let $a$ be the rank of $z$ at $P$. Since
			$P$ and $\widetilde P$ agree at all ranks above $k$, it must be the case that $a>k$. Moreover, since $A(e)=A(f)$ (equivalently $U(e)=U(f)$), it follows that if $k\leq i-1$, then $a \leq i-1$, and if $k\geq i+2$, then $a\geq i+2$. This, together with the fact that $a>k$ implies that $|a-1-i|>1$. Consider the square $Q=(P, P^i, P^{i(a-1)}, P^{a-1}, P)$. The two edges $\{P,P^i\}$ and $\{P^{a-1},P^{i(a-1)}\}$ are opposite edges of the square $Q$. Since $\{P,P^i\}\in \widehat{\mathcal E}$ and $\widehat{\mathcal E}$ satisfies square consistency, it follows that $\{P^{a-1},P^{i(a-1)}\}\in\widehat{\mathcal E}$. Also, $P^{i(a-1)}=P^{(a-1)i}$ because adjacent swaps in squares commute. Therefore, $\{P^{a-1},P^{(a-1)i}\} \in \widehat{\mathcal E}$. It is worth noting here that by the choice of $P^{a-1}$, it follows that $A(e)=A(\{P^{a-1},P^{(a-1)i}\})=A(f)$ (equivalently $U(e)=U(\{P^{a-1},P^{(a-1)i}\})=U(f)$) and $\lambda(e)=\lambda(\{P^{a-1},P^{(a-1)i}\})=\lambda(f)$.
			
			If $P^{a-1}=\widetilde P$, then the proof is complete because it would imply $f=\{\widetilde P, \widetilde P^i\} \in \widehat{\mathcal E}$.
			
			Suppose $P^{a-1}\neq \widetilde P$. Then, it must be the case that $|P^{a-1} \vartriangle \widetilde P|= |P\vartriangle \widetilde P|-1>0$. This is because $\{r_{a-1}(P),r_a(P)=z\} \in P\vartriangle \widetilde P$, but $\{r_{a-1}(P),r_a(P)=z\} \notin P^{a-1} \vartriangle \widetilde P$.
			
			As before, we analogously define 
			$$l=\min\{j:r_j(P^{a-1})\neq r_j(\widetilde P)\}.$$
			
			Using similar arguments as before we get hold of a preference $P^{(a-1)(b-1)}$, where $b$ is the rank of $r_l(\widetilde P)$ at $P^{a-1}$, such that 
			\begin{enumerate}[(i)]
				\item $\{P^{(a-1)(b-1)},P^{(a-1)(b-1)i}\} \in \widehat{\mathcal E}$,
				\item $A(e)=A(\{P^{(a-1)(b-1)},P^{(a-1)(b-1)i}\})=A(f)$ \\(equivalently $U(e)=U(\{P^{(a-1)(b-1)},P^{(a-1)(b-1)i}\})=U(f)$), and 
				\item $\lambda(e)=\lambda(\{P^{(a-1)(b-1)},P^{(a-1)(b-1)i}\})=\lambda(f)$.
				
			\end{enumerate}
			
			If $P^{(a-1)(b-1)}=\widetilde P$, then $f=\{P^{(a-1)(b-1)},P^{(a-1)(b-1)i}\}$ and we are done.
			
			Suppose $P^{(a-1)(b-1)} \neq \widetilde P$. Then, it must be the case that $|P^{(a-1)(b-1)} \vartriangle \widetilde P|=|P^{a-1} \vartriangle \widetilde P|-1= |P\vartriangle \widetilde P|-2>0$. This is because $\{r_{b-1}(P^{a-1}),r_b(P^{a-1})\} \in P^{a-1} \vartriangle \widetilde P$, but $\{r_{b-1}(P^{a-1}),r_b(P^{a-1})\} \notin P^{(a-1)(b-1)} \vartriangle \widetilde P$.
			
			Since $|P \vartriangle \widetilde P|$ is finite, repeating similar arguments as before finitely many times, it follows that $f=\{\widetilde P,\widetilde P^i\}\in\widehat{\mathcal E}$. This completes the proof of the claim.

		\end{proof}
		
		Now, we prove the lemma. Suppose $A(e) \subseteq A(f)$ and let $\lambda(e)=\lambda(f)=\{x,y\}$. Further, assume that $e \in \widehat{\mathcal E}$. We will show that $f \in \widehat{\mathcal E}$. We distinguish two cases:
		
		\textbf{Case (i):} Suppose $|A(f) \setminus A(e)|=0$. This implies $A(f)=A(e)$. Therefore, by Claim \ref{A(e)=A(f)}, it follows that $f \in \widehat{\mathcal E}$, thereby completing the proof of the lemma.
		
		\textbf{Case(ii):} Suppose $|A(f) \setminus A(e)|\geq 1$. Let $D=A(f) \setminus A(e)=\{z_1,\ldots,z_k\}$, where $k \geq 1$. We prove the following claim, which will complete the proof of the lemma.

		\begin{claim}\label{|A(f)|=|A(e)|+k}
			If $A(f)=A(e)\cup D$, then $e\in\widehat{\mathcal E} \implies f\in\widehat{\mathcal E}$.
		\end{claim}
		
		\begin{proof}
			
			Suppose $e=\{P, P^i\} \in \widehat{\mathcal E}$, for some $i \in \{k+1,\ldots,n-1\}$ and let $\lambda(e)=\{x,y\}$, where $r_i(P)=x$. Since $z \notin A(e)$ for each $z \in D$, it must be the case that $z \in U(e)$ for each $z \in D$. 
			
			First, consider $z_1 \in D$. Without loss of generality, assume that $wPz_1$ for all $w \in U(e) \setminus\{z_1\}$, that is, rank of $z_1$ at $P$ is $i-1$. This is because if the rank of $z_1$ at $P$ is $r \leq i-2$, then consider the edge $e'=\{\widebar P, \widebar P^i\}$ where $\widebar P=P^{r(r+1)\cdots(i-2)}$. Notice that $A(e)=A(e')$ and $\lambda(e)=\lambda(e')$. Therefore, by Claim \ref{A(e)=A(f)}, it follows that $e' \in \widehat{\mathcal E}$. The rank of $z_1$ at $\widebar P$ is $i-1$ and hence our assumption is without loss of generality.
			
			Consider the hexagon $H$ where $C(H)=\{x,y,z_1\}$ and $A(H)=A(e) \cup \{z_1\}$. Notice that $e$ is a lower ($\ell$) edge of $H$. Let $e_1=e^h=\{P_1,P_1^{i-1}\}$, where $e^h$ is the opposite higher edge of the lower edge $e$. This implies that $\lambda(e_1)=\lambda(e)=\{x,y\}$ and $A(e_1)=A(e) \cup \{z_1\}$. Since $e \in \widehat{\mathcal E}$ and $\widehat{\mathcal E}$ satisfies hexagon consistency, by condition (i) of hexagon consistency, it follows that $e_1 \in \widehat{\mathcal E}$.  
			
			Next, consider $z_2 \in D$. Once again, using similar arguments as before, without loss of generality assume that $wP_1z_2$ for all $w \in U(e_1) \setminus\{z_2\}$. Consider the hexagon $J$ such that $C(J)=\{x,y,z_2\}$ and $A(J)=A(e_1) \cup \{z_2\}$. Notice that $e_1$ is a lower ($\ell$) edge of the hexagon $J$. Let $e_2=e_1^h=\{P_2,P_2^{i-2}\}$, where $e_1^h$ is the opposite higher edge of the lower edge $e_1$. This implies that $\lambda(e_2)=\lambda(e_1)=\lambda(e)=\{x,y\}$and $A(e_2)=A(e_1) \cup \{z_2\}=A(e) \cup \{z_1,z_2\}$. Since $e_1 \in \widehat{\mathcal E}$ and $\widehat{\mathcal E}$ satisfies hexagon consistency, by condition (i) of hexagon consistency, it follows that $e_2 \in \widehat{\mathcal E}$. 
			
			Repeating similar arguments as before, we conclude that there exists an edge $e_k \in \widehat{\mathcal E}$ such that $\lambda(e_k)=\lambda(e)=\{x,y\}$ and $A(e_k)=A(e) \cup \{z_1,\ldots,z_k\}=A(e) \cup D$. Notice that $A(e_k)=A(e) \cup D=A(f)$ and $\lambda(f)=\lambda(e)=\lambda(e_k)=\{x,y\}$. Since $e_k \in \widehat{\mathcal E}$, by Claim \ref{A(e)=A(f)}, it follows that $f \in \widehat{\mathcal E}$. This completes the proof of the claim.
		\end{proof}
		
		This completes the proof of the lemma.
	\end{proof}
	
	Now, we prove the sufficiency part of the theorem. In what follows, we construct an irreducible experiment $\mathcal M^*$ that identifies $\widehat{\mathcal E}$.
	
	Define $\mathcal M^*_2=\{\lambda(e) : e \in \widehat{\mathcal E} \mbox{ and } A(e)=\lambda(e)\}$. Thus $\mathcal M^*_2$ is the collection of binary menus that are forced by edges in $\widehat{\mathcal E}$ for which $A(e)=\lambda(e)$. For such an edge, Lemma \ref{lem:menu-edge} implies that the only menu capable of
	distinguishing it is $\lambda(e)$ itself.
	
	For $m \in \{3,\ldots,n\}$, define $\mathcal M^*_m = \{A(H): H \in \mathcal H \mbox{ is } \widehat{\mathcal E} \mbox{-primitive}
	\mbox{ and } |A(H)|=m\}$. We index these collections by $m$ to emphasize the level-by-level structure
	of the construction. Hexagons whose changing triple occupies ranks
	$n-2, n-1$ and $n$ are the only hexagons that can generate menus of size
	$3$, whereas hexagons whose changing triple occupies ranks $n-3, n-2$ and $n-1$
	are the only hexagons that can generate menus of size $4$, and so on. This
	is precisely the order encoded by primitivity: a primitive hexagon is an
	active hexagon that is not already accounted for by an active hexagon with
	the same changing triple and a smaller associated set $A(\cdot)$. Thus $\mathcal M^*_m$
	collects the candidate menus of size $m$ generated at the corresponding
	level of the construction.
	
	Finally, define
	$$ \mathcal M^*= \mathcal M^*_2 \cup \mathcal M^*_3 \cup \cdots \cup \mathcal M^*_n.$$
	We will show that the experiment $\mathcal M^*$ identifies
	$\widehat{\mathcal E}$, thereby completing the proof of the sufficiency part of the theorem.
	
	\begin{claim}
		$\mathcal M^*$ identifies $\widehat{\mathcal E}$.
	\end{claim}
	
	\begin{proof} First, we show that $\mathcal M^*$ identifies $\widehat{\mathcal E}$,that is, $\widehat{\mathcal E}_{\mathcal M^*}=\widehat{\mathcal E}$, by proving $\widehat{\mathcal E}_{\mathcal M^*}\subseteq \widehat{\mathcal E}$ and $\widehat{\mathcal E} \subseteq \widehat{\mathcal E}_{\mathcal M^*}$. Let us first prove that $\widehat{\mathcal E}_{\mathcal M^*}\subseteq
		\widehat{\mathcal E}$. Let $f \in \widehat{\mathcal E}_{\mathcal M^*}$. Then, there exists $M\in\mathcal M^*$ such that $f$ is distinguished by $M$. We show that $f \in \widehat{\mathcal E}$.
		
		Suppose that $M=\{x,y\}\in\mathcal M^*_2$. By definition of
		$\mathcal M^*_2$, there exists an edge $e\in\widehat{\mathcal E}$ such that
		$\lambda(e)=\{x,y\}$ and $A(e)=\{x,y\}$. Since $|M|=2$ and $M$ distinguishes $f$, by
		Lemma \ref{lem:menu-edge}, we have
		$\lambda(f)=\{x,y\}\subseteq A(f)$. Hence, $A(e)\subseteq A(f)$, and Lemma
		\ref{lem:monotonicity} implies $f\in\widehat{\mathcal E}$.
		
		Now, suppose that $M \in \mathcal M^*_3 \cup \cdots \cup \mathcal M^*_n$. Therefore, $M=A(H)$ for some primitive hexagon $H$. Notice that $|M| \geq 3$. By Lemma
		\ref{lem:menu-edge}, $\lambda(f)\subseteq M\subseteq A(f)$.
		Choose an edge $g\in\mathcal E$ such that $\lambda(g)=\lambda(f)$
		and $A(g)=M$.
		Since $|M|\geq 3$, the edge $g$ is a higher edge of some hexagon
		$K\in\mathcal H$ with $A(K)=M=A(H)$. By Condition $(ii)$ of hexagon
		consistency, every higher edge of every hexagon $K$ with $A(K)=A(H)$
		belongs to $\widehat{\mathcal E}$. Hence $g\in\widehat{\mathcal E}$.
		Since $A(g)=M\subseteq A(f)$, Lemma \ref{lem:monotonicity} implies
		$f\in\widehat{\mathcal E}$. Therefore, we conclude that $\widehat{\mathcal E}_{\mathcal M^*}\subseteq
		\widehat{\mathcal E}$, that is, no menu in $\mathcal M^*$ distinguishes
		an edge outside $\widehat{\mathcal E}$.
		
		Now, we prove that $\widehat{\mathcal E} \subseteq \widehat{\mathcal E}_{\mathcal M^*}$. Let $f \in\widehat{\mathcal E}$. We proceed by induction on
		$|A(f)|$.
		
		Suppose that $|A(f)|=2$. Then
		$A(f)=\lambda(f)$, and by Lemma \ref{lem:menu-edge}, the menu $A(f)$ distinguishes $f$. Note that by construction of $\mathcal M^*$, the binary menu $A(f)$ belongs to
		$\mathcal M^*_2$. Therefore, $f \in \widehat{\mathcal E}_{\mathcal M^*}$.
		
		Fix $m\geq 3$ and suppose that every edge $e\in\widehat{\mathcal E}$
		with $|A(e)|<m$ is distinguished by some menu in $\mathcal M^*$. Let $|A(f)|=m$. Suppose $f=\{P,P^{n-m+1}\}$ for some $P \in \mathcal P$ and let $\lambda(f)=\{x,y\}$. We will show that $f \in \widehat{\mathcal E}_{\mathcal M^*}$, thereby establishing that $\widehat{\mathcal E} \subseteq \widehat{\mathcal E}_{\mathcal M^*}$.
		
		Since $m \geq 3$, it follows that $n-m+1 \leq n-2$. Define $z=r_{n-m+3}(P)$. Consider the hexagon $$H=(P,P^{n-m+1},P^{(n-m+1){(n-m+2)}},P^{(n-m+1){(n-m+2)}(n-m+1)},P^{{(n-m+2)}(n-m+1)},P^{(n-m+2)},P).$$ Note that $f$ is a higher edge of the hexagon $H$ and $C(H)=\{x,y,z\}$. Therefore, $A(H)=A(f)$ because $f$ is a higher edge of $H$. Let $f^{\ell}$ denote the lower edge of $H$ opposite to $f$. Therefore, $\lambda(f^{\ell})=\lambda(f)$ and $A(f^{\ell})=A(f) \setminus \{z\}$. 
		
		If $f^{\ell} \in \widehat{\mathcal E}$, then by the induction hypothesis $f^{\ell}$ is distinguished by some menu $M \in \mathcal M^*$. Therefore, using Lemma \ref{lem:menu-edge}, it follows that $\lambda(f^{\ell}) \subseteq M \subseteq A(f^{\ell})$. This, together with the facts that $\lambda(f^{\ell})=\lambda(f)$ and $A(f^{\ell})=A(f) \setminus \{z\}$ implies $\lambda(f^{\ell})=\lambda(f) \subseteq M \subseteq A(f^{\ell}) \subseteq A(f)$. By Lemma \ref{lem:menu-edge}, it follows that the menu $M$ distinguishes the edge $f$. Hence, $f \in \widehat{\mathcal E}_{\mathcal M^*}$.
		
		Suppose $f^{\ell} \notin \widehat{\mathcal E}$. Since $f \in \widehat{\mathcal E}$ and $f$ is the opposite higher edge of the lower edge $f^{\ell}$, the hexagon $H$ is $\widehat{\mathcal E}$-active. If $H$ is $\widehat{\mathcal E}$-primitive, then $A(H) \in \mathcal M^*$. Therefore, the menu $A(H)$ distinguishes the edge $f$. Hence, $f \in \widehat{\mathcal E}_{\mathcal M^*}$. If $H$ is not $\widehat{\mathcal E}$-primitive, then there exists an $\widehat{\mathcal E}$-primitive hexagon $H'$ such that $C(H')=C(H)$ and $A(H') \subsetneq A(H)$.\footnote{To see that why such primitive $H'$ must exist, note that if $H'$ is active but not primitive, then we get hold of another active hexagon (say $H''$) such that $C(H'')=C(H')$ and $A(H'')\subsetneq A(H')$. Since $A(\cdot)$ strictly decrease, this process must terminate at a primitive hexagon.} Since $H'$ is $\widehat{\mathcal E}$-primitive, by the construction of $\mathcal M^*$, it follows that $A(H') \in \mathcal M^*$. Since $C(H')=C(H)=\{x,y,z\}$ and $A(H') \subseteq A(H)=A(f)$, we have $\lambda(f)=\{x,y\} \subseteq A(H') \subseteq A(f)$. Therefore, by Lemma \ref{lem:menu-edge}, it follows that the menu $A(H')$ distinguishes the edge $f$. Since, $A(H') \in \mathcal M^*$, it follows that $f \in \widehat{\mathcal E}_{\mathcal M^*}$. This completes the induction on $|A(f)|$. Hence, $\widehat{\mathcal E} \subseteq \widehat{\mathcal E}_{\mathcal M^*}$.
		
		Since $\widehat{\mathcal E}_{\mathcal M^*}\subseteq
		\widehat{\mathcal E}$ and $\widehat{\mathcal E} \subseteq \widehat{\mathcal E}_{\mathcal M^*}$, we conclude that $\widehat{\mathcal E} = \widehat{\mathcal E}_{\mathcal M^*}$. Hence, the experiment $\mathcal M^*$ identifies $\widehat{\mathcal E}$. This completes the proof of the claim and hence completes the proof of the sufficiency part of the theorem.
		
	\end{proof}

	Next, we prove the necessity part of the theorem.
	
	\textbf{Necessity:} Suppose $\widehat{\mathcal{E}}$ is identifiable. This implies that there exists an experiment $\mathcal M$ that identifies $\widehat{\mathcal{E}}$, that is, $\widehat{\mathcal E}_{\mathcal M}=\widehat{\mathcal E}$. We show that $\widehat{\mathcal{E}}$ satisfies square consistency and hexagon consistency.
	
	First, we show that $\widehat{\mathcal E}$ satisfies square consistency. Let $Q=(P,P^i,P^{ij},P^j,P) \in \mathcal S$ be a square, where $P \in \mathcal P$ and $|i-j|>1$. Consider the two opposite edges $e=\{P,P^i\}$ and $f=\{P^j,P^{ji}\}$. Since $e$ and $f$ are opposite edges of the square $Q$, it follows that $\lambda(e)=\lambda(f)$ and $A(e)=A(f)$. Suppose $e \in \widehat{\mathcal E}$. Since
	$\widehat{\mathcal E}_{\mathcal M}=\widehat{\mathcal E}$, there exists a
	menu $M\in\mathcal M$ that distinguishes $e$. By Lemma
	\ref{lem:menu-edge}, $\lambda(e)\subseteq M\subseteq A(e)$. This, together with the facts that $\lambda(e)=\lambda(f)$ and $A(e)=A(f)$, it follows that $\lambda(f)\subseteq M\subseteq A(f)$. Therefore, Lemma
	\ref{lem:menu-edge} implies that the menu $M$ distinguishes the edge $f$. Hence,
	$f\in\widehat{\mathcal E}_{\mathcal M}=\widehat{\mathcal E}$. The reverse
	implication, that is, $f \in \widehat{\mathcal E}$ implies $e \in \widehat{\mathcal E}$, is identical. Since $Q \in \mathcal S$ is arbitrary, it follows that opposite edges of every square are either
	both in $\widehat{\mathcal E}$ or both outside it. Hence,
	$\widehat{\mathcal E}$ satisfies square consistency. 
	
	Next, we show that $\widehat{\mathcal E}$ satisfies hexagon consistency. Let $H \in \mathcal H$ be a
	hexagon, and consider a lower edge $e$ of the hexagon $H$. By $e^h$, we denote the opposite
	higher edge of $e$. Therefore, $\lambda(e)=\lambda(e^h)$ and $A(e) \subseteq A(e^h)$. Suppose $e\in\widehat{\mathcal E}$. Since $\widehat{\mathcal E}_{\mathcal M}=\widehat{\mathcal E}$, there exists a menu
	$M\in\mathcal M$ that distinguishes $e$. By Lemma \ref{lem:menu-edge}, $\lambda(e)\subseteq M\subseteq A(e)$. This, together with the facts that $\lambda(e)=\lambda(e^h)$ and $A(e) \subseteq A(e^h)$, it follows that $\lambda(e^h)\subseteq M\subseteq A(e^h)$. Therefore, Lemma \ref{lem:menu-edge} implies that the menu $M$ distinguishes the edge $e^h$. Therefore, $e^h\in\widehat{\mathcal E}_{\mathcal M}
	=\widehat{\mathcal E}$. Hence,  $\widehat{\mathcal E}$ satisfies Condition (i) of hexagon consistency.
	
	Next, we show that $\widehat{\mathcal E}$ satisfies Condition (ii) of hexagon consistency.  Let $H \in \mathcal H$ be
	a $\widehat{\mathcal E}$-primitive hexagon. Since $H$ is active, there exists a lower edge $f$ of
	$H$ such that $f \notin\widehat{\mathcal E}$ and $f^h\in\widehat{\mathcal E}$. Let $C(H)=\{x,y,z\}$ and $\lambda(f)=\lambda(f^h)=\{x,y\}$. Also, the fact that
	$f^h$ is the opposite higher edge of $f$ implies $A(f^h)=A(H)$ and $A(f)=A(H)\setminus\{z\}$. Since $f^h\in\widehat{\mathcal E}$ and $\mathcal M$ identifies
	$\widehat{\mathcal E}$, it follows that there exists a menu $M\in\mathcal M$ that
	distinguishes the edge $f^h$. Then, by Lemma \ref{lem:menu-edge}, $\{x,y\}\subseteq M\subseteq A(f^h)=A(H)$. Below, we prove a claim which shows that the menu $M$ has to be $A(H)$.
	
	\begin{claim}\label{M=A(H)}
		$M=A(H)$.
	\end{claim}
	
	\begin{proof}
		First, we show that $z\in M$. Suppose, to the contrary, that $z\notin M$. Then $M\subseteq A(H) \setminus\{z\}=A(f)$.
		Together with $\lambda(f)=\{x,y\}\subseteq M$, Lemma
		\ref{lem:menu-edge} implies that $M$ distinguishes the edge $f$. Hence, $f$ is a lower edge of the hexagon $H$ and 
		$f\in \widehat{\mathcal E}$, a
		contradiction to the fact that $f \notin\widehat{\mathcal E}$. Therefore, $\{x,y,z\}\subseteq M\subseteq A(H)$.
		
		Next, we prove that $M=A(H)$. Suppose, to the contrary, that
		$M\subsetneq A(H)$. Consider a hexagon $H_M \in \mathcal H$ such that
		$C(H_M)=\{x,y,z\}$ and $A(H_M)=M$. Notice that such a hexagon always exists. Let $e$ be the lower edge of $H_M$ with $\lambda(e)=\{x,y\}$, and let
		$e^h$ be its opposite higher edge. Note that since $\lambda(e^h)=\{x,y\}\subseteq \{x,y,z\}\subseteq M=A(H_M)=A(e^h)$, by Lemma \ref{lem:menu-edge}, it follows that $M$ distinguishes $e^h$. This, together with the facts that $M \in \mathcal M$ and $\mathcal M$ identifies $\widehat{\mathcal E}$ imply $e^h \in \widehat{\mathcal E}$. 
		
		We claim that $e \notin\widehat{\mathcal E}$. Suppose otherwise, that is, $e \in\widehat{\mathcal E}$. Since
		$\mathcal M$ identifies $\widehat{\mathcal E}$, there exists a menu
		$M'\in\mathcal M$ that distinguishes $e$. By Lemma
		\ref{lem:menu-edge}, $\{x,y\}\subseteq M'\subseteq A(e)=M\setminus\{z\}$.
		Since $M\subseteq A(H)$, it follows that $M'\subseteq A(H)\setminus\{z\}$. Recall that $A(f)=A(H)\setminus\{z\}$ and $\lambda(f)=\{x,y\}$.
		Thus Lemma \ref{lem:menu-edge} implies that $M'$ distinguishes the edge $f$,
		contradicting the fact that $f \notin \widehat{\mathcal E}$. Hence, it must be the case that
		$e \notin \widehat{\mathcal E}$.
		
		Therefore, $H_M$ is an active hexagon because $e \notin\widehat{\mathcal E}$ and $e^h \in\widehat{\mathcal E}$. Moreover, $C(H_M)=C(H)=\{x,y,z\}$ and $A(H_M)=M \subsetneq A(H)$, which is a contradiction to the fact that $H$ is $\widehat{\mathcal E}$-primitive. Thus $M=A(H)$. This completes the proof of the claim.
	\end{proof}
	
	We return back to showing that  $\widehat{\mathcal E}$ satisfies Condition (ii) of hexagon consistency. Let $K \in \mathcal H$ be any hexagon with $A(K)=A(H)$, and let $e$ be any higher
	edge of the hexagon $K$. This, together with Claim \ref{M=A(H)} implies $A(e)=A(K)=A(H)=M$.
	Since $\lambda(e)\subseteq A(e)=M$, Lemma \ref{lem:menu-edge} implies that
	the menu $M$ distinguishes the edge $e$. Hence $e\in\widehat{\mathcal E}_{\mathcal M}= \widehat{\mathcal E}$.
	Thus every higher edge of the hexagon $K$ belongs to
	$\widehat{\mathcal E}$. Since $K$ was any arbitrary hexagon with $A(H)=A(K)$, we conclude that that  $\widehat{\mathcal E}$ satisfies Condition (ii) of hexagon consistency. Therefore, $\widehat{\mathcal E}$ satisfies hexagon consistency. This completes the proof of the necessity part of the theorem.
	
	Finally, we prove the uniqueness statement of the theorem.\\
	\textbf{Uniqueness statement:} It remains to prove that if $\widehat{\mathcal E}$ is identifiable, then there exists a unique irreducible experiment that identifies it. Suppose $\widehat{\mathcal E}$ is identifiable. Consider the experiment $\mathcal M^*$ as defined above in the proof of the sufficiency part of the theorem, that is,
	
	$$ \mathcal M^*= \mathcal M^*_2 \cup \mathcal M^*_3 \cup \cdots \cup \mathcal M^*_n,$$
	where $\mathcal M^*_2=\{\lambda(e) : e \in \widehat{\mathcal E} \mbox{ and } A(e)=\lambda(e)\}$ and $\mathcal M^*_j = \{A(H): H \in \mathcal H \mbox{ is } \widehat{\mathcal E} \mbox{-primitive}
	\mbox{ and } |A(H)|=j\}$ for each $j \in \{3,\ldots,n\}$.
	
	We will show that $\mathcal M^*$ is the \textit{unique irreducible} experiment that identifies $\widehat{\mathcal E}$. Suppose $\mathcal M$ is \textit{any} experiment that identifies $\widehat{\mathcal E}$. First, we will show that $\mathcal M^*\subseteq \mathcal M$, which will then be used to deduce the uniqueness of $\mathcal M^*$.
	
	First, take any menu $\widebar M \in\mathcal M^*_2$. By definition of
	$\mathcal M^*_2$, there exists an edge $e \in \widehat{\mathcal E}$ such that $\widebar M=\lambda(e)=A(e)$.
	Since $\mathcal M$ identifies $\widehat{\mathcal E}$ and
	$e\in\widehat{\mathcal E}$, there exists a menu $M\in\mathcal M$ that
	distinguishes $e$. By Lemma \ref{lem:menu-edge}, $\lambda(e)\subseteq M\subseteq A(e)$.
	Since $A(e)=\lambda(e)=\widebar M$, we have $M=\lambda(e)=\widebar M$. Hence
	$\widebar M\in\mathcal M$.
	
	Next, take any menu $\widebar M\in\mathcal M^*_3\cup\cdots\cup\mathcal M^*_n$.
	Then $\widebar M=A(H)$ for some $\widehat{\mathcal E}$-primitive hexagon
	$H\in\mathcal H$. Since $H$ is active, there exists a lower edge $f$ of
	$H$ such that $f \notin\widehat{\mathcal E}$ and $f^h\in\widehat{\mathcal E}$. Since $f^h\in\widehat{\mathcal E}$ and $\mathcal M$ identifies
	$\widehat{\mathcal E}$, it follows that there exists a menu $M\in\mathcal M$ that
	distinguishes the edge $f^h$. Then, by Claim \ref{M=A(H)}, $M= \widebar M=A(H)$. Hence, $\widebar M \in \mathcal M$. This completes the proof of $\mathcal M^* \subseteq \mathcal M$.

	Since $\mathcal M^*$ itself identifies $\widehat{\mathcal E}$, $\mathcal M^* \subseteq \mathcal M$
	implies that $\mathcal M^*$ is irreducible. This is because if some menu
	$M_0\in\mathcal M^*$ were redundant, then
	$\mathcal M^*\setminus\{M_0\}$ would still identify
	$\widehat{\mathcal E}$, contradicting the fact that every experiment
	identifying $\widehat{\mathcal E}$ must contain $M_0$.
	
	Finally, let $\mathcal N^*$ be any \textit{irreducible} experiment that identifies
	$\widehat{\mathcal E}$. By the preceding argument, $\mathcal M^*\subseteq \mathcal N^*$.
	If the inclusion were strict, choose a menu
	$M_0\in\mathcal N^* \setminus\mathcal M^*$. Then $\mathcal M^*\subseteq \mathcal N^* \setminus\{M_0\}\subseteq\mathcal N^*$.
	Since both $\mathcal M^*$ and $\mathcal N^*$ identify
	$\widehat{\mathcal E}$, the intermediate experiment
	$\mathcal N^*\setminus\{M_0\}$ also identifies $\widehat{\mathcal E}$. Hence
	$M_0$ is redundant in $\mathcal N^*$, contradicting the irreducibility of
	$\mathcal N^*$. Therefore $\mathcal N^*=\mathcal M^*$. Thus $\mathcal M^*$ is the unique irreducible experiment that identifies
	$\widehat{\mathcal E}$. This completes the proof of the theorem.
\end{proof}

	\subsection{Proof of Corollary \ref{corollary}} \label{C}
\vspace{2mm}

\noindent
\begin{proof} Let $T$ be a connected type space. 
	
	\textbf{Only if part:} Suppose $T$ is identifiable. Then there exists an experiment
	$\mathcal M$ such that $T=\mathcal R_{\mathcal M}$. We claim that
	$\partial T=\widehat{\mathcal E}_{\mathcal M}$. To see this, let
	$\{P,P'\}\in\mathcal E$. Since $P$ and $P'$ are adjacent, $|P \vartriangle P'|=1$, that is, they differ in
	the ranking of exactly one pair of alternatives. Thus $P$ and $P'$ belong
	to different types in $\mathcal R_{\mathcal M}$ if and only if there exists a menu
	$M\in\mathcal M$ such that $\mbox{\textit{dom}}_P(M)\neq \mbox{\textit{dom}}_{P'}(M)$.
	Equivalently, $\{P,P'\} \in \partial T$ if and only if $\{P,P'\} \in\widehat{\mathcal E}_{\mathcal M}$. Hence,
	$\partial T=\widehat{\mathcal E}_{\mathcal M}$. Therefore, the set of edges $\partial T$ is identifiable. By Theorem \ref{Characterization_of_Identifiability}, it follows that $\partial T$ satisfies face
	consistency. This completes the proof of the only if part.
	
	\textbf{If part:} Suppose that the set of edges $\partial T$ satisfies face consistency. Theorem \ref{Characterization_of_Identifiability} implies that $\partial T$ is an identifiable set of edges. Hence, there
	exists an experiment $\mathcal M$ such that $\widehat{\mathcal E}_{\mathcal M}=\partial T$.
	This, together with Observation \ref{Obs_1} implies
	$$
	\mathcal R_{\mathcal M}=\mathcal R_{\widehat{\mathcal E}_{\mathcal M}}=\mathcal R_{\partial T}.
	$$
	It remains to show that $\mathcal R_{\partial T}=T$, which would imply that $T$ is identifiable, thereby completing the proof of the if part. Recall that $\mathcal R_{\partial T}$ is the type space whose types are connected components of $G_{\partial T}$. By definition, $\partial T$ is the set of edges of the permutohedron that connects two distinct types of $T$. Therefore, after deleting
	the edges of the permutohedron that belongs to $\partial T$, no remaining edge of the permutohedron connects two distinct types of
	$T$. Hence, every connected component of $G_{\partial T}$ is contained in
	some type of $T$, that is, for every connected component $C$ of $G_{\partial T}$, there exists a type $t \in T$ such that $C \subseteq t$. Therefore, for every type $t' \in \mathcal R_{\partial T}$, there exists a type $t \in T$ such that $t' \subseteq t$.
	
	Now, fix any pair of types $t' \in \mathcal R_{\partial T}$ and $t \in T$ such that $t' \subseteq t$. We will show that $t'=t$, which will conclude that $\mathcal R_{\partial T}=T$. To show that $t=t'$, it is enough to show that $t \subseteq t'$. If $|t|=1$, then non-emptiness of $t'$ implies that $t=t'$. Assume $|t| \geq 2$. Since $T$ is
	connected, for any two preferences $P,P'\in t$, there exists a path between $P$ and $P'$ all
	of whose vertices belong to $t$. No edge on this path belongs to
	$\partial T$, because both endpoints of each such edge are in the same type $t$.
	Therefore, the path remains in the graph $G_{\partial T}$. Thus all preferences in
	$t$ belong to the same connected component of $G_{\partial T}$. Since $t'$ is a connected component of $G_{\partial T}$ and $t' \subseteq t$, it follows that $t \subseteq t'$. Therefore, it must be the case that $t=t'$. Hence, $\mathcal R_{\partial T}=T$ and we conclude that $T$ is identifiable. This completes the proof of the if part.
	
	\textbf{Uniqueness statement:} It remains to prove the uniqueness statement. Suppose $T$ is identifiable. Let $\mathcal M^*$ be an irreducible experiment that identifies $T$, that is, $\mathcal R_{\mathcal M^*}=T$.\footnote{Such $\mathcal M^*$ always exists. To see this, consider any experiment that identifies $T$. Since $X$ is finite, this experiment contains finitely many
		menus. Successively deleting redundant menus therefore yields
		an irreducible identifying experiment $\mathcal M^*$.} This implies that an edge $\{P,P'\}$ of the permutohedron is distinguished by $\mathcal M^*$ if and only if $t(P)\neq t(P')$. Therefore, $\{P,P'\} \in \widehat{\mathcal E}_{\mathcal M^*}$ if and only if $\{P,P'\} \in \partial T$, that is, $\widehat{\mathcal E}_{\mathcal M^*} = \partial T$. Hence, $\mathcal M^*$ identifies the set of edges $\partial T$. Since $\mathcal M^*$ is an irreducible experiment, by the uniqueness part of Theorem \ref{Characterization_of_Identifiability}, it follows that $\mathcal M^*$ is the unique irreducible experiment that identifies $\partial T$. 
	
	Now, let $\mathcal N$ be any irreducible experiment that identifies $T$. Then, by same argument provided above, it follows that $\mathcal N$ identifies $\partial T$. By the uniqueness part of Theorem \ref{Characterization_of_Identifiability}, it follows that $\mathcal N=\mathcal M^*$. Therefore $\mathcal M^*$ is the unique irreducible experiment identifying
	$T$. This completes the proof of the uniqueness statement, thereby completing the proof of the corollary.

\end{proof}

	\subsection{Proof of Proposition \ref{prop:applications-shortlist-dichotomy}} \label{D}
\vspace{2mm}

\noindent	
\begin{proof}
	Let $S$ be a shortlist policy that satisfies (A1)-(A4). First, we prove the only if part.
	
	\textbf{Only if part:} Suppose $T_S$ is identifiable. By Proposition \ref{prop:applications-equivalence}, it follows that $T_S$ is connected and $\partial T_S$ satisfies face
	consistency. 
	
	Symmetry (A4) implies that for any $P, P' \in \mathcal P$ with $|\operatorname{Acc}(P)|=|\operatorname{Acc}(P')|$, it must be the case that $|S(P)|=|S(P')|$. To see this, consider any such $P,P' \in \mathcal P$ with $|\operatorname{Acc}(P)|=|\operatorname{Acc}(P')|$. Then, we can always find a permutation of $Y$ (say $\sigma$) such that $P=\sigma P'$. By symmetry (A4), it follows that $S(\sigma P')=\sigma S(P')$. This, together with the fact that $P=\sigma P'$ implies $S(P)=\sigma S(P')$. Hence, $|S(P)|=|S(P')|$. We will use this consequence repeatedly below, referring to it simply as symmetry (A4).
	
	For any $P \in \mathcal P$, by Properties (A1)-(A3), it follows that
	\begin{enumerate}[(i)]
		\item $|S(P)|=0$ if $|\operatorname{Acc}(P)|=0$, 
		\item $1 \leq |S(P)| \leq |\operatorname{Acc}(P)|$ if  $|\operatorname{Acc}(P)| \geq 1$, and
		\item $S(P)=\{r_1(P),\ldots,r_{|S(P)|}(P)\}$ if $|\operatorname{Acc}(P)| \geq 1$.
	\end{enumerate} 
	
	Next, we characterize $S$ in three claims below.
	\begin{claim}\label{C_1}
		$|S(P)| \in \{1, |\operatorname{Acc}(P)|\}$ for any $P \in \mathcal P$ with $|\operatorname{Acc}(P)| \geq 1$.
	\end{claim} 
	
	\begin{proof}
		Fix any $P \in \mathcal P$ with $|\operatorname{Acc}(P)| \geq 1$. Suppose, to the contrary, that $2 \leq |S(P)| \leq |\operatorname{Acc}(P)|-1$. Then, consider the hexagon $H=(P,P^i,P^{i{(i+1)}},P^{i{(i+1)}i},P^{{(i+1)}i},P^{(i+1)},P)$ where $i=|S(P)|-1$. Note that the changing triple of the hexagon $H$ remains above $o$ at every vertex of $H$. Therefore, acceptable set at every vertex of $H$ is the same as $\operatorname{Acc}(P)$. This implies that $|\operatorname{Acc}(P)|=|\operatorname{Acc}(P')|$ for every vertex $P'$ of $H$. By symmetry (A4), it follows that $|S(P)|=|S(P')|$ for every vertex $P'$ of $H$.
		
		Note that for the hexagon $H$, the changing triple is given by $C(H)=\{r_i(P),r_{i+1}(P),r_{i+2}(P)\}$, where $i=|S(P)|-1$. This, together with the fact that
		$|S(P)|=|S(P')|$ for every vertex $P'$ of $H$ implies that every lower edge of $H$ swaps a
		selected product with an unselected product and therefore belongs to
		$\partial T_S$. On the other hand, every higher edge of $H$ swaps two selected products and
		leaves the shortlist unchanged at both vertices of the edge, so it lies outside $\partial T_S$. Therefore, $\partial T_S$ violates hexagon consistency and hence violates face consistency. This leads to a contradiction to the fact that $\partial T_S$ satisfies face consistency. Hence,
		$|S(P)| \in \{1, |\operatorname{Acc}(P)|\} \mbox{ for every } P \in \mathcal P \mbox{ with } |\operatorname{Acc}(P)| \geq 1$. This proves the claim.
	\end{proof} 
	
	\begin{claim}\label{C_2}
		If $|S(P)|=|\operatorname{Acc}(P)|$ for some $P \in \mathcal P$ with $|\operatorname{Acc}(P)| \geq 2$, then for every $P' \in \mathcal P$ with $|\operatorname{Acc}(P')|>|\operatorname{Acc}(P)|$, $|S(P')|=|\operatorname{Acc}(P')|$.
	\end{claim}
	
	\begin{proof}
		Let $P \in \mathcal P$ with $|\operatorname{Acc}(P)| \geq 2$ such that $|S(P)|=|\operatorname{Acc}(P)|$. Suppose, to the contrary that there exists $P' \in \mathcal P$ with $|\operatorname{Acc}(P')|>|\operatorname{Acc}(P)|$ such that $|S(P')|\neq|\operatorname{Acc}(P')|$. By Claim \ref{C_1}, it follows that $|S(P')|=1$. Fix any $a \in Y$. Consider a preference $\tilde P \in \mathcal P$ such that $r_1(\tilde P)= a$ and $r_2(\tilde P)=o$. Since $|\operatorname{Acc}(\tilde P)|=1$, by claim \ref{C_1} and Property (A2), it follows that $S(\tilde P)=\{a\}$. Also, consider a preference $\widehat P \in \mathcal P$ such that $|\operatorname{Acc}(P')|=|\operatorname{Acc}(\widehat P)|$ and $r_1(\widehat P)=a$. Notice that such preferences $\tilde P$ and $\widehat P$ always exist. Since $|\operatorname{Acc}(P')|=|\operatorname{Acc}(\widehat P)|$, by symmetry (A4), it follows that $|S(P')|=|S(\widehat P)|$. This, together with the fact that $|S(P')|=1$ and $r_1(\widehat P)=a$ implies $S(\widehat P)=\{a\}$. Thus, we have $S(\widehat P)=\{a\}=S(\tilde P)$. Hence, $\widehat P$ and $\tilde P$ belong to the same type (say $\tilde t$). Recall that $T_S$ is connected. Therefore, there exists a path $\pi=(P_1=\tilde P,\ldots,P_i=\widehat P)$ in the permutohedron between $\tilde P$ and $\widehat P$ such that $\pi$ is contained in $\tilde t$. Since consecutive preferences in the path $\pi$ are adjacent preferences, $|\operatorname{Acc}(\tilde P)|=1$ and $|\operatorname{Acc}(\widehat P)|=|\operatorname{Acc}(P')|>|\operatorname{Acc}(P)|\geq 2$, it follows that there must exist $j \in \{1,\ldots,i\}$ such that $|\operatorname{Acc}(P_j)|=|\operatorname{Acc}(P)|$. By symmetry (A4), this implies that $|S(P_j)|=|S(P)|$. This, together with the facts that $|S(P)|=|\operatorname{Acc}(P)|$ and $|\operatorname{Acc}(P)| \geq 2$ implies  $|S(P_j)|\geq 2$. This establishes that $P_j \notin \tilde t$. Therefore, the path $\pi$ is not contained in $\tilde t$, which is a contradiction. This completes the proof of the claim.
	\end{proof}
	
	In the next and final claim, we show that $S$ must be either $S^{\mathrm{first}}$ or $S^{\mathrm{acc}}$, which would complete the proof of the only if part of the proposition.
	
	\begin{claim}\label{C_3} One of the following must hold:
		\begin{enumerate}[(i)]
			\item $|S(P)|=1$ for every $P \in \mathcal P$ with $|\operatorname{Acc}(P)|\geq 1$.
			\item $|S(P)|=|\operatorname{Acc}(P)|$ for every $P \in \mathcal P$ with $|\operatorname{Acc}(P)|\geq 1$
		\end{enumerate}
		
	\end{claim}
	
	\begin{proof}
		If (i) holds, then by Properties (A1) and (A2) it follows that $S=S^{\mathrm{first}}$. Let us assume that (i) does not hold. Define $$s_0=\min\{|\operatorname{Acc}(P)| : P \in \mathcal P, |S(P)|=|\operatorname{Acc}(P)| \mbox{ and }|\operatorname{Acc}(P)|\geq 2\}.$$
		
		Since (i) does not hold, using Claim \ref{C_1} we can conclude that $s_0$ exists. Note that by definition, $s_0\geq 2$. By symmetry (A4), together with Claim \ref{C_1} and Claim \ref{C_2}, it follows that $|S(P)|=1$ for all $P \in \mathcal P$ such that $1 \leq |\operatorname{Acc}(P)|<s_0$ and $|S(P)|=|\operatorname{Acc}(P)|$ for all $P \in \mathcal P$ such that $|\operatorname{Acc}(P)|\geq s_0$. We will show that $s_0=2$, which would establish that $S=S^{\mathrm{acc}}$, thereby completing the proof of the claim.	
		
		Suppose, to the contrary, that $s_0\geq 3$. We will show that this leads to a contradiction, thereby completing the proof of the claim. Consider a preference $\tilde P \in \mathcal P$ such that $|\operatorname{Acc}(\tilde P)|=s_0$.  Note that such a preference $\tilde P$ always exists because $\mathcal P$ consists of all possible preferences. Since $|\operatorname{Acc}(\tilde P)|=s_0$, it follows that $r_{s_0+1}(\tilde P)=o$. Let $r_1(\tilde P)=b$. Consider the hexagon  $H=(\tilde P,\tilde P^i,\tilde P^{i{(i+1)}},\tilde P^{i{(i+1)}i},\tilde P^{{(i+1)}i},\tilde P^{(i+1)},\tilde P)$ where $i=s_0-1$. Since $s_0\geq 3$, we have $i=s_0-1\geq 2$.
		The changing triple of $H$ therefore excludes the first position at every vertex,
		so $r_1(P')=b$ at every vertex $P'$ of $H$.
		Consider the lower edge $e=\{\tilde P, \tilde P^{(i+1)}\}$ of the hexagon $H$. Note that $|\operatorname{Acc}(\tilde P)|=s_0$ and $|\operatorname{Acc}(\tilde P^{(i+1)})|=s_0-1$. This, together with the definition of $s_0$ and the fact that $r_1(\tilde P^{(i+1)})=b$, it follows that  $S(\tilde P^{(i+1)})=\{b\}$ and $|S(\tilde P)|=|\operatorname{Acc}(\tilde P)|=s_0\geq 3$. This implies that the edge $e \in \partial T_S$. Now, consider the opposite higher edge of $e$, $e^h=\{\tilde P^{i{(i+1)}},\tilde P^{i{(i+1)}i}\}$. Notice that $|\operatorname{Acc}(\tilde P^{i{(i+1)}})|=s_0-1$ and $|\operatorname{Acc}(\tilde P^{i{(i+1)}i})|=s_0-2$. Since $s_0\geq 3$, it follows that $s_0-2\geq 1$. These, together with the definition of $s_0$ and the fact that $r_1(\tilde P^{i{(i+1)}})=b=r_1(\tilde P^{i{(i+1)}i})$ implies $S(\tilde P^{i{(i+1)}})=S(\tilde P^{i{(i+1)}i})=\{b\}$. Therefore, both preferences $\tilde P^{i{(i+1)}}$ and $\tilde P^{i{(i+1)}i}$ belong to the same type. Thus, $e^h=\{\tilde P^{i{(i+1)}},\tilde P^{i{(i+1)}i}\} \notin \partial T_S$. Hence, $\partial T_S$ violates hexagon consistency, and consequently violates face consistency, a contradiction. This completes the proof of the claim. 
	\end{proof}

	This completes the proof of the only if part of the proposition. Next, we prove the if part of the proposition.
	
	\textbf{If part and unique irreducible experiment:} Suppose either $S=S^{\mathrm{first}}$ or $S=S^{\mathrm{acc}}$. It is straightforward to verify that both $S^{\mathrm{first}}$ and $S^{\mathrm{acc}}$ satisfy (A1)-(A4). First, we show that $T_{S^{\mathrm{first}}}$ and $T_{S^{\mathrm{acc}}}$ are connected. Recall that two preferences $P, P' \in \mathcal P$ belong to the same type of $T_{S^{\mathrm{first}}}$ (respectively, $T_{S^{\mathrm{acc}}}$) if and only if $r_1(P)=r_1(P')$ (respectively, $\operatorname{Acc}(P)=\operatorname{Acc}(P')$). Therefore, for any two preferences $P,P'$ belonging to the same type, consider a shortest path $\pi=(P_1=P,\ldots,P_i=P')$ in the permutohedron between $P$ and $P'$.\footnote{As already mentioned in Footnote \ref{shortest-path}, existence of a shortest path in the permutohedron between any two preferences is well established in the literature. See Proposition 4.1 in \citet{sato2013sufficient}.} Along a shortest path, that is, a path with minimum number of edges, every pair of alternatives on which
	both $P$ and $P'$ agree retains its relative ranking. If
	$r_1(P)=r_1(P')$, the common top ranked alternative therefore remains
	top ranked throughout the path. Similarly, if
	$\operatorname{Acc}(P)=\operatorname{Acc}(P')$, every product
	remains on the same side of $o$ throughout the path. Since $\pi$ is a shortest path, it follows that $r_1(P_j)=r_1(P)=r_1(P')$ for all $j \in \{1,\ldots,i\}$ (respectively, $\operatorname{Acc}(P_j)=\operatorname{Acc}(P)=\operatorname{Acc}(P')$ for all $j \in \{1,\ldots,i\}$). Therefore, $\pi$ lies inside the same type. Hence, $T_{S^{\mathrm{first}}}$ (respectively, $T_{S^{\mathrm{acc}}}$) is connected. Next, we show that $\partial T_{S^{\mathrm{first}}}$ and $\partial T_{S^{\mathrm{acc}}}$ satisfy face consistency, which would establish (using Proposition \ref{prop:applications-equivalence}) that $T_{S^{\mathrm{first}}}$ and $T_{S^{\mathrm{acc}}}$ are identifiable.
	
	First, we consider $\partial T_{S^{\mathrm{first}}}$. Notice that $\partial T_{S^{\mathrm{first}}}=\{\{P,P^1\}: P \in \mathcal P\}$, that is, $\partial T_{S^{\mathrm{first}}}$ contains precisely every edge that swaps the first two positions. Since opposite edges of a square swap the same pair of alternatives in the same positions, it follows that $\partial T_{S^{\mathrm{first}}}$ satisfies square consistency. Next, we show that $\partial T_{S^{\mathrm{first}}}$ satisfies hexagon consistency. Consider any hexagon $H$ with $A(H)=X$.\footnote{Recall that $A(H)$ is the associated set of the hexagon $H$, consisting of the changing triple together with all alternatives ranked below them.} Every higher edge of $H$ is of the form $\{P,P^1\}$, where $P \in \mathcal P$. Therefore, all higher edges of $H$ belong to $\partial T_{S^{\mathrm{first}}}$. Also, none of the lower edges of $H$ belong to $\partial T_{S^{\mathrm{first}}}$ because none of the lower edges have the form $\{P,P^1\}$ where $P \in \mathcal P$. Hence, $H$ is an active hexagon. For every other hexagon $H'$ with $A(H')\neq X$, none of its edges belong to $\partial T_{S^{\mathrm{first}}}$ because none of the edges have the form $\{P,P^1\}$ where $P \in \mathcal P$. Therefore, $H$ is also a primitive hexagon. This implies that $\partial T_{S^{\mathrm{first}}}$ contains precisely all higher edges of every primitive hexagon $H$ with $A(H)=X$. Hence, $\partial T_{S^{\mathrm{first}}}$ satisfies hexagon consistency. Therefore, $\partial T_{S^{\mathrm{first}}}$ satisfies face consistency. By Proposition \ref{prop:applications-equivalence}, it follows that $T_{S^{\mathrm{first}}}$ is identifiable.
	
	Let $\mathcal M^*_{\mathrm{first}}$ be the unique irreducible experiment that identifies $T_{S^{\mathrm{first}}}$. As discussed in Section \ref{Ex_illustrate_results} and also discussed in the proof of Theorem \ref{Characterization_of_Identifiability}, we have $$\mathcal M^*_{\mathrm{first}}=\mathcal M^*_2 \cup \mathcal M^*_3 \cup \cdots \cup \mathcal M^*_n,$$
	where $\mathcal M^*_2=\{\lambda(e) : e \in \partial T_{S^{\mathrm{first}}} \mbox{ and } A(e)=\lambda(e)\}$ and $$\mathcal M^*_j = \{A(H): H \in \mathcal H \mbox{ is } \partial T_{S^{\mathrm{first}}} \mbox{-primitive}
	\mbox{ and } |A(H)|=j\} \mbox{ for each }j \in \{3,\ldots,n\}.$$ 
	Every edge $e\in\partial T_{S^{\mathrm{first}}}$ swaps the
	first two positions, so $A(e)=X$. Since $|X|=n=m+1\geq3$,
	we have $A(e)\neq\lambda(e)$, and hence
	$\mathcal M_2^*=\emptyset$. Since $\partial T_{S^{\mathrm{first}}}$ contains precisely all higher edges of every primitive hexagon $H$ with $A(H)=X$, it follows that $\mathcal M^*_j =\emptyset$ for all $3 \leq j<n=|X|$. Hence, $\mathcal M^*_{\mathrm{first}}=\{X\}$.
	
	Next, we consider $\partial T_{S^{\mathrm{acc}}}$. Notice that $\partial T_{S^{\mathrm{acc}}}$ contains precisely those edges that
	swap $o$ with a product, that is, $\partial T_{S^{\mathrm{acc}}}=\{e \in \mathcal E:\lambda (e)=\{o,x\}, \mbox{ where } x \in Y\}$.\footnote{Recall that $\mathcal E$ is the set of all edges of the associated permutohedron.} Since opposite edges of a square swap the same pair of alternatives in the same positions, it follows that $\partial T_{S^{\mathrm{acc}}}$ satisfies square consistency. Next, we show that $\partial T_{S^{\mathrm{acc}}}$ satisfies hexagon consistency. Consider any hexagon $H$ and any lower edge $e$ of $H$. Both edges $e$ and $e^h$ swap the same pair of alternatives. Therefore, either both $e$ and $e^h$ belong to $\partial T_{S^{\mathrm{acc}}}$ or both lie outside $\partial T_{S^{\mathrm{acc}}}$, hence, condition (i) of hexagon consistency is satisfied. Further, this also implies that there is no active hexagon. Therefore, there is no primitive hexagon and hence condition (ii) of hexagon consistency is vacuously satisfied. This establishes that $\partial T_{S^{\mathrm{acc}}}$ satisfies hexagon consistency. Thus, $\partial T_{S^{\mathrm{acc}}}$ satisfies face consistency. By Proposition \ref{prop:applications-equivalence}, it follows that $T_{S^{\mathrm{acc}}}$ is identifiable.
	
	Let $\mathcal M^*_{\mathrm{acc}}$ be the unique irreducible experiment that identifies $T_{S^{\mathrm{acc}}}$. As discussed in Section \ref{Ex_illustrate_results} and also discussed in the proof of Theorem \ref{Characterization_of_Identifiability}, we have $$\mathcal M^*_{\mathrm{acc}}=\mathcal M^*_2 \cup \mathcal M^*_3 \cup \cdots \cup \mathcal M^*_n,$$
	where $\mathcal M^*_2=\{\lambda(e) : e \in \partial T_{S^{\mathrm{acc}}} \mbox{ and } A(e)=\lambda(e)\}$ and $$\mathcal M^*_j = \{A(H): H \in \mathcal H \mbox{ is } \partial T_{S^{\mathrm{acc}}} \mbox{-primitive}
	\mbox{ and } |A(H)|=j\} \mbox{ for each }j \in \{3,\ldots,n\}.$$ 
	Since $\partial T_{S^{\mathrm{acc}}}$ contains no primitive hexagons, it follows that  $\mathcal M^*_j =\emptyset$ for all $3\leq j\leq n=|X|$. For each $x\in Y$, choose a preference (say, $\tilde P$) that ranks $o$ and
	$x$ in the last two positions. The edge $e=\{\tilde P, \tilde P^m\}$ obtained by
	swapping them belongs to $\partial T_{S^{\mathrm{acc}}}$
	and satisfies $A(e)=\lambda(e)=\{o,x\}$. Since $\partial T_{S^{\mathrm{acc}}}=\{e \in \mathcal E:\lambda (e)=\{o,x\}, \mbox{ where } x \in Y\}$, it follows that
	$\mathcal M_2^*=\{\{o,x\}:x\in Y\}$. Hence, $\mathcal M^*_{\mathrm{acc}}=\{\{o,x\}: x\in Y\}$. This completes the proof of the proposition.

\end{proof}

\end{appendices}
\bibliographystyle{ecta}
\bibliography{thesis}

\end{document}